\documentclass{ucrthesis-1chair} 

\input{commands/thesisCommands}

\usepackage{hyperref}
\hypersetup{
    colorlinks,
    allcolors = blue
}

\usepackage[style=ieee-alphabetic,backend=biber,useprefix=false,url=false]{biblatex}
\usepackage{makeidx}
\makeindex

\usepackage[all]{nowidow}

\numberwithin{equation}{section}

\begin{document}

%
%

\title{Shape and Spectrum\\ On the Heat and Volume of Self-Similar Fractals}
\author{William Edward Hoffer III}
\degreemonth{September}
\degreeyear{2025}
\degree{Doctor of Philosophy}
\chair{Dr. Michel L. Lapidus}
\othermembers{
    Dr. David Weisbart\\
    Dr. Bun Wong}
\numberofmembers{2} 
\field{Mathematics}
\campus{Riverside}

\maketitle
\copyrightpage{}
\approvalpage{}

\degreesemester{Spring}

\begin{frontmatter}

\begin{acknowledgements}
I am deeply grateful for the support, encouragement, and inspiration that has been provided to me along the road toward this doctoral dissertation. Firstly, I would like to thank the teachers and professors who have supported me along my academic journey and guided me toward success. I would like to thank my dissertation advisor, Dr. Michel L. Lapidus, for all of his support and invaluable insight, as well as for introducing me to a beautiful field of mathematics. Thank you to Dr. David Weisbart and to Dr. Yat Tin Chow for your support, conversations, and mentorship. I would also like to thank Dr. Ovidiu Costin for his guidance in my learning of mathematical physics and asymptotic analysis and preparing me for the rigors of graduate education. 

Next, I would like to thank the teachers who inspired my love of mathematics and physics. Thank you Mr. D'Alessandro for the support through many hours of physics and robotics projects. Thank you Mr. Beluan and Mrs. Wade for fostering my love of mathematics and guiding me toward the joys of proofs and problem-solving. Thank you Mr. Tarnowski for dropping objects out of windows, lighting tubes on fire, and smashing open electronics all in the name of science. The inspiration and guidance I received from all of you at different stages of my journey has led to my success today. Thank you for everything. 

In addition to the teachers and professors who have taught me, I am deeply indebted to the wonderful friends, family, and peers who have been a part of my journey. Thank you Alex for your love and support I felt in every cup of tea and coffee you made me and every time you listened intently to my woes and triumphs. Thank you to my friends and roommates, Raymond Matson, Kevin Su, and Elliott Vest for all the laughs, games, and mutual mathematics support through all the trials of graduate school. Thank you Michael McNulty for your mentorship as I followed in many of your footsteps, and thank you Adam Richardson and Matthew Overduin for all the wonderful discussions of our lovely field of fractal geometry. A huge thank you to my parents, whose love and support was essential for every step of the way, and further to my siblings Hannah and Michael, as well as to our lovely dogs, Hogan, Valen, and Tucker, who never failed to put a smile on my face. Thank you as well to all my friends back in Ohio. 

I am blessed to have so many more people to thank for supporting me. Thank you to Dr. Michael Maroun for your insightful talks and questions on my presentations. Thank you Margarita for supporting and cheering all of us graduate students on every step of the way. Thank you to all of the staff in the mathematics department for all of your support and troubleshooting. Thank you to graduate division, to the American Mathematics Society, to the UCR Mathematics Department, and to John C. Fay for the fellowships and travel support during my time as a graduate student. 

Thank you all so much. This wonderful achievement could not have been done without your inspiration, without your mathematical guidance, and without your kindness.
\end{acknowledgements}

\begin{dedication}
\null\vfil
{\large
\begin{center}
    \begin{tabular}{@{}l@{}l@{}l@{}l@{}l@{}}
        &\large To           \\
        && \large the        \\
        &&& \large infinite  \\
        && \large beauties   \\[1em]
        & \large in          \\[1em]
        \large my            \\[1em]
        & \large life:       \\[1em]
        && \large friends,   \\
        &&& \large family,   \\
        &&&& \large \&       \\
        &&& \large favorite    \\
        && \large furballs.   
    \end{tabular}
\end{center}}
\vfil\null
\end{dedication}

%
%
\begin{abstract}

In this work, we examine the relationship between geometry and spectrum of regions with fractal boundary. The relationship is well-understood for fractal harps in one dimension, but largely open for fractal drums in larger dimensions. To that end, we study fractals arising as attractors of self-similar iterated function systems with some separation conditions. On the geometric side, we analyze the tube zeta functions and their poles, called complex dimensions, which govern the asymptotics of the volume of tubular neighborhoods of such fractals. On the spectral side, we study a Dirichlet problem for the heat equation, closely related to spectrum of the Laplacian. We show that the asymptotics of the total heat content are controlled by the same set of possible complex dimensions. Our method is to establish scaling functional equations and to solve by means of truncated Mellin transforms, wherefrom the scaling ratios of the underlying dynamics can be seen to govern both the geometry and spectra of these self-similar fractal drums. 

\end{abstract}

\tableofcontents
\listoffigures

\end{frontmatter}

\chapter*{Introduction}
\addcontentsline{toc}{chapter}{Introduction}
%
%

\section*{Hook}

Questions are inspiring. Can you hear the shape of a drum? How do snowflakes melt? How can a finite country have an infinitely long coastline? What is dimension? The interplay of these questions leads to our main object of study, a rough type of shape called a fractal, and asking questions about its shape and its spectrum, which encompasses quantities like sound and heat. 

To understand this work, we must understand three key ideas: self-similarity, spectrum, and dimension. It turns out that the dynamics which govern the self-similarity of a shape can be seen to control both the shape and spectrum of a fractal shape. In particular, the self-similar nature of a fractal controls quantities called complex dimensions. Every complex number has a real and an imaginary part, and for a complex dimension the real part is an amplitude and the imaginary part is a period. Together, that makes a \textit{geometric oscillation}. In fact, the existence of geometric oscillations (characterized by having a non-zero imaginary part of a complex dimension, which then corresponds to the oscillation) in a shape is a suitable way to answer the question, what even is a fractal? 

What we will show in this work is that self-similar fractals have complex dimensions controlled simply by the scaling ratios that define its self-similarity in its constructions. These same complex dimensions govern both the geometry, in the form of volume of tubular neighborhoods, and the spectrum, in the form of the heat content of that fractal set. 

\section*{Line}

In order to capture these results, we must first introduce the players on our stage. Chief among them are the self-similar fractals; these are our main objects of study. Chapter~\ref{chap:fractals} introduces these shapes by means of their dynamics. They arise from a special type of an object called an iterated function system, namely those which contain only similitudes. When we study these shapes, they will dance in pairs: there is both the fractal itself and a nearby set to which we compare the fractal and its qualities. Together they form an instrument, a fractal harp or drum, whose shapes and sounds we aim to know. 

Next, we take a bird's eye view of our work. Both volume and heat, our later objects of study, can be viewed through one unified lens. Chapter~\ref{chap:SFEs} is all about standing atop a lighthouse and looking through binoculars that can pick out scenes from across the entire landscape that self-similarity has to offer. In this setting, the binoculars are a type of integral transform, called a Mellin transform, that is key to parsing the geometric oscillations in self-similar shapes. Compared to its more relatively well-known cousin, the Fourier transform, the Mellin transform detects oscillation in multiplicative scale as opposed to those with respect to linear shifts. 

The climb to reach the top of this lighthouse is comparatively the most challenging. There is a good amount of abstract ideas and technical features that are foundation for the stairs the lead to the stage. But the view is worth the climb, as the scaling functional equations that we introduce and study can solve problems in heat and volume alike. It is here that we prove the main results which serve to establish our thesis.

Atop this framework, we may now bring our attention back down on where we started. Here, in Chapter~\ref{chap:geometry}, we take another look at the self-similar fractals with which we began. Now, we consider a question of volume. Specifically, what happens when we look at tubular neighborhoods of these fractals. With the tools of Chapter~\ref{chap:SFEs}, we have the ability to establish explicit formulae for these quantities and, moreover, to deduce the possible complex dimensions of these fractals. Further, we see that these complex dimensions play a pivotal role in the very nature of these expansions. 

The complex dimensions control not only the geometry, but also quantities of spectra. As we shift our gaze laterally, we next arrive on the nature of heat content. In Chapter~\ref{chap:heat}, we now consider what happens when heat is modelled to flow into regions with fractal boundary. The total amount of heat, called heat content, is the principle object of study. Through our unified framework, we establish explicit formulae for this heat content. Once more, the possible complex dimensions govern the nature of these expansions, appearing as the indices of the powers of the time parameter. 

When it is at last time to metaphorically descend back down, we will do so with a new perspective. Before starting our climb, the relationship between geometry and spectrum in one dimension fit into a beautiful frame. While the view from higher dimensions does not yet fit into its own portrait, we can now see definitively a strong connection which persists. Moreover, we will see that the dynamics and similarity is the unifying thread between our two related pictures, a thread which can be followed to get yet another glimpse of this relationship at large. 

\section*{Sinker}

At the end of this work, the most important takeaway is not a definitive answer to all the questions that herein we can pose. Such is neither the nature of life nor mathematics. Instead, we leave a recipe. The ingredients are dynamics and scaling ratios, and combined together into different scaling functional equations we can produce formulae for volume and heat. By understanding these ingredients and how they mix and bake together, one can further study other quantities like sound and the spectrum directly. We learn how to develop a picture of the complex dimensions of different fractals, now self-similar and perhaps later more strange and beautiful, and to turn this knowledge into better descriptions of these shapes and their properties.

\chapter{Elements of Fractal Geometry}
\label{chap:fractals}
%
%

In this chapter, we introduce the main players on our stage. Firstly, we begin with the notion of \textit{self-similar sets}. It turns out that to define such sets, we must understand how they are obtained dynamically. Namely, self-similar sets are a special type of \textit{invariant set} for a collection of contracting maps called an \textit{iterated function system}, occurring when all the elements of this system are \textit{similitudes}. This formalism for defining such fractals is due to Hutchinson \cite{Hut81}. For more examples and information on the theory of general fractal geometry, see \cite{LapRad24_IFG,Bar88,Fal90}.

Next, we discuss two types of fractal structures. The first are \textit{fractal harps} (or strings) and the second is the main object of study in this work, namely \textit{relative fractal drums}. These fractal drums consist of a set $X\subset\RRN$, typically the fractal of interest, and another set $\Omega\subset\RRN$ relative to which we consider the fractal itself. For example, we often have that $X=\partial\Omega$ is the boundary of some region $\Omega$, and thus we consider the nature of a fractal boundary with respect to the interior of that region. Fractal harps were introduced and studied in \cite{LapvFr13_FGCD}, and relative fractal drums were introduced and studied in \cite{LRZ17_FZF}.

Next, we introduce three types of fractal dimension we focus on in this work: \textit{Minkowski dimension}, \textit{similarity dimension}, and \textit{complex dimensions}. These dimensions will play a critical role in what follows and in this chapter we give a brief introduction before these concepts are revisited in later context. Similarity dimension, and a notion of \textit{lower similarity dimension} defined in a similar way to similarity dimension, will play a key role in Chapter~\ref{chap:SFEs}. Complex dimensions will be thoroughly discussed in Chapter~\ref{chap:geometry} where we also give explicit results regarding the determination of possible complex dimensions of self-similar sets from knowledge of their associated \textit{self-similar system} alone. Complex dimensions were introduced by Lapidus and van Frankenhuijsen \cite{LapvFr13_FGCD} in one dimension for fractal harps and then extended to (relative) fractal drums in a higher dimensional setting by Lapidus, Radunovi\'c, and \v Zubrini\'c in \cite{LRZ17_FZF}. 

Lastly, we define our main type of explicit example, \textit{generalized von Koch snowflakes}. These fractals are inspired by the eponymous von Koch curve introduced and studied by von Koch \cite{Koch1904,Koch1906}. When we prove general results in Chapter~\ref{chap:geometry} regarding \textit{tube functions} and in Chapter~\ref{chap:heat} regarding \textit{heat content}, we will then apply both of these results to these fractal snowflakes.

\section{Fractals from Dynamics}
\label{sec:IFS}
%
%

Complexity arises from the iteration of simple rules. In this section, we explore this tenant and see how dynamics yield fractals. The formalism that we introduce here, that of an \textit{iterated function system (IFS)}, was introduced by Hutchinson \cite{Hut81}. Put simply, an IFS is a finite set of contraction mappings on a metric space. Hutchinson proved that an IFS yields an operator that, viewed as a contraction mapping on a complete metric space (such as the space of nonempty, compact subsets of Euclidean space equipped with the Hausdorff metric), has a unique fixed point. This fixed point, or the attractor of the IFS, is a convenient way to define many fractals \cite{Fal90,Bar88,LapRad24_IFG}. 

\subsection{Iterated Function Systems}

Dynamics are generated by two things: progression of time and the change at each time step. We will focus our attention on metric spaces $(X,d)$, consisting of a set $X$ and metric $d:X\times X\to [0,\infty)$, with dynamics generated by endomorphisms $f:X\to X$. The time steps will be discrete, generated by repeated function composition. The change is given by a finite list of functions to be iterated over and over again. 

An iterated function system, or IFS for short, will be this collection of mappings to generate the dynamics. We consider metric spaces because we will impose one requirement on our mappings: they must shrink distances between points. 
\begin{definition}[Iterated Function System]
    \label{def:IFS}
    \index{Iterated function system (IFS)}
    An \textbf{iterated function system (IFS)}  $\Phi$ on a complete metric space $(X,d)$ is a finite set of \textit{contraction mappings}, $\Phi := \set{\ph_k:X\to X}_{k=1}^m.$
\end{definition}
A mapping $\ph$ is called a \index{Contraction}\textbf{contraction} if there exists a constant $r\in [0,1)$ such that for every $x,y\in X$,
\[ d(\ph(x),\ph(y)) \leq  r\,d(x,y). \]
If $r$ is the smallest such constant, it is called the \index{Lipschitz constant!of a contraction}\textbf{Lipschitz constant} of $\ph$. Typically, we shall assume that each such Lipschitz constant is positive, or equivalently that each mapping is nontrivial. Additionally, we will generally only look at iterated function systems with cardinality at least two, i.e. $|\Phi|\geq 2$. Later, when we specialize to an IFS on Euclidean space, we will be concerned with something called the \textit{invariant set} of the IFS. In the case of an IFS with a single mapping, the invariant set will merely be a singleton: the origin. 

A special case of iterated function systems are self-similar (iterated function) systems. These systems are comparatively simpler because their mappings are similitudes that uniformly scale the elements that they act on. 
\begin{definition}[Self-Similar System]
    \label{def:SSS}
    \index{Iterated function system (IFS)!Self-similar system}
    A \textbf{self-similar system} $\Phi$ on a complete metric space $(X,d)$ is an iterated function system such that each mapping in $\Phi$ is a nontrivial (contractive) \textbf{similitude}.
\end{definition}
Equivalently, this means that for each map $\ph\in\Phi$, there exists a constant $\lambda_\ph\in (0,1)$, called the \index{Similitude!Scaling ratio}\textbf{scaling ratio} of $\ph$, such that for any $x,y\in X$, 
\[ d(\ph(x),\ph(y)) = \lambda_\ph\, d(x,y). \]
In general, a \index{Similitude}\textbf{similitude} may have any scaling ratio in $[0,\infty)$. For convenience, we rule out the existence of trivial mappings by imposing that $\lambda_\ph>0$. Contractivity occurs if and only if $\lambda_\ph<1$. Lastly, note that the scaling ratio is equivalent to the Lipschitz constant of a similitude.

\subsection{Invariant Sets, Fixed Points, and Attractors of an IFS}

To understand the dynamics of an IFS, one useful tool is to look for \textit{invariant sets}. When we look at dynamics on sets, these invariant sets, obtained as a \textit{fixed point} of a suitably defined operator, will give rise to many fractals of interest. This approach also gives a way to approximate these fractals, as they are \textit{attractors}: iterates of other elements will converge to the fixed point in question. Each of these labels, defining the same object, reflects the different properties of that object and will later be used interchangeably.

We will now specialize to the following complete metric space. Let $\Cpt(\RR^\dimension)$ be the space of non-empty, compact subsets of $\RR^\dimension$. Let $d_H$ denote the \index{Hausdorff distance}\textbf{Hausdorff distance} between sets, defined for any $X,Y\in\Cpt(\RR^\dimension)$ by
\[ d_H(X,Y) := \max\set{\sup_{x\in X}d(x,Y),\sup_{y\in Y}d(y,X)}, \]
where $d(x,U):=\inf_{u\in U}d(x,u)$ is the standard distance of a point $x$ to a set $U$ and where $d(x,y)=|x-y|$ is the usual Euclidean metric on $\RR^\dimension$. Together, $(\Cpt(X),d_H)$ is a complete metric space. 

In this framework, given an IFS $\Phi$ on $\RR^\dimension$, Hutchinson defined the corresponding operator on $\Cpt(X)$ given by
\[ H_\Phi(X) := \bigcup_{\ph\in\Phi} \ph[X], \]
where here $\ph[X]$ denote the image of $\ph$ on $X$. He proved that there exists a unique \index{Iterated function system (IFS)!Fixed point}\textbf{fixed point} $X_\Phi$ in $\Cpt(\RR^\dimension)$, a point such that $H_\Phi(X_\Phi)=X_\Phi$, using the Banach fixed point theorem and the fact that the operator $H_\Phi$ is contractive \cite{Hut81}. 

Such a fixed point $X_\Phi$ has other aliases. Firstly, this fixed point is an \index{Iterated function system (IFS)!Attractor}\textbf{attractor} in the sense that the iterated image of any other non-empty, compact set $Y$ will become arbitrarily close with respect to the Hausdorff distance to $X_\Phi$. By nature of being a fixed point, we have that 
\begin{equation}
    \label{eqn:defInvariantSet}
    X_\Phi = \bigcup_{\ph\in\Phi} \ph[X_\Phi],
\end{equation}
and thus $X_\Phi$ may be called the \index{Iterated function system (IFS)!Invariant set}\textbf{invariant set} of $\Phi$.

\subsection{Self-Similarity}

\index{Similitude}\textbf{Similitudes} in Euclidean spaces are compositions of translations, rotations, reflections, and homotheties, or equivalently, any similitude is the composition of isometries (distance preserving transformations, like translations, rotations, and reflections) and uniformly scaling mappings. Either of these definitions is equivalent to the definition of a similitude $\ph$ on a metric space $(X,d)$, namely that there exists a number $\lambda_\ph\in[0,\infty)$ (called the \index{Similitude!Scaling ratio}\textbf{scaling ratio} of $\ph$) such that for any $x,y\in X$, 
\[ d(\ph(x),\ph(y)) = \lambda_\ph\, d(x,y). \]


Self-similar fractals shall be those that arise from self-similar iterated function systems, reflecting the fact that such a fractal is realized as a union of similar copies of itself.
\begin{definition}[Self-Similar Set]
    \label{def:selfSimilarSet}
    \index{Self-similar set}
    A set $X\subset\RR^\dimension$ is called \textbf{self-similar} if there exists a self-similar system $\Phi$ such that $X$ is the invariant set of $\Phi$ in the sense of Equation~\ref{eqn:defInvariantSet}.  
\end{definition}
We note that if a set is self-similar, the self-similar system $\Phi$ for which $X$ is its invariant set is very much not unique. Even when modified in statement to the uniqueness with respect to a minimally generated or presented IFS, this is a challenging problem. The earliest study appears to date back to the work of Feng and Wang \cite{FW09}, and the problem has been studied with assumptions (be it homogeneity of the system, the dimension, and/or separation conditions) in works including but not limited to that of Yao and Li \cite{YL15,YL16,Yao15,Yao17} and that of Deng and Lau \cite{DL17,DL13}; see also the references therein for more history and relevance of the problem.

\subsection{Lattice/Nonlattice Dichotomy}

In the study of self-similar fractals with multiple scaling ratios, there is a \textit{lattice/nonlattice dichotomy} in behavior depending on whether or not the ratios are \textit{arithmetically related} or not. This dichotomy has also been called the arithmetic/non-arithmetic dichotomy, and has been discussed in the work of Lalley in \cite{Lal88,Lal89} and his paper in \cite{BKS91}, in the work of Lapidus and collaborators (such as in \cite{LapvFr13_FGCD,LRZ17_FZF,Lap93_Dundee,LP10,LPW11}), and for the generalized von Koch snowflakes by van den Berg and collaborators in their on heat content \cite{vdB00_generalGKF,vdB00_squareGKF,vdBGil98,vdBHol99}. See \cite{DGM+17} and the references therein for more information about this dichotomy in fractal geometry.

\begin{definition}[Lattice/Nonlattice Dichotomy]
    \label{def:latticeDichotomy}
    \index{Lattice/nonlattice dichotomy}\index{Scaling ratios!Lattice case}\index{Scaling ratios!Nonlattice case}\index{Scaling ratios!Arithmetically related}\index{Scaling ratios!Non-arithmetically related}
    A set of scaling ratios $\set{\lambda_k}_{k=1}^K$ is said to be \textbf{arithmetically related} if the group 
    \[ G = \prod_{k=1}^K \lambda_k^\ZZ \]
    is a discrete subgroup of the positive real line with respect to multiplication. In this case, there exists a generator $\lambda_0$ such that every $\lambda_k=\lambda_0^{m_k}$ for some positive integer $m_k$.
    This is called the \textbf{lattice case}. 
    \medskip
    
    If $G$ is not a discrete subgroup of the positive real line (in which case it is a dense subgroup), then the scaling ratios are said to be \textbf{non-arithmetically related}. This is called the \textbf{nonlattice case}.
\end{definition}

An important feature of this dichotomy is that the \textit{complex dimensions} of self-similar fractals behave very differently in the lattice and nonlattice case. As we will see in Chapter~\ref{chap:geometry} (as well as Chapter~\ref{chap:heat}), the locations of complex dimensions has important implication for the asymptotics of the quantities being studied. For more information about the structure results in the different cases, see Theorem~2.16 (for the one dimensional setting) and Theorem~3.6 (for a generalized result which may be applied to certain self-similar sets in higher dimensions) of \cite{LapvFr13_FGCD}.

\section{Fractal Harps and Drums}
\label{sec:harpsAndDrums}
%
%

Next, we take a metaphorical trip to a figurative orchestra made of fractal harps and drums. To understand fractals, we shall listen to their oscillations using the space around them. This idea leads to the study of tubular neighborhoods, osculating sets, and integration over the space of scales. 

\subsection{Fractal Harps}

On the real line, the anatomy of a bounded open set is simple: each such set is a countable union of open intervals, with the lengths of intervals decaying to zero. If we think of each of the lengths as a single string to be plucked, then a \textit{fractal harp} will be the sequence of all of these lengths, ordered from longest to shortest. Unlike musicians, though, we may repeat strings of the same length in the harp if a length has multiplicity greater than one. A \textit{generalized fractal harp} can have multiplicities which need not be positive integers and may even have an uncountable multitude of strings to pluck, be we shall still require that the strings be of finite length. 

A fractal harp is the essential information about the space around a fractal, where the fractal in question is the boundary of the open set which defines the harp. For example, if $\CantorSet$ is the middle thirds Cantor set in $[0,1]$, then its corresponding fractal harp (Figure~\ref{fig:CantorHarp}) will be all of the (lengths of the) intervals removed from $[0,1]$ in the construction of the set. It turns out that it is only the lengths of the intervals that we need, as the geometry and spectrum of the fractal are independent of the particular way in which the intervals are embedded in the real line \cite{LapvFr13_FGCD}. 
\begin{figure}[t]
    \centering 
    \includegraphics[width=10cm]{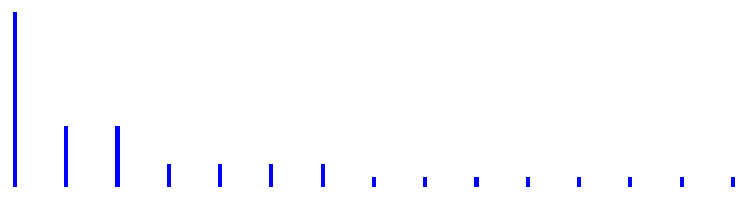}
    \caption[Fractal harp of the Cantor set]{The fractal harp of the middle-thirds Cantor set with each length represented as a vertical line segment arranged in decreasing order. Only the lengths of size one third, one ninth, one twenty-seventh, and one eighty-first are represented.}
    \label{fig:CantorHarp}
\end{figure}

For brevity, we shall jump discontinuously to the definition of a generalized fractal harp, of which ordinary fractal harps are a special case. Of note, \textit{fractal harp} and \textit{fractal string} may be used interchangeably; in this work we shall use the former terminology in analogy with fractal drums. As a preliminary, we define a \index{Measure!Local complex measure}\textbf{local complex measure} $\eta$ to be a measure such that its restriction $\eta_K$ to any compact subset $K$ of its domain is a complex measure. Essentially, this removes the boundedness constraint from the definition of a complex measure. 
\begin{definition}[Fractal Harp]
    \label{def:fractalHarp}
    \index{Fractal!Generalized fractal harp}
    A \textbf{generalized fractal harp} (equivalently, a generalized fractal string) is a local complex measure $\eta$ on $[0,\infty)$ with no mass in a neighborhood of zero. That is, there exists $\delta>0$ so that  
    $\eta([0,\delta))=0$.
\end{definition}
For a given $\ell^{-1}\in(0,\infty)$, $\eta(\set{\ell})$ is the multiplicity of the reciprocal length $\ell^{-1}$, which may be complex. It is a convention that we support the measure on reciprocal lengths rather than on the lengths themselves. The effect is that the reciprocal lengths tend to infinity rather than toward zero, and the lack of mass near zero constraint is equivalently the imposition that every length is finite. 

\subsection{Relative Fractal Drums}

Fractal drums are the instruments of dimensions greater than one. The simplest example is of a bounded open set $\Omega$ which is the drum for its boundary $\partial\Omega$, as was the case in one real dimension. However, in higher dimensions it is often more useful to consider a more general relative fractal drum: the set of interest $X$ and an open set of finite measure $\Omega$ which is close to $X$. By close, we mean in the sense that there exists a positive number $\delta$ so that $X$ is contained within a \index{Tubular neighborhood}\textbf{tubular neighborhood} $\Omega_\delta$, defined by 
\[ \Omega_\delta := \set{x\in\RR^\dimension: \exists \,y\in\Omega,\, d(x,y)<\delta}, \]
where here $d(x,y)=|x-y|$ denotes the standard Euclidean distance. (The Hausdorff distance between sets, discussed in Section~\ref{sec:IFS}, is related to the smallest value of $\delta$ for which this is possible. However, here we do not require containment in the reverse direction.) This containment is automatically satisfied when $X=\partial\Omega$. 
\begin{definition}[Relative Fractal Drum]
    \label{def:RFD}
    \index{Fractal!Relative fractal drum (RFD)}
    \index{Relative fractal drum (RFD)}
    A pair $(X,\Omega)$ is called a \textbf{relative fractal drum (RFD)} if $X\subset\RR^\dimension$ and if $\Omega$ is an open, finite Lebesgue measure set in $\RR^\dimension$ such that there exists $\delta>0$ so that $ X\subseteq \Omega_\delta$. 
\end{definition}

For example, the (boundary of the) von Koch snowflake $K_{3,\frac13}$ (which we will define formally in Section~\ref{sec:fractalExamples}) may be considered relative to the interior region of $\RR^2$ carved out by the union of three von Koch curves; in this case the RFD is exactly the pair $(\partial K_{3,\frac13},K_{3,\frac13})$. For an example which is not the interior of a set and the boundary of that set, one could consider the Sierpi\'nski triangle relative to a hexagonal polygon $\Omega$, as depicted in Figure~\ref{fig:SierpinskiRFD}. 
\begin{figure}[t]
    \centering
    \includegraphics[width=6cm]{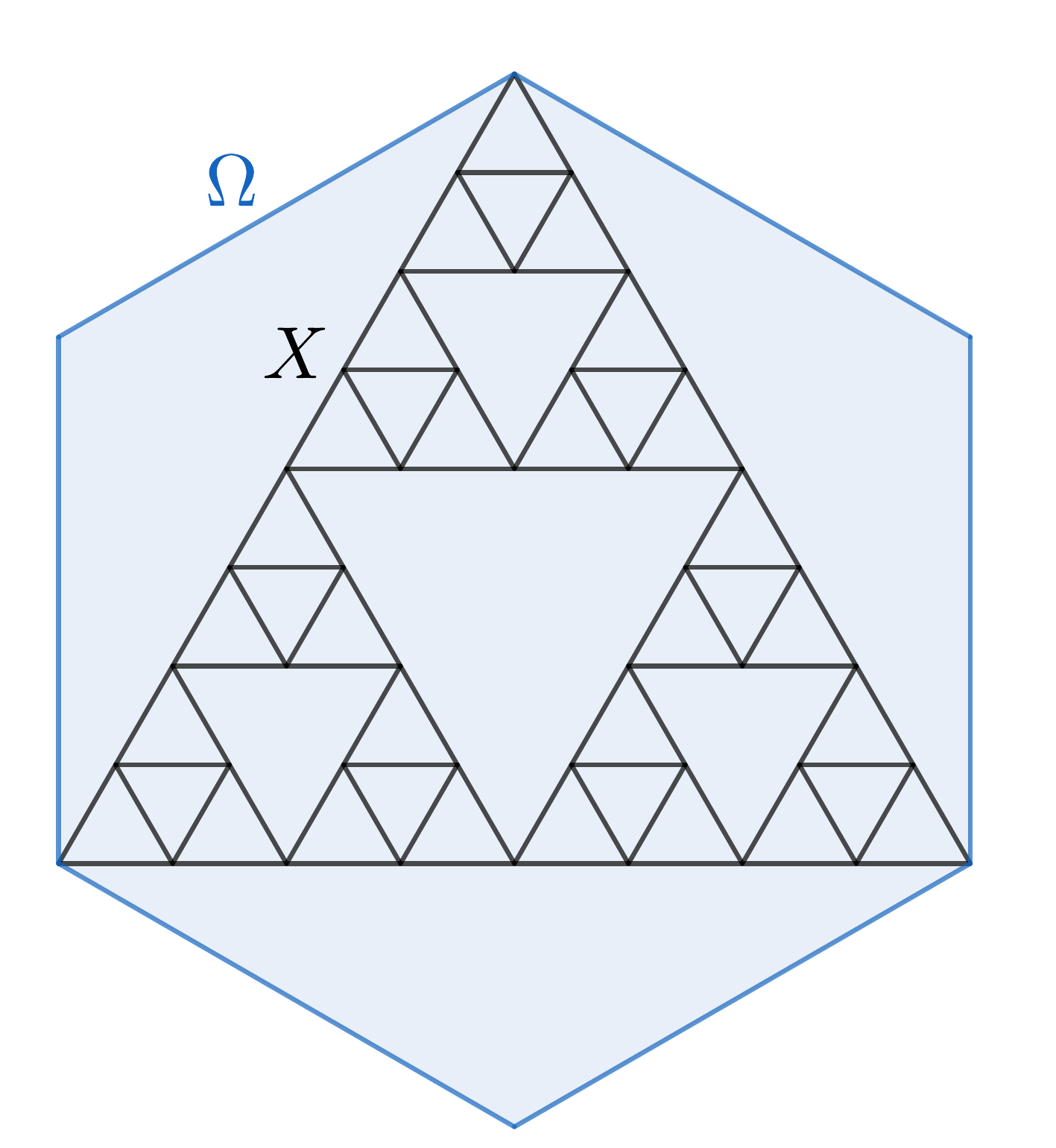}
    \caption[A relative fractal drum for the Sierpi\'nski triangle]{A relative fractal drum $(X,\Omega)$ for the Sierpi\'nski triangle $X$. Here, $X$ is depicted according to its second prefractal approximation with respect to an ordinary equilateral triangle and $\Omega$ is the interior of a regular hexagon containing $X$.}
    \label{fig:SierpinskiRFD}
\end{figure}

The terminology of a fractal drum makes the most literal sense in dimension two, when $X$ is the boundary of a membrane $\Omega$, thought of as a collection of connected drums of varying shapes. A perhaps more general term suitable for this construction in dimensions three and larger might be a \textit{fractal membrane} or, borrowing terminology common in string theory, as a \textit{fractal} $\dimension$\textit{-brane.} While we shall use the terminology of fractal drum in this work since our main class of explicit examples are indeed two dimensional shapes, this terminology is arguably more suitable in the general setting.

\section{Fractal Dimensions}
\label{sec:fractalDims}
%
%

Perhaps the most intriguing aspect of a fractal, and indeed among the best ways to define such shapes, is that their dimension does not behave like an ordinary nonnegative integer. This is unlike classical shapes such as points, lines, squares, triangles, cubes and the like. In fact, we must call into question, what even is \textit{dimension}? As it turns out, there are many different types of functions which input a set in $\RR^\dimension$ and output a number in a way consistent with the classical notions that a point has zero dimension, a line one, squares two, and so forth for cubes and onwards. 

There are a great many different types of fractal dimension that one can define (see for instance \cite{Fal90,LapRad24_IFG}). In this work, we will focus on the three most relevant in the present context: Minkowski (or box-counting) dimension, similarity dimension, and complex dimension. Notably, we will not at present consider the Hausdorff dimension, a popular notion of fractal dimension. In the context of comparing geometry and spectrum, it turns out that the Minkowski dimension is the more relevant; see, for example, the role of the Minkowski dimension (rather than the Hausdorff dimension) in a partial resolution of the modified Weyl-Berry conjecture \cite{Lap91}.

However, for the Hausdorff dimension enthusiasts, we note that it is still relevant in this work in the guise of equalling both the Minkowski and similarity dimensions. For the self-similar sets that we will consider, we impose a condition known as the open set condition (see Definition~\ref{def:OSC}). It is a theorem of Moran \cite{Mor46} that self-similar sets which satisfy the open set condition have equivalent Hausdorff, Minkowski, and similarity dimension.

\subsection{Minkowski Dimension}

Minkowski dimension is perhaps the most ``geometric" dimension among our three types, defined by means of counting the number of elements of a uniform partition of the ambient space the set intersects relative to the partition's mesh size. It plays an important role in the theory of fractal harps \cite{LapvFr13_FGCD} and (relative) fractal drums \cite{LRZ17_FZF}, and in fact it will be an example of a (real-valued) complex dimension. 

What we dub ``the" \textit{Minkowski dimension} is actually a slight abuse of words. Strictly speaking, there are two notions of Minkowski dimension, an \textit{upper} and a \textit{lower} Minkowski dimension, and \textit{the} Minkowski dimension is only defined when these two dimensions exist and equal the same value. 
In order to define upper and lower Minkowski dimensions, we will first need to define \textit{upper and lower Minkowski contents}. A geometric content is a function which assigns a numerical value to a set, much like a measure, but with less properties than a measure would. (The Minkowski contents, notably, fail to be finitely additive \cite{LapRad24_IFG}.) In the case of Minkowski content, it measures the amount by which the volume of a tubular neighborhood of a set (which is, as discussed in Chapter~\ref{chap:geometry}, is called the \textit{tube function} of that set) scales relative to the length parameter defining that neighborhood. 

We recall from Section~\ref{sec:harpsAndDrums} that the \index{Tubular neighborhood}\textbf{tubular neighborhood} of a set $X$ is, for any $\e>0$, defined by 
\[ X_\e = \set{y\in\RR^\dimension\suchthat \exists\,x\in X, d(x,y)<\e}, \]
where again $d$ represents the standard Euclidean distance between points. In what follows, let us denote by $|X|_\dimension$ the $\dimension$-dimensional Lebesgue measure of a set $X$. 
\begin{definition}[(Upper/Lower) Minkowski $t$-Dimensional Content]
    \label{def:MinkowskiContent}
    \index{Minkowski content}\index{Minkowski content!Upper Minkowski content}\index{Minkowski content!Lower Minkowski content}
    Let $X\subseteq \RR^\dimension$. The \textbf{upper and lower} $\mathit t$\textbf{-dimensional Minkowski contents} in $\RR^\dimension$, respectively $\M_{\dimension,t}^*(X)$ and $\M_*^{\dimension,t}(X)$, are defined by 
    \begin{align*}
        \M_{\dimension,t}^*(X) := \limsup_{\e\to0^+} \e^{-(\dimension-t)}|X_\e|_\dimension; &&
        \M_*^{\dimension,t}(X) := \liminf_{\e\to0^+} \e^{-(\dimension-t)}|X_\e|_\dimension.
    \end{align*}
    If the limit exists (i.e. if $\M_{\dimension,t}^*(X)=\M_*^{\dimension,t}(X)$), then its shared value is called \textbf{the} $\mathit t$\textbf{-dimensional Minkowski content} of $X$.
\end{definition}

For most values of $t\in\RR$, the upper and lower Minkowski contents are either $+\infty$ or zero. However, there is a critical value for each at which this value transitions from infinite to zero. This special value, at which the contents may be any value in $[0,\infty]$, is precisely the upper or lower Minkowski dimension, respectively. 
\begin{definition}[(Upper/Lower) Minkowski Dimension]
    \label{def:MinkowskiDim}
    \index{Minkowski dimension}\index{Minkowski dimension!Upper Minkowski dimension}\index{Minkowski dimension!Lower Minkowski dimension}
    Let $X\subseteq \RR^\dimension$. The \textbf{upper and lower} \textbf{Minkowski dimensions} of $X$, respectively $\overline{\dim}_\Mink(X)$ and $\underline{\dim}_\Mink(X)$, are defined by 
    \begin{align*}
        \overline{\dim}_\Mink(X)  := \inf\set{t\geq 0\suchthat \M^*_{\dimension,t}<\infty}; &&
        \underline{\dim}_\Mink(X) := \inf\set{t\geq 0\suchthat \M_*^{\dimension,t}<\infty}.
    \end{align*}
    If $\overline{\dim}_\Mink(X)=\underline{\dim}_\Mink(X)$, then its shared value $\dim_\Mink(X)$ is called \textbf{the Minkowski dimension} of $X$.     
\end{definition}

\subsection{Similarity Dimension}

Similarity dimension will arise from Moran's equation (see Equation~\ref{eqn:Moran} below), which is established directly for a fractal from a self-similar system having that set as its attractor/fixed point/invariant set. This equation depends only on the scaling ratios of the mappings in the system (and their respective multiplicities). Per Moran's theorem \cite{Mor46}, for self-similar fractals which satisfy the open set condition (see Definition~\ref{def:OSC} below), this equation yields an algebraic means to calculate the Minkowski dimension of a set. 

Suppose that $X\subset\RR^\dimension$ is self-similar, i.e. there exists a self-similar system $\Phi$ whose invariant set is $X$. For each $\ph\in\Phi$, let $\lambda_\ph$ denote its scaling ratio. Moran's equation is given by 
\begin{equation}
    \label{eqn:Moran}
    1 = \sum_{\ph\in\Phi} \lambda_\ph^s. 
\end{equation}
A real solution $s=D=\simdim(\Phi)$ to Equation~\ref{eqn:Moran} will be called the \index{Similarity dimension}\textbf{similarity dimension} of $\Phi$ (see Definition~\ref{def:upperSimDim}). 

We would like to say that such a solution $D$ is then the similarity dimension of the set $X$ itself. However, given a set $X$, there can be more than one IFS having $X$ as its attractor. In fact, we can always construct a new self-similar system still having $X$ as its attractor but with a different similarity dimension. Thus, we will need to impose some condition on $X$ in order for the definition to be well-defined.

First, though, we demonstrate the problem of non-uniqueness, as it gives some insight into how to avoid this problem. Given an IFS $\Phi$ (not necessarily self-similar, but without any trivial mappings) with invariant set $X$, choose a map $\ph\in\Phi$ and define the new map $\ph'=\ph\circ\ph$ which is distinct from $\ph$ (for instance, it has Lipschitz constant $r^2$, where $r$ is the Lipschitz constant of $\ph$). However, $\ph'[X]\subset \ph[X]$ since $\ph[X]\subseteq X$. If we define the new IFS $\Phi'=\Phi\cup\set{\ph'}$, we have that 
\[
    X = \bigcup_{\ph\in\Phi}\ph[X] = \bigcup_{\ph\in\Phi'}\ph[X],
\] 
which means that $X$ is also the invariant set of $\Phi'$. 

There is no reason to suspect that this artificial alteration of $\Phi$ should accurately reflect the geometry of $X$. In fact, if we now suppose that $\Phi$ is a self-similar system with similarity dimension $D=\simdim(\Phi)$ as defined by Moran's equation, we have that 
\[
    1 = \sum_{\psi\in\Phi} \lambda_\ph^D < (\lambda_{\ph'}^2)^D+\sum_{\psi\in\Phi} \lambda_\psi^D = \sum_{\psi\in\Phi'} \lambda_\psi^D,
\]
noting that the scaling ratio of $\ph'$ is $\lambda_\ph^2$. It follows that $\simdim(\Phi')\neq\simdim(\Phi)$: the similarity dimension cannot be the same. 

The open set condition (Definition~\ref{def:OSC}) is a natural separation condition that helps to rectify this problem. Introduced by Moran, this condition was used to prove his theorem which implies that this similarity dimension is exactly the Minkowski (and Hausdorff) dimension of the set $X$ \cite{Mor46}. So, for any other self-similar system having $X$ as its attractor, the similarity dimension must be the same; it is a geometric property of $X$ itself, not just $\Phi$. 

\begin{definition}[Open Set Condition]
    \label{def:OSC}
    \index{Open set condition}
    An iterated function system $\Phi=\set{\ph_k}_{k=1}^K$ on $\RR^\dimension$ satisfies the \textbf{open set condition (OSC)} if there exists a nonempty open set $U\subset\RR^\dimension$ such that:
    \begin{enumerate}
        \item $U\supset \bigcup\limits_{k=1}^K \ph_k[U]$;
        \item For each $k,j\in\set{1,...,K}$ with $k\neq j$, $\ph_k[U]\cap \ph_j[U]=\emptyset$. 
    \end{enumerate}
\end{definition}

In general, we can impose that if a bounded subset $X\subset\RR^\dimension$, $\dimension>0$, is the attractor of an IFS, then $|\Phi|\geq 2$. This is because the only nonempty, bounded set $X$ such that $X=\ph[X]$ for some contraction mapping $\ph$ is the singleton containing the origin, $\set{\mathbf{0}}$. For such a set $X$, a simple inductive argument shows that for any $n>0$ and $x_1\in X$, there is another point $x_n\in X$ so that $x=\ph^{\circ n}(x_n)$, or equivalently that $x_1\in \ph^{\circ n}[X]$. If $X$ is bounded, then there exists $C>0$ so that $X\subseteq B_C(\mathbf{0})$, a ball of radius $C$ about the origin $\mathbf{0}$. Letting $r\in[0,1)$ denote the Lipschitz constant of $\ph$, we have that $\ph[X]\subseteq B_{rC}(\mathbf{0})$ and (by another inductive argument) that for any $n>0$, $\ph^{\circ n}[X]\subseteq B_{r^nC}(\mathbf{0})$. It follows that for any $\e>0$, there exists $n>0$ so that $|x|\leq r^nC<\e$, which implies that $x=\mathbf{0}$. That $\mathbf{0}\in X$ follows from the fact that $X$ is nonempty and that $\mathbf{0}=\ph(\mathbf{0})$ for any contraction on Euclidean space. 

\begin{theorem}[Equivalence of Similarity and Minkowski Dimension]
    \label{thm:similarityDimension}
    Let $X\subset\RR^\dimension$ be a nonempty, bounded self-similar set which is not the singleton containing the origin. Suppose that there exists a self-similar system $\Phi$ having $X$ as its attractor satisfying the open set condition. Then there is a unique real number $\simdim(X)$, called the \index{Similarity dimension}\textbf{similarity dimension} of $X$, satisfying Equation~\ref{eqn:Moran} with $s=\simdim(X)$ and not depending on $\Phi$. In fact, $\simdim(X)=\dim_\Mink(X)$ and we have that $\simdim(X)\in(0,\dimension]$.
\end{theorem}

The independence of the similarity dimension of a given set $X$ on any particular self-similar system $\Phi$ with $X$ as its attractor is due to Moran and his proof of the equivalence of different fractal dimensions under the OSC \cite{Mor46}. The existence of a unique real solution to a particular incarnation of Moran's equation (Equation~\ref{eqn:Moran}), as well as the inequalities that it satisfies, is actually a rather elementary exercise in calculus which is short enough to include here. The ideas of this portion of the proof will be useful later in Chapter~\ref{chap:SFEs}.

Given a self-similar system $\Phi$, define the polynomial
\[ p(t) := \sum_{\ph\in\Phi} \lambda_\ph^t. \]
We have that $p(0)=|\Phi|$ and, by the discussion preceding the theorem, that $p(0)\geq 2$ so long as $X$ is bounded and not simply the singleton containing the origin. Further, we claim that $p(\dimension)\leq 1$. To see this, let $U$ be a nonempty open set such that $U\supset \cup_{\ph\in\Phi}\ph[U]$. By properties of the Lebesgue measure, it follows that
\[ m(U)\geq \sum_{\ph\in\Phi} \lambda_\ph^\dimension m(U) > 0. \]
Positivity follows from the openness of $U$ (and the positivity of each scaling ratio $\lambda_\ph$). Dividing the above inequality by the measure of $U$ yields the desired inequality. 

Since $p(0)\geq 2 \geq 1 \geq p(\dimension)$, by continuity and the intermediate value theorem, there exists a $t_0\in (0,\dimension]$ so that $p(t_0)=1$, proving the existence of a solution. To see uniqueness, note that since each $\lambda_\ph\in(0,1)$, we have that each $\log(\lambda_\ph)<0$. With this and the positivity of exponential maps for real inputs, it follows that 
\[ p'(t) = \sum_{\ph\in\Phi} \lambda_\ph^t\log(\lambda_\ph) < 0. \]
So, $p$ is strictly decreasing, whence the value $t_0$ such that $p(t_0)=1$ is unique (with respect to a fixed IFS $\Phi$).

\subsection{Complex Dimensions}

Lastly, and most importantly in this work, is the notion of a \textit{complex dimension}, of which a set can have multiple according to its component pieces at different scales. We will give a short overview here, but we will more formally revisit this topic in Chapter~\ref{chap:geometry} after having established the general framework we need in order to compute explicitly the possible complex dimensions of self-similar sets. 

These complex dimensions will again feature in a foundational way in Chapter~\ref{chap:heat}, as these are fundamental to describing spectral, as well as geometric, quantities of a self-similar fractal. Indeed, we expect this to be true at large as well for many more, if not all, types of fractals. Complex dimensions, notably, also give a very simple but elegant candidate for \text{the} definition of a fractal (as proposed in \cite{LapvFr13_FGCD}): a fractal is a set which has at least one non-real complex dimension. 

Firstly, let $\eta$ denote a generalized fractal harp. The \textbf{geometric zeta function} $\zeta_\eta$ is given by the Mellin transform of the measure:
\[ \zeta_\eta(s) := \int_0^\infty t^{-s}\,d\eta(t). \]
This defines $\zeta_\eta$ for $s\in\CC$ with sufficiently large real part, and it may be analytically continued in the complex plane. Let $W\subset\CC$ be such a domain on which $\zeta_\eta$ is meromorphic. Then we may define the \index{Complex dimensions!of generalized fractal harps}\textbf{complex dimension} of $\eta$ in $W$ to be the set $\Dd_\eta(W)$ of poles of $\zeta_\eta$ contained in $W$. 

For example, Figure~\ref{fig:CantorComplexDims} depicts the complex dimensions of the Cantor harp (as in Figure~\ref{fig:CantorHarp}), as well as one such choice of window. Note that in this particular example, the zeta function of the Cantor harp has a meromorphic continuation to all of $\CC$, so the choice of window is arbitrary.
\begin{figure}
    \centering
    \includegraphics[width=6cm]{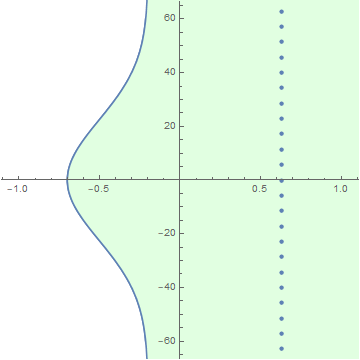}
    \caption[Complex dimensions of the Cantor harp]{The set of complex dimensions $\Dd_\eta(W)$ of the Cantor harp $\eta=\sum_{n=1}^\infty 2^{n-1}\delta_{3^{-n}}$ in the window $W$, depicted as a shaded region. Explicitly, the complex dimensions are the points of the form $\omega=\log_3(2)+ 2\pi i k/\log 3$ for $k\in\ZZ$.}
    \label{fig:CantorComplexDims}
\end{figure}

For a relative fractal drum $(X,\Omega)$, we define the set of complex dimensions of $X$ relative to $\Omega$. As before, we must define a suitable zeta function, analytically continue it, and then its poles will be the complex dimensions of the RFD. As before, let $X_\e$ denote a tubular neighborhood of $X$. Now, though, we will consider a relative tubular neighborhood, the set $X_\e\cap\Omega$, and take the Lebesgue measure of this set. The \index{Tube zeta function!Relative tube zeta function}\textbf{relative tube zeta function} $\tubezeta_{X,\Omega}$ of the RFD $(X,\Omega)$ is then given by 
\[ \tubezeta_{X,\Omega}(s;\delta) = \int_0^\delta t^{s-1}|X_t\cap\Omega|_\dimension\,dt, \]
where $\delta>0$. This zeta function may be analytically continued in the complex plane. 

Let $W\subset\CC$ be a domain on which $\tubezeta_{X,\Omega}$ is meromorphic; then the \index{Complex dimensions!of a relative fractal drum}\textbf{complex dimensions} of $X$ relative to $\Omega$ in $W$ is the set $\Dd_{X,\Omega}(W)$ of poles of $\tubezeta_{X,\Omega}$ in $W$. We note that the set of poles is independent of $\delta>0$. Also, we note that there are different constructions of related fractal zeta functions, most notably the \textit{distance zeta function}, which have equivalent sets of poles \cite{LRZ17_FZF}. For more information on complex dimensions in this work, see Chapter~\ref{chap:geometry} in which we explicitly compute the possible complex dimensions of self-similar relative fractal drums with appropriate separation conditions (such as the open set condition). For more information about complex dimensions in general, see \cite{LapvFr13_FGCD,LRZ17_FZF}.

\section{Examples: Fractal Snowflakes}
\label{sec:fractalExamples}
%
%

We close this chapter with explicit examples of self-similar fractals. We will focus mostly on a class of fractal snowflakes, called generalized von Koch fractals, as we will provide explicit forms of our results on these shapes. However, there are a great many varied examples from which one might choose to apply these results to, including but not limited to Sierpi\'nski carpets and gaskets, Menger sponges, Cantor sets and dust, etc. 

\subsection{Generalized von Koch Curves}


Firstly, we introduce the von Koch curve (depicted as the leftmost curve in Figure~\ref{fig:vKCurves}). This curve was originally introduced and studied by von Koch as a geometrically constructed example of a planar curve having nowhere defined tangents \cite{Koch1904,Koch1906}. The curve on the right of Figure~\ref{fig:vKCurves} is a modification of this construction: instead of equilateral triangles, squares are used in a way which will be made precise presently. This curve, since it follows a modification of von Koch's original construction, is called a generalized von Koch curve (GKC). 

\begin{figure}[t]
    \centering
    \subfloat{\includegraphics[width=4cm]{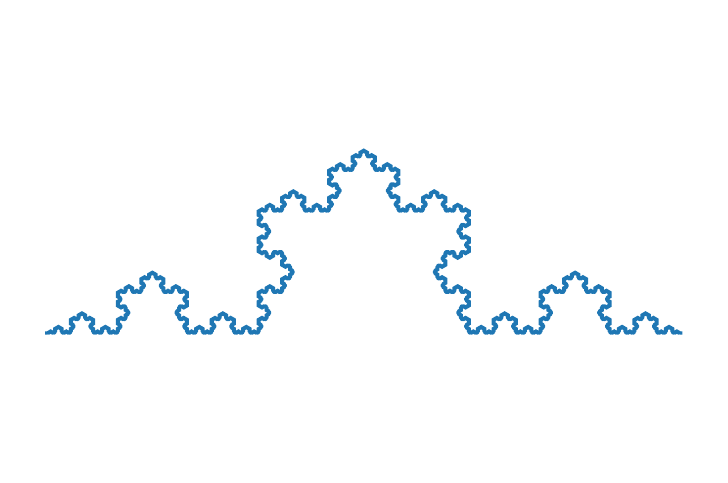}}
    \qquad
    \subfloat{\includegraphics[width=4cm]{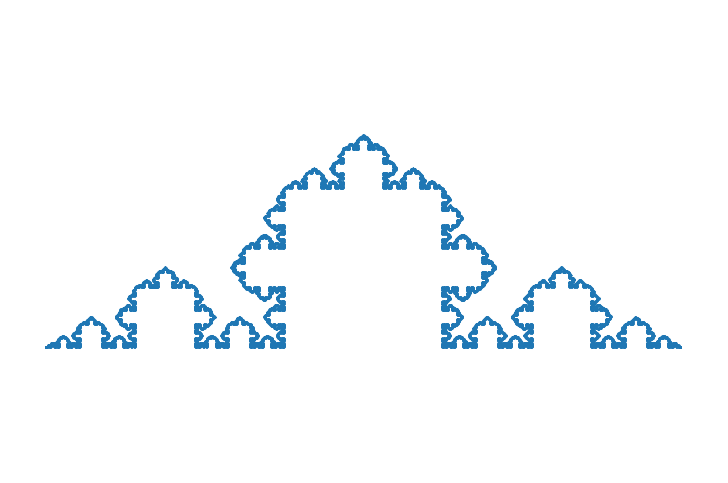}}
    \caption[The von Koch curve and a generalization thereof]{Two fractal curves: a planar curve having nowhere-defined tangent introduced by von Koch (left) \cite{Koch1904} and a generalized von Koch curve (right).}
    \label{fig:vKCurves}
\end{figure}

We will define these curves by constructing an explicit self-similar iterated function system $\Phi$ whose attractor is the curve in question. There will be two parameters in the construction: $n\geq 3$, which corresponds to the type of regular $n$-gon used in constructing the curve, and $r\in(0,1)$, a scaling ratio corresponding to the length of the gap in the middle of the curve. There is an additional conjugate scaling ratio, which we call $\ell$, given by $\ell=(1-r)/2$. In the von Koch curve, $r=\ell$, but for all of the other GKCs that we will consider they will be distinct. This is because we must impose constraints on the admissible values of $r$ in order to avoid overlapping in the curves (see Figure~\ref{fig:intersectingGKFs} for a depiction of such overlap in the context of generalized von Koch snowflakes, obtained as unions of GKCs).


In order to construct the self-similar system $\Phi=\Phi_{n,r}$, we will need some preliminaries. Firstly, we shall define some important angles related to a regular $n$-gon. Define $\theta_n:=\frac{2\pi}n$ to be the exterior angle (equivalent to the central angle) and define $\alpha_n:=\pi-\frac{2\pi}n$ to be the interior angle of a regular $n$-gon. See Figure~\ref{fig:ngonAngles} for a depiction of these angles on a hexagon.

\begin{figure}[t]
    \centering
    \includegraphics[width=0.4\linewidth]{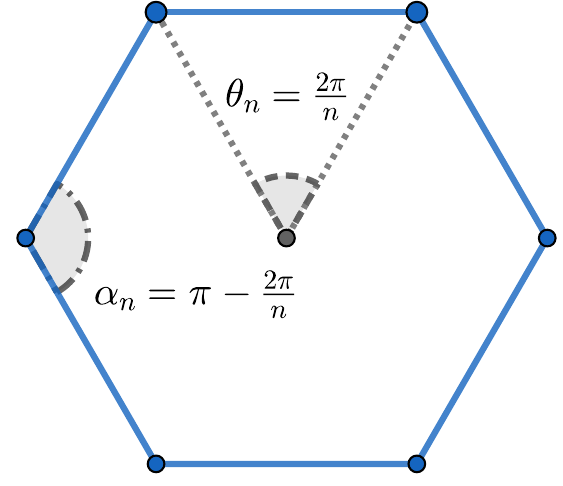}
    \caption[Angles of a regular $n$-gon]{A depiction of the central angle $\theta_n=2\pi/n$ and the interior angle $\alpha_n=\pi-2\pi/n$ of a regular $n$-gon, illustrated on a hexagon where $n=6$.}
    \label{fig:ngonAngles}
\end{figure}

Next, we define some building block transformations of $\RR^2$. This includes translations $T_{(a,b)}$, rotations $R_\theta$ by an angle $\theta$, and uniform scaling maps $S_\lambda$. 
\begin{equation}
    \label{eqn:transforms}
    \begin{aligned}
        T_{(a,b)}(x,y)  &:= (x+a,y+b),                                              &(a,b)  &\in\RR^2;  \\
        R_\theta(x,y)     &:= (x\cos\theta-y\sin\theta, x\sin\theta+y\cos\theta),     &\theta &\in\RR;    \\
        S_\lambda(x,y)  &:= (\lambda x,\lambda y),                                  &\lambda&\in\RR^+. 
    \end{aligned}
\end{equation}
With all of this geometric information, we may now explicitly write a self-similar system whose attractor will be a generalized von Koch curve. In general, a generalized von Koch curve with parameters $(n,r)$ is any set which is isomorphic to the particular attractor that our explicit self-similar system constructs. 

Let $n\geq 3$ be an integer, let $r\in(0,1)$, and define $\ell=\frac{1-r}2$. Let $T_{(a,b)},S_\lambda,$ and $R_\theta$ be transformations of $\RR^2$ as defined in Equation~\ref{eqn:transforms} and let $\theta_n$ and $\alpha_n$, respectively, be the central and interior angles of a regular $n$-gon (such as those depicted in Figure~\ref{fig:ngonAngles} for $n=6$). Define the self-similar system $\Phi_{n,r}$ by 
\begin{align}
    \label{eqn:defGKCsystem}
    \begin{split}
        \Phi_{n,r}  &:= \set{ \ph_L, \ph_R, \psi_k:\RR^2\to\RR^2, k=1,...,n-1 }, \\
        \ph_L       &:= S_\ell, \\
        \ph_R       &:= T_{(\ell+r,0)}\circ S_\ell, \\
        \psi_1      &:= T_{(\ell,0)}\circ R_{\alpha_n} \circ S_r, \\
        \psi_k      &:= T_{\psi_{k-1}(1,0)}\circ R_{\alpha_n-(k-1)\theta_n}\circ S_r,\quad k>1. 
    \end{split}
\end{align} 
The maps $\ph_L$ and $\ph_R$ correspond to the left and right pieces of a GKC and the mappings $\psi_k$, for $k=1,...,n-1$, correspond to the $n-1$ edges of a regular $n$-gon which are adjoined about the middle gap. 
\begin{definition}[${(n,r)}$-von Koch Curve]
    \label{def:vkCurve}
    \index{Generalized von Koch fractal!Generalized von Koch curve}
    Let $n\geq3$ and $r\in(0,1)$ and define $\Phi_{n,r}$ as in Equation~\ref{eqn:defGKCsystem}. The invariant set $\Cnr$ of $\Phi_{n,r}$, or any subset of $\RR^2$ which is isometric to $\Cnr$, is called an ${(n,r)}$\textbf{-von Koch curve}.    
\end{definition}

An algorithmic approach to constructing the curve is given by iterating the system $\Phi_{n,r}$ on the unit interval $[0,1]\times\set{0}$. The first step removes the middle $r\nth$ piece of the interval on the $x$-axis and adjoins the $n-1$ other sides of a regular $n$-gon with length $r$. Each successive step repeats this process on every line segment, again removing the middle $r\nth$ portion of the line and adjoining the edges of a polygon whose side length is $r$ times that of the line segment's length. The regular polygon is always added with the same orientation with respect to the line segment. 

\subsection{Generalized von Koch Snowflakes}

\begin{figure}[t]
    \centering
    \subfloat{\includegraphics[trim=60 0 60 0,clip,width=3cm]{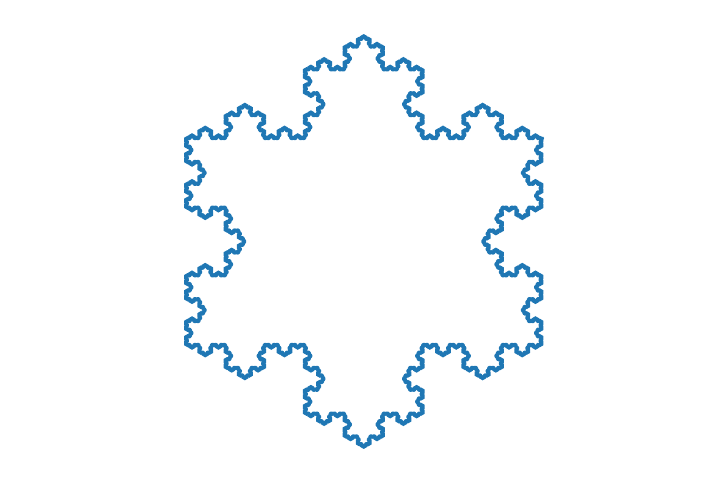}}
    \qquad 
    \subfloat{\includegraphics[trim=60 0 60 0,clip,width=3cm]{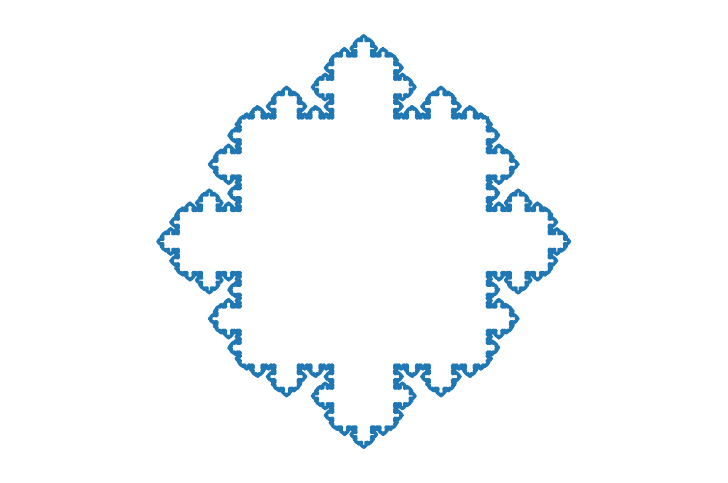}}
    \qquad 
    \subfloat{\includegraphics[trim=60 0 60 0,clip,width=3cm]{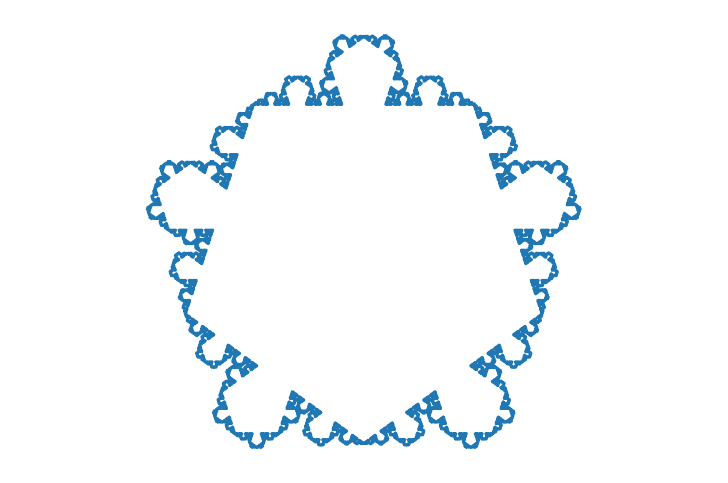}}
    \caption[The von Koch snowflake and two generalizations thereof]{The von Koch snowflake (left) and two of its generalizations, the ``squareflake'' (middle) and the ``pentaflake'' (right).}
    \label{fig:threeGKFs}
\end{figure}

Given $n$ total $(n,r)$-von Koch curves, a generalized von Koch snowflake is simply the union of these curves placed about the edges of a regular $n$-gon (of side length one), as depicted in Figure~\ref{fig:threeGKFs}.
\begin{definition}[$(n,r)$-von Koch Snowflake]
    \label{def:vkSnow}
    \index{Generalized von Koch fractal!Generalized von Koch snowflake}
    Let $n\geq 3$ be an integer, let $r\in(0,1)$, and let $p_k$, $k=1,...,n$ be the vertices of a regular $n$-gon with unit side length in $\RR^2$. For each $k<n$, let $C_{n,r}^k$ be the $(n,r)$-von Koch curve having endpoints $p_k$ and $p_{k+1}$, or in the case of $k=n$ with endpoints $p_{n}$ and $p_1$, and oriented so that the curve protrudes outwards with respect to the polygon formed by the points $p_k$. 

    An $(n,r)$\textbf{-von Koch snowflake} is any subset of $\RR^2$ which is isometric to the union of these $(n,r)$-von Koch curves $C_{n,r}^k$, $k=1,...,n$, as placed about the edges of a regular $n$-gon.
\end{definition}
Note that the set $(3,\frac13)$-von Koch snowflake is a standard von Koch snowflake and depicted leftmost in Figure~\ref{fig:threeGKFs}. Additionally, we have depicted a ``squareflake'' and a ``pentaflake,'' which are generalized snowflake fractals with fourfold and fivefold symmetry, respectively. Figure~\ref{fig:pentaflakeZoom} depicts a prefractal approximation of the pentaflake $K_{5,\frac15}$ with two extra levels of zoom onto one of the fringes, which can be seen to be pentagons.

\begin{figure}
    \centering
    \subfloat{\includegraphics[trim= 60 0 60 0, clip, width=3cm]{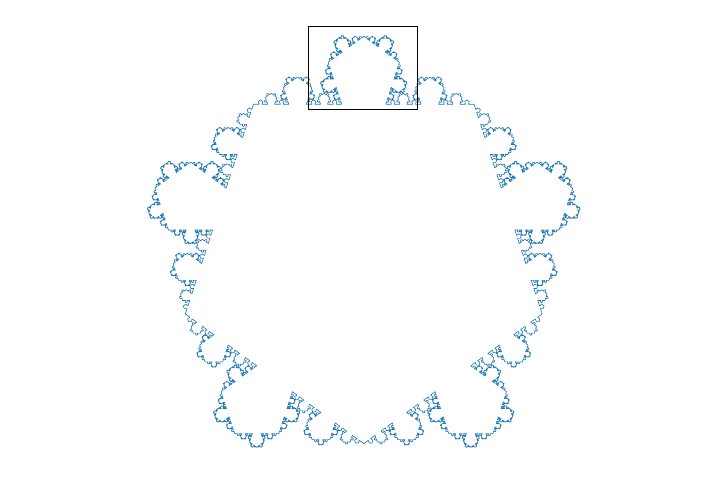}}
    \qquad
    \subfloat{\includegraphics[width=3cm]{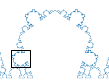}}
    \qquad
    \subfloat{\includegraphics[width=3cm]{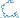}}
    \caption[Zooming in on the pentaflake fractal]{A depiction of $K_{5,\frac15}$ (at the fourth stage of the prefractal approximation) together with two zoomed-in images of the pentagonal frills.}
    \label{fig:pentaflakeZoom}
\end{figure}

Additionally, it is of note that we are identifying these von Koch snowflakes as unions of curves. One might instead define a snowflake to be the region(s) enclosed by the union of $(n,r)$-von Koch curves, in which case what we call the snowflakes here would be the boundary of this set. In this setting, it would be suitable to call $\Knr$ an $(n,r)$-von Koch \textit{snowflake boundary} or a \textit{snowflake curve} to be unambiguous. In Chapter~\ref{chap:heat}, this distinction will be necessary as we consider a relative fractal drum formed by a snowflake boundary relative to its corresponding interior. 

An important consideration is whether or not an $(n,r)$-von Koch snowflake (boundary) is a topologically simple curve or not. When $r$ is large enough, the construction of an $(n,r)$-von Koch curve can self-intersect such as for two of the fractals depicted in Figure~\ref{fig:intersectingGKFs}. The rightmost figure is obviously self-intersecting, but this is much harder to see for the middle figure. Thus, we need a precise criterion for when such curves self-intersect or not. Paquette and Keleti proved an upper bound for $r$, depending on $n$, which is a sufficient condition for the resulting $(n,r)$ curves and snowflakes to be free of self-intersection, or self-avoidant \cite{KP10}.

\begin{figure}[t]
    \centering
    \subfloat{\includegraphics[trim=80 0 80 0,clip,width=3cm]{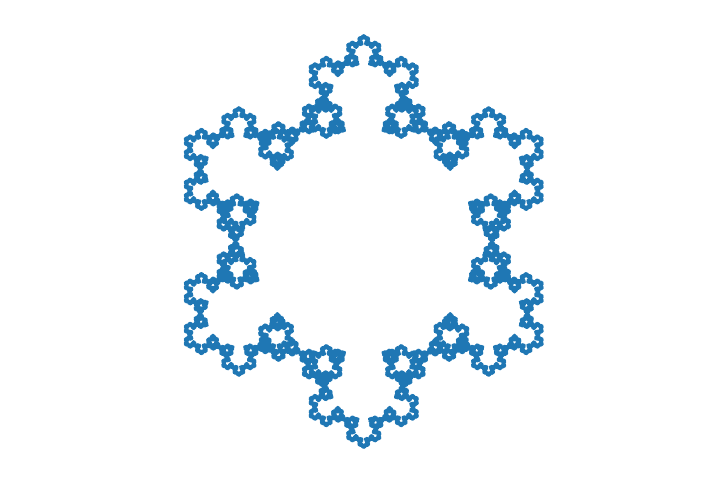}}
    \qquad
    \subfloat{\includegraphics[trim=80 0 80 0,clip,width=3cm]{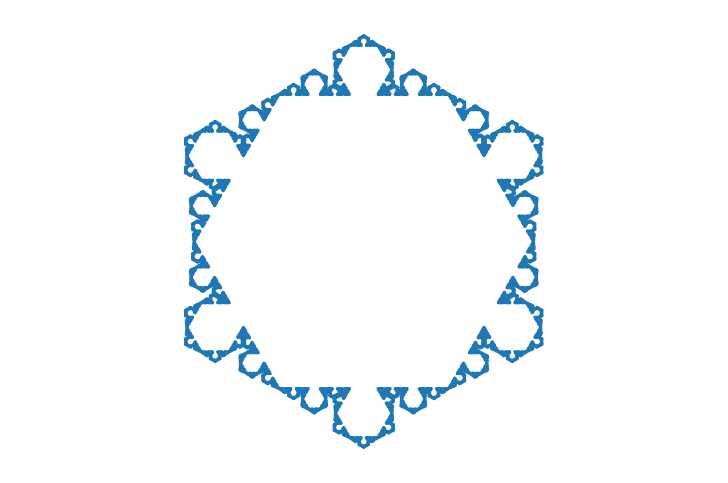}}
    \qquad
    \subfloat{\includegraphics[trim=80 0 80 0,clip,width=3cm]{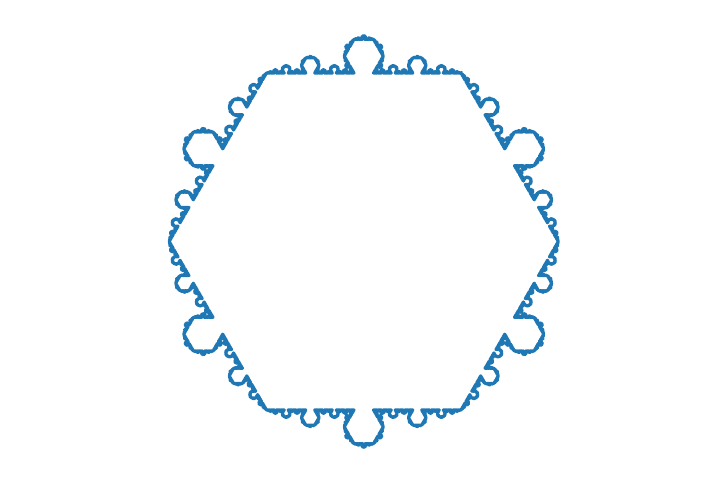}}
    \caption[Emergence of self-intersection in hexaflake fractals]{Three $(6,r)$-von Koch snowflakes, with values of $r$ equal to $0.3$, $\frac16$, and $0.1$ respectively from left to right. The fractal curve is topologically simple if $r<1-\frac{\sqrt{3}}2$, such as with the rightmost figure \cite{KP10}. For the other two fractals, depicted in the middle and on the left, the curves are self-intersecting.}
    \label{fig:intersectingGKFs}
\end{figure}

\begin{proposition}[Self-Avoidance of GKFs \cite{KP10}]
    \label{prop:selfAvoid}
    An $(n,r)$-von Koch curve has no self-intersections if the scaling ratio $r>0$ satisfies
    \begin{align*} 
        r &< \frac{\sin^2(\pi/n)}{\cos^2(\pi/n)+1}, &&\text{ if }n\text{ is even, and }\\
        r &< 1-\cos(\pi/n),                         &&\text{ if }n\text{ is odd.}
    \end{align*}
    The boundary of a corresponding $(n,r)$-von Koch snowflake is topologically simple under these conditions.
\end{proposition}
In particular, this proposition implies that for each $n\geq 3$, there is an interval of admissible scaling ratios $0<r<r_0$. The condition in Proposition~\ref{prop:selfAvoid} is sufficient but not necessary, as the fractals can interleave for certain values of $n$ and $r$. We refer the interested reader to Keleti and Paquette's paper \cite{KP10} for more information. Of note, $n=6$ is the first value for which $1/n$ exceeds the value of $r_0$ provided in the proposition.

\chapter{Scaling Functional Equations}
\label{chap:SFEs}
%
%

In this chapter, we consider a type of functional equation related to multiplicative scaling, which we call scaling functional equations. Strictly speaking, at the cost of verbosity they might be called approximate scaling functional equations in the sense that they possess an error term. The strength and novelty of our results lies in this ability to handle situations which are not cleanly separated and/or not without error, as these exact situations typically occur when the fractals of interest are realized as unions of disjoint pieces with no overlap. Furthermore, we take a multiplicative lens in our approach, as opposed to the additive lens typically taken in application of the renewal theorem of Feller, which serves to directly reveal the role of complex dimensions in solving these functional equations. 

The renewal theorem is due to Feller, who was studying queuing problems in probability theory \cite{Fel50}. This may be regarded as the earliest incarnation of solving additive functional equations, where a function is expressed in terms of shifts of its input. These methods have been applied to great success in the realm of fractal geometry, including but not limited to the work of Strichartz on self-similar measures \cite{Str1,Str2,Str3}, the work of Lapidus of the vibrations of fractal drums \cite{Lap93_Dundee}, and the work of Kigami and Lapidus on the Weyl problem \cite{KL93}. Additionally, in the case of von Koch fractals, the renewal theorem has ben used to study the heat content of the von Koch snowflake \cite{FLV95a} as well as of generalized von Koch fractals \cite{vdBGil98,vdBHol99,vdB00_generalGKF,vdB00_squareGKF}. Furthermore, the notion of a (scaling) functional equation was employed by Deniz, Ko\c cak, \"Ozdemir, and \"Ureyen in \cite{DKOU15} to provide a new proof of tube formulae for self-similar sprays (originally established in the work of Lapidus and Pearse \cite{LP10,LPW11}), a class of disjoint self-similar fractals. 

\section{Mellin Transforms}
\label{sec:Mellin}
%
%

We first define Mellin transforms and their truncated counterparts, which are our main tool in solving scaling functional equations. The Mellin transform is an integral transform associated with the Haar measure of the positive real line with respect to multiplication in the same sense that the Fourier transform is associated with the Haar measure of the real line with respect to translation (which is simply the Lebesgue measure). The space of $(\RR^+,\cdot)$, the positive real line viewed as a group with respect to multiplication, can be thought of as the space of multiplicative scales. As fractals generally have scale invariance or approximate scale invariance, the Mellin transform associated with this space is well-suited to study such shapes. 

\subsection{Mellin Transform}

Let $C^0(\RR^+)$ denote the space of continuous functions $f:\RR^+\to\RR$. The standard \index{Mellin transform}\textbf{Mellin transform} of $f$ is the integral 
\begin{equation}
    \label{eqn:defMellinTransform}
    \Mm[f](s) := \int_0^\infty x^{s-1}f(x)\,dx,
\end{equation}
defined for all $s\in\CC$ for which the integral converges. Strictly speaking, we identify the Mellin transform with its analytic continuation (which we shall restrict to be a subset of the complex plane). Here, the factor of $x^{-1}$ is due to the fact that the Haar measure associated to $(\RR^+,\cdot)$ is, in terms of the Lebesgue measure, $dx/x$. Under a scaling transformation $S_\lambda(x)=\lambda x$, note that $d(\lambda x)/\lambda x=dx/x$ is invariant.

In general, the class of functions possessing well-defined Mellin transforms is far more broad than just continuous functions. In general, integrability of the function $f$ and the existence of some polynomial growth conditions are sufficient for the transform to exist on a vertical strip in $\CC$, as we will briefly discuss later in this chapter. The functions that we shall consider in our applications (in Chapters~\ref{chap:geometry} and \ref{chap:heat}), however, will be continuous so we may impose this constraint for simplicity's sake.

\subsection{Truncated Mellin Transforms}

A truncated, or equivalently a restricted, Mellin transform is simply an integral of the same integrand as Equation~\ref{eqn:defMellinTransform}, but over an interval (i.e. connected subset) of $(0,\infty)$. 
\begin{definition}[Truncated Mellin Transform]
    \label{def:truncatedMellin}
    \index{Mellin transform!Truncated Mellin transform}
    Let $f\in C^0(\RR^+,\RR)$ and fix $\alpha,\beta \geq 0$, with $\alpha<\beta$. The truncated Mellin transform of $f$, denoted by $\Mm_\alpha^\beta[f]$, is given by 
    \[ \Mm_\alpha^\beta[f](s) := \int_\alpha^\beta t^{s-1}f(t)\,dt \]
    for all $s\in\CC$ for which the integral is convergent and analytically continued when possible.\footnote{In this work, we shall restrict ourselves to the case when these analytic continuations are a subset of the complex plane, rather than a more general \textit{Riemann surface}.}
\end{definition}
If $\alpha=0$, we write $\Mm^\beta$ for the truncated transform. If $\alpha=0$ and we take $\beta=\infty$, then $\Mm=\Mm_0^\infty$ is simply the standard Mellin transform. Note that we may equivalently define $\Mm_\alpha^\beta[f]$ to be the Mellin transform of $f$ times the characteristic function of the interval $(\alpha,\beta)$, viz. 
\[ \Mm_\alpha^\beta[f] = \Mm[f\cdot\1_{[\alpha,\beta]}]. \] 
This allows us to compare the convergence of the two directly, and it shows that the truncated transform inherits the properties of its standard counterpart, for example its linearity. 

\subsection{Convergence of Mellin Transforms}

This connection also allows us to deduce convergence properties of the truncated transform from standard results on the Mellin transform. To wit, a truncated Mellin transform converges automatically if the standard Mellin transform converges. It is known (see, for example, Chapter 6 in \cite{Gra10}) that $\Mm[f](s)$ is holomorphic in the vertical strip $\sigma_-<\Re(s)<\sigma_+$, where 
\begin{align*}
    \sigma_- :=& \inf\set{\sigma\in\RR : f(x) = O(x^{-\sigma}) \text{ as }x\to0^+},     \\
    \sigma_+ :=& \sup\set{\sigma\in\RR : f(x) = O(x^{-\sigma}) \text{ as }x\to\infty}.
\end{align*}
When $\alpha<\beta<\infty$, $f\cdot\1_{[\alpha,\beta]}\equiv 0$ as $x\to\infty$, whence $\sigma_+=\infty$. Similarly, if $0<\alpha<\beta$ then $f\cdot\1_{[\alpha,\beta]}\equiv 0$ as $x\to0^+$, whence $\sigma_-=-\infty$. So, if $0<\alpha<\beta<\infty$, the restricted transform is automatically entire, i.e. holomorphic in all of $\CC$. If $0=\alpha<\beta<\infty$ and $f=O(x^{-\sigma_0})$ as $x\to0^+$, then $\Mm^\beta[f](s)$ is holomorphic in the half-plane in $\CC$ defined by $\Re(s)>\sigma_0$. 

We collect these observations into the following lemma which will be of use later. In what follows, let 
\[ \HH_{\sigma} := \set{s\in\CC\suchthat \Re(s)>\sigma} \]
denote an \index{Half-plane}\textbf{open right half-plane} in $\CC$.
\begin{lemma}[Holomorphicity of Truncated Mellin Transforms]
    \label{lem:MellinHolo}
    \index{Mellin transform!Convergence}
    Let $f$ be integrable on $(\alpha,\beta)\subset\RR^+$ and let $\Mellin_\alpha^\beta[f]$ be its (truncated) Mellin transform.
    \begin{itemize}
        \item Let $\beta<\infty$. Suppose that $f$ is bounded away from $0$ (i.e. on any interval of the form $[\e,\beta]$ for $\e>0$) and that $f(x)=O(x^{-\sigma_0})$ as $x\to0^+$. Then $\Mellin^\beta[f]$ is convergent and holomorphic in the open right half-plane $\HH_{\sigma_0}$.
        \item Let $0<\alpha<\beta<\infty$. If $f$ is bounded on $[\alpha,\beta]$, then $\Mellin_\alpha^\beta[f]$ is entire. 
    \end{itemize}
\end{lemma}
\begin{proof}
    First, consider the case where $\beta<\infty$ and $\alpha=0$. Since $f(t)=O(t^{-\sigma_0})$ as $t\to0^+$ and is bounded on any interval of the form $[\e,\delta]$, there exists $C>0$ so that $|f(t)|\leq C\,t^{-\sigma_0}$ for all $t\in (0,\delta]$. Thus, we may estimate the Mellin transform to find that 
    \begin{align}
        \label{eqn:MellinEstimate}
        \begin{split}
            |\Mellin_0^\beta[f](s)| 
            &\leq \int_0^\delta t^{\Re(s)-1}|f(t)|\,dt \\
            &\leq C\int_0^\delta t^{\Re(s)-1-\sigma_0}\,dt \\
            &= \frac{C}{\Re(s)-\sigma_0}\,\delta^{\Re(s)-\sigma_0},
        \end{split}
    \end{align}
    provided $\Re(s)>\sigma_0$. It follows that $\Mellin_0^\beta[f](s)$ is convergent in $\HH_{\sigma_0}$. Further, this bound also implies that it is holomorphic in $\HH_{\sigma_0}$. This may be seen, for instance, by the application of Lebesgue's dominated convergence theorem to compute its derivative or alternatively by application of Morera's theorem and the Fubini-Tonelli theorem. 

    Now let $0<\alpha<\beta<\infty$. Then the function $f(t)\1_{(\alpha,\beta)}(t)$ is integrable and in fact $\Mellin_\alpha^\beta[f](s)=\Mellin^\beta[f(t)\1_{(\alpha,\beta)}(t)](s)$ for any $s$ for which it is convergent, since $f(t)=f(t)\1_{(\alpha,\beta)(t)}(t)$ for all $t\in(\alpha,\beta)$. Note that $f(t)\1_{(\alpha,\beta)}(t)$ is bounded on $[0,\beta]$ and that for any $n>0$, we have that $f(t)=O(t^n)$ as $t\to0^+$ since $\1_{(\alpha,\beta)}(t)\equiv 0$ whenever $t<\alpha$. By the first part of the proof, we have that $\Mellin^\delta[f(t)\1_{(\alpha,\beta)}(t)]$ is a convergent integral and holomorphic in any half-plane of the form $\HH_{-n}$. Consequently, $\Mellin_\alpha^\beta[f]$ converges for all $s\in\CC$ and is an entire function. 
\end{proof}

Note that the estimates in Equation~\ref{eqn:MellinEstimate} are \textit{independent of the imaginary part} of the parameter. Thus, if the real part is bounded (as is the case in a vertical strip), we may obtain uniform estimates for the Mellin transform. If $\alpha=0$, note that we must choose a vertical strip which lies in the half-plane $\HH_{\sigma_0}$, where $f(t)=O(t^{-\sigma_0})$, but the restriction is removed when $\alpha>0$ since we may choose $\sigma_0$ to be arbitrarily small. 

\begin{corollary}[Uniform Boundedness on Vertical Strips]
    \label{cor:MellinBounds}
    Let $f$ be integrable on $(\alpha,\beta)\subset\RR^+$ and let $\Mellin_\alpha^\beta[f]$ be its (truncated) Mellin transform, with $0\leq\alpha<\beta<\infty$. Suppose that either of the two hypotheses of Lemma~\ref{lem:MellinHolo} hold. In the first case, let $a>\sigma_0$ and in the second let $a\in\RR$ be arbitrary. 
    \medskip

    Then $\Mellin_\alpha^\beta[f]$ is bounded in any vertical strip of the form $\HH_a^b=\set{s\in\CC\suchthat a<\Re(s)<b}$ where either $a>\sigma_0$ in the first case (when $\alpha=0$ and $f(t)=O(t^{-\sigma_0})$ as $t\to0^+$) or where $a$ is arbitrary when $\alpha>0$. We may also choose any vertical line of the form $\Re(s)\equiv a$ or the closed strip with $a\leq \Re(s)\leq b$. 
\end{corollary}
\begin{proof}
    It is easiest to prove for closed vertical strips and deduce the other results for open strips and vertical lines as corollaries. Suppose that $\sigma_0<a\leq \Re(s)\leq b$. Then $|\Re(s)-\sigma_0|\leq a-\sigma_0>0$ is bounded from below, whence its reciprocal is bounded. The function $\delta^{t}$ is continuous on $[a,b]$, and thus bounded. It follows from Equation~\ref{eqn:MellinEstimate} that $\Mellin_\alpha^\beta[f]$ is bounded on this closed vertical strip. If $\alpha>0$, for any $a\in\RR$ choose $\sigma_0<a$ and apply the previous argument.  

    For vertical lines, take $\sigma_0<a=b$. For an open strip, note that is contained in its closure where the function is bounded by the result for closed strips, and by assumption $\sigma_0<a$. 
\end{proof}

Notably, the restricted transform may converge even if the full Mellin transform integral diverges. Indeed, this improvement to the convergence is the main reason to use this modification of the Mellin transform, as later we will see that truncation complicated the scaling property of the transform (viz. Lemma~\ref{lem:MellinScaling}) slightly. Perhaps the most important class of functions for which this occurs is that of polynomials. Without loss of generality, let $f(t)=t^k$ be a monomial. If $\Re(s)>-k$, then $\Mm^\beta[f](s)$ converges; however, $\Mm[f](s)$ is divergent for all $s\in\CC$. Thus, for a general polynomial $p(t)=\sum_{k=0}^n a_kt^k$, we have that 
\[ \Mm^\beta[p](s) = \sum_{k=0}^n \frac{a_k}{s+k}\beta^{s+k},   \]
which is valid when $\Re(s)>\max\set{-k: a_k\neq 0}$. Negative powers become admissible if the lower bound $\alpha$ of the truncation is positive. 

\subsection{Scaling Properties of Mellin Transforms}

Owing to the role of the scale invariant Haar measure $dx/x$ on $(\RR^+,\cdot)$ defining the Mellin transform, it possesses a very simple formula for scaling. Let $\lambda>0$ and define the scaling function $S_\lambda(x):=\lambda x$. Then for the standard Mellin transform,
\begin{equation}
    \label{eqn:MellinScaling}
    \Mellin[f\circ S_\lambda](s) = \lambda^{-s} \Mellin[f](s). 
\end{equation}
This property can be seen as a consequence of the change of variables formula, the scale invariance of the Haar measure $dx/x$, and the scale invariance of the domain of integration, $\RR^+$. 

For the truncated Mellin transforms, a similar property to Equation~\ref{eqn:MellinScaling} holds. However, the domain of integration, $(\alpha,\beta)$, is \textit{not} scale invariant. Thus, the cutoffs of the truncated transform change. In what follows, we assume that the function is defined (and continuous, for simplicity) on all of $\RR^+$ so that the change of domain does not present any new issues.
\begin{lemma}[Scaling Property of Truncated Mellin Transforms]
    \label{lem:MellinScaling}
    \index{Mellin transform!Scaling properties}
    Let $f\in C^0(\RR^+)$ , let $\alpha,\beta\in[0,\infty)$ with $\alpha<\beta$, let $\lambda>0$, and define $S_\lambda(x):=\lambda x$. 

    Then, provided that the transforms are convergent, we have that 
    \begin{equation}
        \label{eqn:truncMellinScaling}
        \Mellin_\alpha^\beta[f\circ S_\lambda](s) = \lambda^{-s}\Mellin_{\lambda\alpha}^{\lambda\beta}[f](s).
    \end{equation}
    If $\alpha=0$, then Equation~\ref{eqn:truncMellinScaling} becomes
    \begin{equation}
        \label{eqn:MellinScalingFunctional}
        \Mellin^\beta[f\circ S_\lambda](s)=\lambda^{-s}\Mellin^\beta[f](s) 
            + \lambda^{-s}\Mellin_{\beta}^{\lambda\beta}[f](s).
    \end{equation}
\end{lemma}
The purpose of Equation~\ref{eqn:MellinScalingFunctional} will become clear in our later usage of this property. Essentially, it establishes a functional relation akin to the standard property of Mellin transforms (Equation~\ref{eqn:MellinScaling}) up to the addition of an entire function, $\lambda^{-s}\Mellin_{\beta}^{\lambda\beta}[f](s)$. Crucially, this means that $\Mellin^\beta[f\circ S_\lambda]$ and $\lambda^{-s}\Mellin^{\beta}[f]$ have the exact same set of singularities.

\section{Scaling Operators and Scaling Zeta Functions}
\label{sec:scalingOps}
%
%

Next, we define two types of objects: scaling operators and scaling zeta functions. A \textit{scaling operator} will be used to describe a \textit{scaling functional equation} later in the chapter. Given some function, say continuous on the positive real line, these scaling functional equations will take the form $f=L[f]+R$, for some remainder term $R$. Given a particular self-similar system $\Phi$, we will define a specific scaling operator $L_\Phi$ constructed from the scaling ratios of the maps in $\Phi$ which will appear in constructing such functional equations. 

Secondly, we define a \textit{scaling zeta function} $\zeta_L$ associated to a given scaling operator $L$. These scaling zeta functions have explicit descriptions and behave exactly like the zeta functions of self-similar fractal harps (or strings) as in \cite{LapvFr13_FGCD}. Thus, we may analyze their holomorphicity, understand the locations of poles, and estimate their growth rates simply from knowing the scaling ratios used to construct the scaling operator or the self-similar system that defined said scaling operator. In the latter case, we will study the scaling zeta function $\zeta_\Phi$ associated to a self-similar system $\Phi$. 

In this setting, it will becomes clear that the properties of solutions to scaling functional equations are governed by the scaling ratios of the operator alone, whether it be a tube function (as in Chapter~\ref{chap:geometry}) or a heat content (as in Chapter~\ref{chap:heat}). In the geometric setting, the complex dimensions of a set will be contained within the poles of the associated scaling zeta function. In both settings, the poles of the scaling zeta function will index the sums of the expansions for these functions satisfying scaling functional equations, appearing as exponents of the relevant parameter. So, these scaling operators and scaling zeta functions will be key to understanding the thesis of this work. 

\subsection{Scaling Operators}
First, we define scaling operators which act on $C^0(\RR^+)$. A \textbf{pure scaling operator} $M_\lambda$, where $\lambda\in\RR^+$, shall act by precomposition with a scaling function, namely
\begin{equation}
    \label{eqn:defPureScalingOp}
    M_\lambda[f](x) := f(x/\lambda).
\end{equation}
This convention of inverting the scaling factor shall be convenient for our applications. A general scaling operator shall be a linear combination of such pure scaling operators. 
\begin{definition}[Scaling Operator]
    \label{def:scalingOp}
    \index{Scaling operator}
    For any $f\in C^0(\RR^+)$ and $\lambda>0$, define $M_\lambda$ by $M_\lambda[f](x) := f(x/\lambda)$. A \textbf{scaling operator} $L$ is any finite linear combination of such operators $M_{\lambda_k}$ with real coefficients $a_k$, where $k=1,...,K$. Explicitly, $L=\sum_{k=1}^K a_k M_{\lambda_k}$.
\end{definition}
For any $f\in C^0(\RR^+)$, we have that
\[ L[f](x) = \sum_{k=1}^K a_k f(x/\lambda_k). \]
Later, we will add the constraint that the multiplicities $a_k$ be positive and integral.

\subsection{Scaling Zeta Functions}

Given such a scaling operator $L$, we define a \textit{scaling zeta function} associated to $L$. This function shall play a key role in describing solutions to scaling functional equations, owing to its relation to the Mellin transform of such functions and the role of its singularities. 
\begin{definition}[Scaling Zeta Function]
    \label{def:scalingZetaFunction}
    \index{Scaling zeta function!of a scaling operator}
    Let $L=\sum_{k=1}^K a_k M_{\lambda_k}$ be a scaling operator. The \textbf{scaling zeta function} $\zeta_L$ associated to $L$ is the analytic continuation of the function defined by 
    \[ \zeta_L(s) := \cfrac1{1-\sum_{k=1}^K a_k \lambda_k^s} \]
    for all $s\in\CC\setminus\Dd_L$, where $\Dd_L$ the (discrete) set of singularities of $\zeta_L$.
\end{definition}

Given a self-similar system $\Phi$, we define a scaling operator $L_\Phi$ associated to $\Phi$, and thus also a scaling zeta function associated to $\Phi$. First, the \index{Scaling operator! of a self-similar system}\textbf{scaling operator} $L_\Phi$ \textbf{of the self-similar system} $\Phi$ is defined by 
\begin{equation}
    \label{eqn:defSystemScalingOp}
    L_\Phi := \sum_{\ph\in\Phi} M_{\lambda_\ph},
\end{equation}
where $\lambda_\ph$ is the scaling operator of $\ph\in\Phi$ and where $M_{\lambda_\ph}$ defined as before. 

\begin{definition}[Scaling Zeta Function of a Self-Similar System]
    \label{def:SZFofaSystem}
    \index{Scaling zeta function!of a self-similar system}
    Let $\Phi$ be a self-similar system and for each $\ph\in\Phi$, let $\lambda_\ph$ denote the scaling ratio of $\ph$. The \textbf{scaling zeta function} $\zeta_{\Phi}$ is simply the scaling zeta function of the scaling operator $L_\Phi$ associated to $\Phi$ as in Equation~\ref{eqn:defSystemScalingOp}. It is given by the analytic continuation (to all of $\CC$ except a discrete set) of the function defined by
    \begin{equation}
        \zeta_\Phi(s) = \frac1{1-\sum_{\ph\in\Phi}\lambda_\ph^s}.
    \end{equation}
\end{definition}

These scaling zeta functions of a self-similar system $\Phi$, as we will show in Chapter~\ref{chap:geometry}, determine the possible complex dimensions of the invariant set $X$ of $\Phi$. Further, they have the advantage of being comparatively simple, with explicit formulae that can be easily plotted. In Figure~\ref{fig:twoScalingZetaPlots}, we plot the scaling zeta functions associated to the self-similar systems $\Phi_{3,\frac13}$ and $\Phi_{4,\frac14}$ (see Definition~\ref{def:vkCurve}), whose invariant sets are $(n,r)$-von Koch curves used to define the von Koch snowflake $K_{3,\frac13}$ and the squareflake $K_{4,\frac14}$, respectively, which are depicted in Figure~\ref{fig:threeGKFs}. In these plots, the poles are represented by the peaks where the magnitude of the function becomes arbitrarily large. In addition to these poles controlling the possible complex dimensions as we will see in Chapter~\ref{chap:geometry}, we will also see that these poles govern the nature of explicit expansions for the heat content on these fractals in Chapter~\ref{chap:heat}. 

\begin{figure}[t]
    \centering
    \subfloat{\includegraphics[width=4cm]{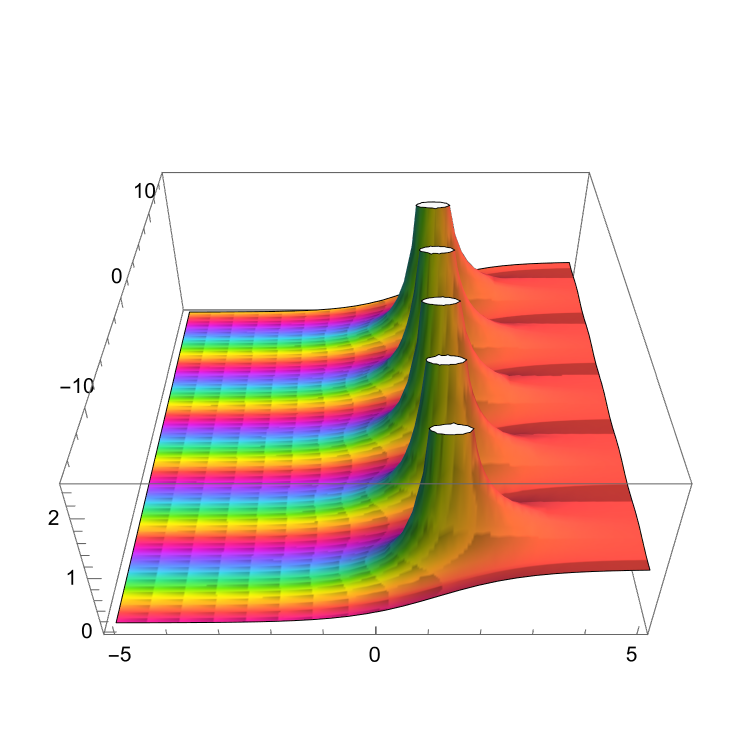}}
    \qquad
    \subfloat{\includegraphics[width=4cm]{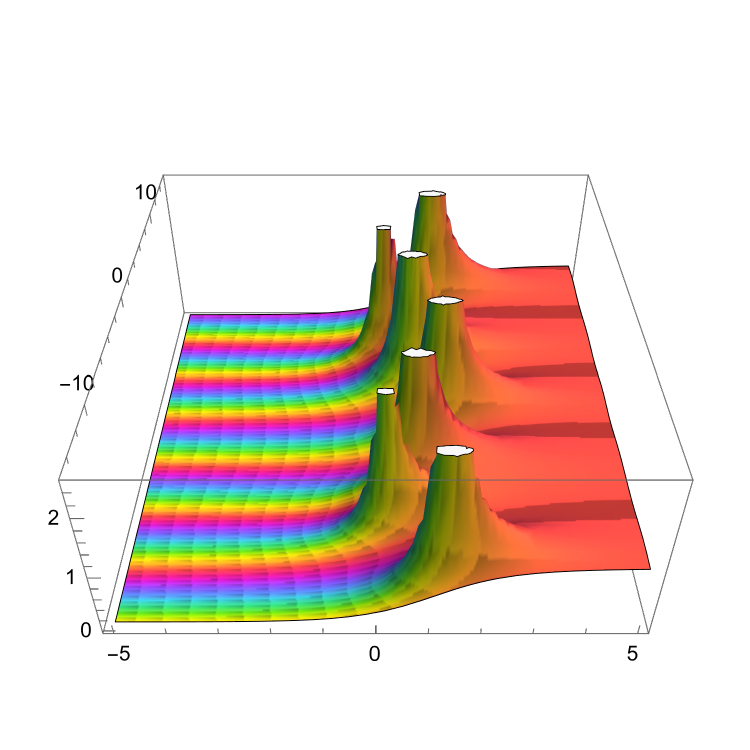}}
    \caption[Plots of two scaling zeta functions]{Plots of associated scaling zeta functions of the scaling operator $L_{3,\frac13}:=4 M_{\frac13}$ (left) and the operator $L_{4,\frac14}:= 2M_{\frac{3}8}+3M_{\frac14}$ (right.) We will show in Section~\ref{sec:appToGKFs} that the possible complex dimensions of the fractals $K_{3,\frac13}$ and $K_{4,\frac14}$, respectively, occur at the poles of these functions.}
    \label{fig:twoScalingZetaPlots}
\end{figure}

\subsection{Similarity Dimensions and Singularities}

In this section, we specialize to scaling operators which arise from a self-similar system $\Phi$. For such scaling operators $L_\Phi$, we have that the scaling ratios lie in $(0,1)$ and occur with positive, integral multiplicities. With these assumptions, we can develop certain properties of the scaling zeta function $\zeta_\Phi$, most notably regarding the locations of its poles. We define the \textit{similarity dimension} of a self-similar system, defined by means of the Moran equation (Equation~\ref{eqn:Moran2} below), which is exactly the similarity dimension of its invariant set when the open set condition is imposed. This dimension may also be called an \textit{upper} similarity dimension since it provides the abscissa of convergence of $\zeta_\Phi$, with all its other poles having smaller real parts.  

Additionally, we define a \textit{lower similarity dimension} counterpart. It is defined in a similar fashion, relying only on the scaling ratios of $\Phi$ and their multiplicities, but corresponds to a lower bound for the poles of $\zeta_\Phi$. The poles of $\zeta_\Phi$, in fact, behave just like the \textit{complex dimensions} of self-similar fractal harps (or strings). In fact, we will use known results applicable to the zeta functions of self-similar fractal harps to characterize the properties of $D_\Phi$, the set of poles of $\zeta_\Phi$. See Chapters~2 and 3, and in particular Theorem~3.6, of \cite{LapvFr13_FGCD}.

\begin{definition}[(Upper) Similarity Dimension]
    \label{def:upperSimDim}
    \index{Similarity dimension!Upper Similarity Dimension} \index{Similarity dimension}
    Let $\Phi$ be a self-similar system and let $\set{\lambda_\ph}_{\ph\in\Phi}$ denote the set of scaling ratios of the similitudes $\ph\in\Phi$. Then there is a unique real solution $D$ to Moran's equation,
    \begin{equation}
        \label{eqn:Moran2}
        1 = \sum_{\ph\in\Phi} \lambda_\ph^D. 
    \end{equation}
    This number $D=\simdim(\Phi)$ is called the \textbf{(upper) similarity dimension} of $\Phi$. 
\end{definition}
We will also denote $\simdim(\Phi)=\overline{\dim}_S(\Phi)$ when we wish to explicitly differentiation the upper similarity dimension from its lower counterpart.

Given a scaling operator $L$ with scaling ratios $\lambda_k\in(0,1)$ and positive integral multiplicities $a_k$ for $k=1,...,K$, one may construct a self-similar system $\Phi$ for which $L=L_\Phi$. In this way, we may extend this definition of similarity dimension of a self-similar system to certain scaling operators. 

\begin{proposition}[Existence and Uniqueness of Similarity Dimension]
    \label{prop:simDimExistence}
    For $k=1,...,K$, let $\lambda_k\in(0,1)$ and let $a_k\in\NN$ be positive integers. Then there is a unique real solution to the Moran equation given by
    \begin{equation}
        1 = \sum_{k=1}^K a_k \lambda_k^s. 
    \end{equation}
    If there are at least two distinct scaling ratios, or if at least one scaling ratio has multiplicity two or greater, then $D$ is positive.
\end{proposition}
\begin{proof}
    The existence and uniqueness of a real solution, $D$, is same as in the proof of Theorem~\ref{thm:similarityDimension}, using the same calculations on the polynomial $p(t)=\sum_{k=1}^K a_k \lambda_k^t$. Here, we have that $p(0)\geq2$ when either $m\geq2$ and $a_k\geq 1$ or since there is at least one $a_k\geq2$ by assumption. 
\end{proof}

\begin{proposition}[Holomorphicity of Scaling Zeta Functions]
    \label{prop:holomorphicitySZF}
    Let $\Phi$ be a self-similar system, let $\zeta_\Phi$ be its associated scaling zeta function, and let $D=\simdim(\Phi)$ be its similarity dimension. 

    Then $\zeta_\Phi$ is holomorphic in the open right half-plane $\HH_D$. It possesses a meromorphic continuation to all of $\CC$, with poles at points the set $\Dd_L=\Dd_L(\CC)$ defined by means of the complexified Moran equation, viz.
    \[ \Dd_\Phi(W) := \Big\{\omega\in W\subset\CC: 1=\sum_{\ph\in\Phi} \lambda_\ph^\omega \Big\}. \]
\end{proposition}

\begin{proof}
    To see holomorphicity, let $p(t)$ be defined as in the proof of Theorem~\ref{thm:similarityDimension} in Chapter~\ref{chap:fractals}. Then for all $s\in \CC$ such that $\sigma=\Re(s)>D$, we have that  
    \[ 
        \left|\sum_{k=1}^K a_k\lambda_k^s \right| \leq \sum_{k=1}^K a_k \lambda_k^\sigma = p(\sigma) < p(D) = 1.
    \]
    Letting $z=p(s)$, it follows that $z$ lies strictly within the unit disk, whence $\frac1{1-z}$ is well-defined and holomorphic in $z$. That $\zeta_L$ has a meromorphic extension to $\CC$ follows by the global analytic continuation of this function.
\end{proof}

Lastly, we conclude with a lower bound for the real parts of the poles of $\zeta_\Phi$. We construct such a lower bound by means of the root of an explicit Dirichlet polynomial, and for this reason we call it a \textit{lower} similarity dimension. We call $\simdim(\Phi)=\overline{\dim}_S(\Phi)$ the \textit{upper} similarity dimension of a self-similar system when we need to explicitly distinguish between it and the lower dimension to be defined presently.
\begin{definition}[Lower Similarity Dimension of a Self-Similar System]
    \label{def:lowerSimDim}
    \index{Similarity dimension! Lower similarity dimension}
    Let $\Phi$ be a self-similar system and let $\set{r_k}_{k=1}^M$ denote the \textit{unique} scaling ratios of the mappings in $\Phi$ with corresponding multiplicities $m_k$. Suppose that the scaling ratios are ordered by size, i.e. $r_1\geq r_2\geq ...\geq r_M$. The unique real solution $D_\ell$ to the equation 
    \begin{equation}
        \label{eqn:defLowerDim}
        \frac1{m_M}(r_M^{-1})^{D_\ell}+\sum_{k=1}^{M-1}\frac{m_k}{m_M}\left(\frac{r_k}{r_M}\right)^{D_\ell} = 1,
    \end{equation}
    is called the \textbf{lower similarity dimension} of $\Phi$, denoted by $\lowersimdim(\Phi)$.
\end{definition}
The following argument (from the proof of Theorem~3.6 in \cite{LapvFr13_FGCD}) justifies that this is well-defined. Let $\set{r_k}_{k=1}^M$ denote the set of unique scaling ratios of $\Phi$, arranged in decreasing order $r_1\geq ...\geq r_M$, and let $m_k$ denote the multiplicity of $r_k$. Define the function 
\begin{equation}
    \label{eqn:defLowerPoly}
    p(t) = \frac1{m_M}(r_M^{-1})^t+\sum_{k=1}^{M-1}\frac{m_k}{m_M}\left(\frac{r_k}{r_M}\right)^t.
\end{equation}
Since $p(t)$ is a linear combination, with positive coefficients, of exponential functions with bases larger than one, we can readily deduce its properties. Namely, its range is $(0,\infty)$ and it is strictly increasing. Thus, there is a unique value $D_\ell$ for which $p(D_\ell)=1$, i.e. a solution to Equation~\ref{eqn:defLowerDim}. 

The importance of this lower dimension is that it specifies a lower bound for a vertical strip in $\CC$ for which the zeta function $\zeta_\Phi$ has poles. In fact, by Theorem~3.6 of \cite{LapvFr13_FGCD}, $D_\ell\leq \inf\set{\Re(\omega)\suchthat \omega\in\Dd_\Phi}$, with equality in the nonlattice case. It will play a role in establishing some growth estimates known as languidity for the scaling zeta function $\zeta_\Phi$ of the self-similar system $\Phi$.

\section{Languid Growth}
\label{sec:languidity}
%
%

The purpose of this section is to state and prove certain growth estimates which will be necessary for solving scaling functional equations. The notion of languid growth was introduced in \cite{LapvFr13_FGCD} and then refined in \cite{LRZ17_FZF} in the context of proving explicit tube formulae for generalized fractal harps or relative fractal drums, respectively. Strictly speaking, these definitions apply to the geometric object (the harp or the drum) but involve conditions on their associated zeta functions. As in \cite{LRZ17_FZF}, we extend these definitions to apply to a (zeta) function itself, and the corresponding geometric object, including now a self-similar system, may be defined to be languid when its associated zeta function is languid.

There are two types of languid growth, standard (or weak) languid growth and strong languid growth. We will provide the definitions of both, as the scaling zeta functions of a self-similar system will in fact be strongly languid. The formulae which we will obtain for tube and heat zeta functions by means of solving a general scaling functional equation, however, will be only languid, owing to the nature of the contribution from a remainder term in the (approximate) scaling functional equation. Languid growth is the content of Definition~5.1.3 in \cite{LRZ17_FZF} and Definition~5.2 in \cite{LapvFr13_FGCD} and strongly languid growth is the content of Definition~5.1.4 in \cite{LRZ17_FZF} and Definition~5.3 in \cite{LapvFr13_FGCD}. 

As the names would suggest, strong languid growth implies (weak) languid growth. However, there is an important caveat to this implication. A strongly languid function is only languid with respect to some \textit{screen} (to be defined), not any specified screen. In order to compare the languid growth of two different functions, we introduce a notion of \textit{joint languidity} where the \textit{same} screen is used for both functions.

\subsection{Languidity Hypotheses}
%
%

Let $f$ be a function defined on a subset of the complex plane. We suppose that there is an open right half-plane $\HH_D$ on which $f$ is holomorphic. Next, let $S:\RR\to\RR$ be a bounded, Lipschitz continuous function.
We say that a \index{Screen}\textbf{screen} $\Ss$ is a set of the form 
\[ \Ss = \set{S(\tau)+i\tau\suchthat \tau\in\RR}\subset\CC. \]
By a slight abuse of notation, we say that the function $S$ itself is a screen and identify its graph with the set $\Ss$, which is a rotated embedding of its graph in $\RR^2$. 

We say that a set $W\subseteq\CC$ is a \index{Screen!Associated window}\textbf{window associated to a screen} when 
\[ W=\set{s\in\CC\suchthat \Re(s)\geq S(\Im(s))} \] 
for some screen $S$. Given a holomorphic function $f$, we will assume that $f$ has an analytic continuation to a neighborhood of the window $W$. Note that this implies $\HH_D\subset W$. The function $f$ may have a discrete set of exceptional points on which it is singular, but will be holomorphic on the complement of this set. Given a screen $S$, we say that a function $f$ is \index{Screen!Extension to a screen}\textbf{extended with respect to the screen} $S$ if it admits analytic continuation in a neighborhood of the window associated to $S$. (Note that this implies that $f$ is holomorphic in a neighborhood of $S$, as $S\subseteq\partial W$.) 

Additionally, let $\set{\tau_n}_{n\in\ZZ}$ be a doubly infinite sequence with the following properties:
\begin{align}
    \label{eqn:admissibleHeights}
    \lim_{n\to\infty} \tau_n = \infty, \quad \lim_{n\to-\infty}\tau_n =-\infty, \text{ and } 
        \tau_{-n}<0<\tau_n \text{ for all }n\geq1.
\end{align}
For brevity, we will say that $\set{\tau_n}_{n\in\ZZ}$ is a \index{Languidity!Sequence of admissible heights}\textbf{sequence of admissible heights}.

Languidity consists of two hypotheses, the first of which concerns power-law/polynomial growth along a sequence of horizontal lines. The horizontal lines each pass from a point on a screen $S$ to some point in a half-plane in which the function is holomorphic.
\begin{definition}[Languidity Hypothesis \textbf{L1}]
    \label{def:languidL1}
    \index{Languidity!Hypothesis \textbf{L1}}
    We say that $f$ satisfies languidity hypothesis \textbf{L1} with exponent $\kappa=0$ and with respect to the screen $S$ if the following hold. 
    \medskip 
    
    Firstly, there must exist some half-plane $\HH_D$ in which $f$ is holomorphic and $f$ must admit an extension to the screen $S$. Secondly, there must exist a positive constant $C>0$, a constant $\beta>D$, and a sequence of admissible heights $\set{\tau_n}_{n\in\ZZ}$ (in the sense of Equation~\ref{eqn:admissibleHeights}) such that for any $\sigma\in[S(\tau_n),\beta]$,
    \[ |f(\sigma+i\tau_n)| \leq C(|\tau_n|+1)^\kappa. \]
\end{definition}
In other words, on every horizontal contour of the form $I_n = [S(\tau_n) +i\tau_n,\,\beta +i\tau_n ]$, $f$ has at most power-law $\kappa$ growth with respect to the magnitude of its imaginary part $\tau_n$ as $|\tau_n|\to\infty$. Alternatively, on each horizontal contour $I_n$, $f$ is uniformly bounded by the constant $C_n$, where $C_n=O((|\tau_n|+1)^\kappa)$ as $|n|\to\infty$.

The second hypothesis concerns power-law/polynomial growth along a vertical curve, the screen $S$. The standard (or weak) version of languidity concerns growth along a single curve, and strong languidity concerns growth along a sequence of such curves moving to the left. These curves are precisely the screens of the window(s) to which $f$ admits an extension. 
\begin{definition}[Languidity Hypothesis \textbf{L2}]
    \label{def:languidL2}
    \index{Languidity!Hypothesis \textbf{L2}}
    A function $f$ is said to satisfy languidity hypothesis \textbf{L2} with exponent $\kappa$ and with respect to the screen $S$ if there exists a positive constant $C>0$ such that for all $\tau \in\RR$ with $|\tau|\geq 1$, $|f(S(\tau)+i\tau)| \leq C|\tau|^\kappa$.
\end{definition}

A \textit{function} which satisfies both hypotheses \textbf{L1} and \textbf{L2} with respect to the same screen $S$ and exponent $\kappa$ is said to be languid, or equivalently to have languid growth. A geometric structure (such as a generalized fractal harp, a relative fractal drum, or a self-similar system) may be called languid if its associated zeta function (respectively, its geometric zeta function, its relative zeta function, or its scaling zeta function) is languid, i.e. it satisfies languidity hypotheses \textbf{L1} and \textbf{L2}. 

\begin{definition}[(Standard/Weak) Languidity]
    \label{def:languid}
    \index{Languidity}
    We say that $f$ is \textbf{languid} with exponent $\kappa$ if there exists a screen $S$ such that $f$ satisfies both languidity hypotheses \textbf{L1} (Definition~\ref{def:languidL1}) and \textbf{L2} (Definition~\ref{def:languidL2}) with exponent to $\kappa$ and with respect to the screen $S$. 
    \medskip 

    Given a particular screen $S$, we say that a function $f$ is languid \text{with respect to} $S$ (with exponent $\kappa$) if it satisfies these hypotheses (with exponent $\kappa$) with respect to the given screen.
\end{definition}
We will not have use of this presently, but one may also say that a function $f$ is languid if there exists some $\kappa\in\RR$ and some screen $S$ for which $f$ is languid with exponent $\kappa$ and with respect to $S$.

\subsection{Strong Languidity}
%
%

Strong languidity may be seen as a sequence of languidity conditions. Namely, it concerns languidity along \index{Screen!Sequence converging to $-\infty$}\textbf{a sequence of screens} $\set{S_m}_{m\in\NN}$ \textbf{converging to} $-\infty$. To be more precise, let $\sup S$ denote the supremum of the elements of the range of $S$, i.e. $\sup S:=\sup\set{S(\tau)\suchthat\tau\in\RR}$. Then a sequence of screens $\set{S_m}_{m\in\NN}$ is said to converge to $-\infty$ if $\lim_{m\to\infty}\sup{S_m}=-\infty$. 

Additionally, we will impose that the sequence of screens has \index{Screen!Uniformly bounded Lipschitz constants}\textbf{uniformly bounded Lipschitz constants}. For a given screen $S_m$, which is by definition Lipschitz continuous, the \index{Lipschitz constant!of a screen}\index{Screen!Lipschitz constant}\textbf{Lipschitz constant} $K_m$ is the smallest constant such that for all $x,y\in\RR$, we have that $|S_m(x)-S_m(y)|\leq K_m|x-y|$. Then to say that the Lipschitz constants of a sequence of screens $\set{S_m}_{m\in\NN}$ is uniformly bounded means that there exists a constant $K$ such that for every Lipschitz constant $K_m$ of $S_m$, $K_m\leq K$. (In other words, the set of all Lipschitz constants, $\set{K_m}_{m\in\NN}$ is bounded by some finite constant $K$.)

\begin{definition}[Strong Languidity of a Function]
    \label{def:stronglyLanguid}
    \index{Languidity!Strong languidity}
    Let $f$ be a holomorphic function possessing an analytic continuation to the whole complex plane except possibly for a discrete set of singular points. Then $f$ is said to be \textbf{strongly languid} with exponent $\kappa$ if there exists a sequence $\set{S_m}_{m\in\NN}$ of screens converging to $-\infty$ with uniformly bounded Lipschitz constants such that:
    \begin{itemize}
        \item[\textbf{L1}] $f$ satisfies languidity hypothesis \textbf{L1} with respect to each screen $S_m$ and the (fixed) constant $\kappa$. (Equivalently, with respect to the replacement $S(\tau)\equiv -\infty$.)
        \item[\textbf{L2}$'$] There exist positive constants $C$ and $B$ such that for all $\tau\in\RR$ and $m\geq 1$, $f$ satisfies $|f(S_m(\tau)+i\tau)|\leq CB^{|S_m(\tau)|}(|\tau|+1)^\kappa$.
    \end{itemize}
    
\end{definition}
Note that the strong languidity condition allows for a prefactor with exponential growth related to $\norm{S_m}_\infty$, but is otherwise analogous to \textbf{L2}. Condition \textbf{L2}$'$ implies \textbf{L2} for each of the screens with finite supremum. Also, without loss of generality, the constant $C$ may be chosen to be the same in both hypotheses. 

\subsection{Strong Languidity of Self-Similar Systems}
%
%

Self-similar systems and their associated scaling zeta functions will play a critical role in the analysis to follow. As it will turn out, a self-similar system is always strongly languid in the sense that its associated scaling zeta function $\zeta_\Phi$ is strongly languid with exponent $\kappa$. We first state this for scaling operators having positive, integral multiplicities and scaling ratios in $(0,1)$.
\begin{proposition}[Languidity of Scaling Zeta Functions]
    \label{prop:languidSZF}
    \index{Languidity!of scaling zeta functions}
    Let $L:=\sum_{i=1}^m a_i M_{\lambda_i}$ be scaling operator with distinct scaling ratios $\lambda_i\in(0,1)$ and positive integral multiplicities $a_i$. Then its associated scaling zeta function $\zeta_L$ is strongly languid with respect to exponent $\kappa=0$ as in Definition~\ref{def:stronglyLanguid}. 
\end{proposition}

\begin{proof}
    This result is a corollary of Theorem~3.26 and the discussion in Section~6.4 of \cite{LapvFr13_FGCD}, as $\zeta_L$ is essentially the same as that of a self-similar fractal string (or harp). (Let the length and gap parameters be one.) Note in particular the explicit estimate in Equation~6.36 of \cite{LapvFr13_FGCD} used to establish the uniform bounds on screens with arbitrarily small real parts.
\end{proof}

For a self-similar system $\Phi$, its associated scaling operator $L_\Phi$ will always have scaling ratios in $(0,1)$ and positive, integral multiplicities. Thus, it is an immediate corollary of Proposition~\ref{prop:languidSZF} that the scaling zeta function $\zeta_\Phi=\zeta_{L_\Phi}$ is strongly languid. Extending the convention that a geometric structure itself may be dubbed languid when its associated zeta function is languid, we may say that self-similar systems are languid.  
\begin{corollary}[Languiditiy of Self-Similar Systems]
    \label{cor:languidSystem}
    \index{Languidity!of self-similar systems}
    Let $\Phi$ be a self-similar system on $\RR^\dimension$. Then $\Phi$ is strongly languid (with exponent $\kappa=0$) in the sense that its associated scaling zeta function $\zeta_\Phi$ is strongly languid (with exponent $\kappa=0$) as in Definition~\ref{def:stronglyLanguid}.
\end{corollary}
\subsection{Joint Languidity}
\label{sub:jointLanguidity}
%
%

In what follows, we will require the existence of screens $S$ for which two functions $f$ and $g$ are both languid with respect to the \textit{same} screen $S$. We shall call this \textit{joint languidity}. 

As a preliminary, we define an open right \index{Half-plane!of maximal joint extension}\textbf{half-plane of maximal joint extension} for two holomorphic functions $f$ and $g$: it is a half-plane $\HH_{\sigma_0}$, for some $\sigma_0\in[-\infty,\infty)$, such that $f$ and $g$ each admit an analytic continuation to $\HH_{\sigma_0}$ (except perhaps for a discrete subset of singularities) and with the property that for any other half-plane $\HH_{\sigma_1}$ with this property, $\HH_{\sigma_0}\subseteq\HH_{\sigma_1}$. (This is equivalent to imposing $\sigma_0\leq\sigma_1$, however one may define a more general type of domain of maximal joint extension where containment is the appropriate condition.) 

If we assume that there is some half-plane $\HH_{\sigma_1}$ in which $f$ and $g$ are both holomorphic, then a half-plane of maximal joint extension always exists. It may be $\HH_{\sigma_1}$ or else it must contain $\HH_{\sigma_1}$. If a half-plane of maximal joint extension exists, it is necessarily unique. (Any two candidates must contain the other as a subset, which implies that they are the same set.) 

\begin{definition}[Joint Languidity]
    \label{def:jointlyLanguid}
    \index{Languidity!Joint languidity}
    Let $f$ and $g$ be functions which are both holomorphic in some half-plane $\HH_{\sigma_1}$ and let $\HH_{\sigma_0}$ be their half-plane of maximal joint extension. 
    \medskip 

    We say that $f$ and $g$ are \textbf{jointly languid} with exponent $\kappa=0$ if there exists a screen $S$ contained in $\HH_{\sigma_0}$ such that the following hold. 
    \begin{itemize}
        \item There is a sequence of admissible heights for which $f$ and $g$ both satisfy languidity hypothesis \textbf{L1} with exponent $\kappa=0$ and with respect to $S$ using this sequence.
        \item Both $f$ and $g$ satisfy languidity hypothesis \textbf{L2} with exponent $\kappa=0$ and with respect to $S$. 
    \end{itemize}
\end{definition}
\section{Scaling Functional Equations}
\label{sec:generalSFEs}
%
%

In order to treat \textit{tube functions} (see Chapter~\ref{chap:geometry}) and \textit{heat content} (see Chapter~\ref{chap:heat}) in a unified way, we introduce a general framework to describe the nature of both of these functions simultaneously. Each of these quantities is a function $f_X(t)=f(t;X)$ depending on a set $X$ and a (nonnegative) parameter $t$. When we consider a self-similar set $X$, arising as the attractor of self-similar system $\Phi$, we can typically express $f_X$ in terms of the functions $f_{\ph[X]}$, with $\ph\in X$. The key idea is that, after an appropriate normalization of the quantity (e.g. volume or heat content), we obtain a scaling relation of the form $f(t;\ph[X])=f(t/\lambda_\ph^\alpha;X)$, where $\alpha$ is fixed (depending on the context, and in this work either $\alpha=1$ or $\alpha=2$) and $\lambda_\ph$ is the scaling ratio of $\ph$. 

So, if the self-similar system $\Phi$ \textit{induces a decomposition} of the form 
\[ f(t;X) = \sum_{\ph\in\Phi} f(t;\ph[X]) + R(t), \]
then we may rewrite each $f(t;\ph[X])$ in terms of a single function, $f_X$, alone but with a rescaled input. This leads to a new expression, 
\[ f_X(t) = \sum_{\ph\in\Phi} f_X(t/\lambda_\ph^\alpha) + R(t), \]
which is exactly a \textit{scaling functional equation} (with an error term $R$) induced by $\Phi$. Defining the scaling operator $L_\Phi^\alpha=\sum_{\ph\in\Phi}M_{\lambda_\ph^\alpha}$, this equation may be written as $f_X=L_\ph^\alpha[f_X]+R$. 


\subsection{Induced Decompositions and Scaling Functional Equations}
%
%

In our applications, we will have a self-similar system $\Phi$ on $\RR^\dimension$ and a collection of functions $f_A$ from a subset of $\RR$ into $\RR$. $A$ will either be a subset of $\RR^\dimension$ or else a relative fractal drum $(X,\Omega)$ in $\RR^\dimension$. In what follows, we will give meaning to the notion of a self-similar system \textit{inducing a scaling functional equation} in a theoretical manner. The idea is that this section outlines this approach for studying more general sorts of functions indexed by some collection of geometric or algebraic objects. 

The terminology is generalized from the situations which we will see in Chapter~\ref{chap:geometry} and Chapter~\ref{chap:heat} and treats them in a unified way. Further, it suggests how to proceed in an arbitrary context: find a scaling law, which may depend on some parameter (or even a more general sort of transformation), as well as a decomposition rule, and then a scaling functional equation will appear. 

Let $\Phi$ be a self-similar system on $\RR^\dimension$ (or, more generally, on a metric space) and let $\Aa$ be a collection of objects with the following closure property: if $A\in\Aa$, then $\ph[A]$ is well-defined and $\ph[A]\in \Aa$ for each $\ph\in\Phi$. For brevity, we will say that $\Aa$ is $\Phi$-closed. Next, let $\set{f_A}_{A\in\Aa}$ be a family of functions $f_A:I\to G$ for some set $I$ and an additive semigroup $(G,+)$. We will identify $f$ with the family and write $f(t;A)=f_A(t)$ for $t\in I$ and $A\in\Aa$.


\begin{definition}[Induced Decomposion]
    \label{def:inducedDecomp}
    \index{Decomposition!Induced decomposion}
    Let $\Phi$ be a self-similar system on $\RR^\dimension$, let $\Aa$ be a $\Phi$-closed collection of objects, and let $f=\set{f_A}_{A\in\Aa}$ be a family of functions $f_A:I\to G$ for some set $I$ and an additive semigroup $(G,+)$. Let also $R:I\to G$ be a function. 
    \medskip 

    A self-similar system $\Phi$ is said to \textbf{induce a decomposition} of $f$ on $I$ with remainder $R$ if for all $t\in I$, 
    \begin{equation}
        \label{eqn:inducedDecomp}
        f(t;A) = \sum_{\ph\in\Phi} f(t;\ph[A]) + R(t).
    \end{equation}
\end{definition}
In this work, $I$ is always a subset of $\RR$ and the codomain is the field $\RR$ with its standard addition operation. Further, $A$ is some geometric structure, either a subset of $\RR^\dimension$ or a relative fractal drum $(X,\Omega)$. If $A\subset\RR^\dimension$, $\ph[A]$ is its pointwise image. For a relative fractal drum, we define $\ph[(X,\Omega)]:=(\ph[X],\ph[\Omega])$, noting that this also defines a relative fractal drum.

An induced decomposition will be paired with a scaling law satisfied by the function $f(t;A)$. A linear scaling (invariance) law would be of the form $f(\lambda t;\lambda A)=f(t;A)$. However, depending on the context, the real and set parameters may scale differently. For example, heat contents (see Chapter~\ref{chap:heat}) will obey a quadratic scaling law: $f(\lambda^2\,t;\lambda A)=f(t;A)$. We provide a general notion for defining the scaling relation which allows one to specify the parameter $\alpha$ dictating this relationship. 
\begin{definition}[$\alpha$-Scaling Law]
    \label{def:scalingLaw}
    \index{Scaling law}
    Let $f=\set{f(\cdot;A)}_{A\in\Aa}$ be a family of functions $f_A:\RR^+\to\RR$ where $\Aa$ is a $\set{\ph}$-closed collection of objects for any similitude $\ph$ of $\RR^\dimension$. 
    \medskip 
    
    We say that $f$ has an $\alpha$\textbf{-scaling law} if for any similitude $\ph$ with scaling ratio $\lambda>0$ and for any $A\in\Aa$, $f(t;\ph[A])=f(t/\lambda^\alpha;A)$. 
\end{definition}

We note that this equation could alternatively have been written as $f(\lambda^\alpha t;\ph[A])=f(t;A)$. If we were to define the action of $\ph$ on the family $f$ by $\ph[f_A]:=f_{\ph[A]}\circ (t\mapsto\lambda^\alpha t)$, then this equation says that each function is invariant with respect to this action, viz. $\ph[f_A]=f_A$ for any $A\in\Aa$ and similitude $\ph$. However, the way we have written will be more useful for rewriting induced decompositions. In this form, the equation is more akin to a type of equivariance: $f_{\ph[A]}(t)=f_A(\ph*t)$, where $\ph*t:=t/\lambda^\alpha$ is an action on $(\RR^+,\cdot)$. 

We conclude this section with the main idea: a scaling law and an induced decomposition for a family of functions together lead to the scaling functional equation of a single function.
\begin{proposition}[Scaling Functional Equation (SFE)]
    \label{prop:inducedSFE}
    \index{Scaling functional equation (SFE)}
    Let $f=\set{f_A}_{A\in\Aa}$ be a family of functions on $I$ satisfying an $\alpha$-scaling law. Suppose that $\Phi$ is a self-similar system which induces a decomposition of $f$ on $I$ with remainder term $R$. Define the scaling operator $L_\Phi^\alpha$ by $L_\Phi^\alpha:=\sum_{\ph\in\Phi}M_{\lambda_\ph^\alpha}$. 
    \medskip 

    Then each function $f_A$ in the family satisfies the \textbf{scaling function equation} $f_A=L_\Phi^\alpha[f_A]+R$ on $I$ with error term $R$, which is to say that for all $t\in I$,
    \[ f_A(t)=L_\Phi^\alpha[f_A](t)+R(t). \]
    We say that $\Phi$ induces a scaling functional equation for $f_A$ with operator $L_\Phi^\alpha$.
\end{proposition}
Proposition~\ref{prop:inducedSFE} is, in essence, an untangling of definitions meant to lead to a single key concept: \textit{a self-similar system induces a scaling functional equation} for functions with scaling laws. It is, however, a consequence of the definitions that any function in this family satisfies an induced scaling functional equation. The proof is merely to apply the $\alpha$-scaling law to each term of an induced decomposition and then use the definition of the scaling operator $L_\Phi^\alpha$ to rewrite the decomposition in terms of a single function.  

In general, any function $f$ which satisfies an equation of the form $f=L[f]+R$ for some scaling operator $L$ is said to satisfy a scaling functional equation (with error $R$). If $R\equiv 0$, the scaling functional equation is said to be exact, and if not then it is said to be an approximate scaling functional equation. 

\subsection{Admissible Remainders}
%
%

As part of our analysis of scaling functional equations, we will be required to show that certain functions are languid (see Section~\ref{sec:languidity}). In so doing, we will need to constrain the types of remainders which can appear in order for these results to be established. In this section, we introduce terminology for such \textit{admissible remainders} relative to a given self-similar system as well as some sufficient conditions for admissibility. 

The admissibility of a remainder is essentially tied to a need for \textit{joint languidity}. Two functions, the Mellin transform of the remainder $\zeta_R$ and the scaling zeta function $\zeta_\Phi$, must share the same screen and the same horizontal contours on which they are both uniformly bounded. They must both be languid (with exponent $\kappa$) on a shared screen. It is for this reason that we introduced the notion of \textit{joint languidity} (Definition~\ref{def:jointlyLanguid}).

\begin{definition}[Admissible Remainders and Screens]
    \label{def:admissibleRem}
    \index{Scaling functional equation (SFE)!Admissible remainders}
    Let $R$ be a function such that its truncated Mellin transform $\zeta_R(s;\delta)=\Mellin^\delta[R](s)$ is holomorphic in some half-plane $\HH_{\sigma_0}$. This occurs, for instance, if $R$ is continuous on $(0,\delta]$ and $R(t)=O(t^{-\sigma_0})$ as $t\to0^+$ (per Lemma~\ref{lem:MellinHolo}).
    \medskip

    Let $L$ be a scaling operator with associated scaling zeta function $\zeta_L$. We say that $R$ is an \textbf{admissible remainder} for $L$ if there exists a screen $S$ such that $\zeta_L$ and $\zeta_R$ are jointly languid with exponent $\kappa=0$ with respect to $S$ in the sense of Definition~\ref{def:jointlyLanguid}. Any such screen $S$ is called an \textbf{admissible screen} for $R$ and $L$ (or their respective zeta functions). 
    \medskip 

    Given a self-similar system $\Phi$, we say that $R$ is an admissible remainder for $\Phi$ if it is admissible for the associated operator $L_\Phi$ (and similarly for admissible screens). 
\end{definition}

The easiest way to obtain a nontrivial admissible screen occurs when $\zeta_R$ is holomorphic in a half-plane to the left of the vertical strip containing all of the poles of $\zeta_L$. This ensures that a screen may be chosen to the left of any pole of $\zeta_L$ but also in a region where $\zeta_R$ is holomorphic. Supposing that $L=L_\Phi$ is the scaling operator associated to a self-similar system, we have that this strip is bounded explicitly by the similarity dimensions of $\Phi$: if $\omega\in\Dd_\Phi$, then $\lowersimdim(\Phi)\leq\Re(\omega)\leq\simdim(\Phi)$. So, the first criterion is to ensure an estimate for $R$ that guarantees its abscissa of holomorphic convergence is strictly smaller than the lower similarity dimension.

\begin{theorem}[Lower Dimension Criterion for Admissibility]
    \label{thm:lowerDimAdmissibility}
    Let $\Phi$ be a self-similar system and let $D_\ell=\underline{\dim}_S(\Phi)$ be its lower similarity dimension (see Definition~\ref{def:lowerSimDim}). For any $R\in C^0(\RR^+)$, if there exists $\sigma_0<D_\ell=\underline{\dim}_S(\Phi)$ such that as $t\to0^+$, $R(t)=O(t^{-\sigma_0})$, then $R$ is an admissible remainder for $\Phi$ and any screen of the form $S_\e(\tau)\equiv \sigma_0+\e$, with $0<\e<D_\ell-\sigma_0$, is admissible. 
\end{theorem}
\begin{proof}
    Let $D_\ell=\underline{\dim}_S(\Phi)$ be the lower similarity dimension of $\Phi$ and let $S_\e(\tau)\equiv \sigma_0+\e$ be a constant screen in $\HH_{\sigma_0}$, where $\e>0$ is such that $\sigma_0+\e<D_\ell$. We will show that $\zeta_R$ and $\zeta_\Phi$ are jointly languid with exponent $\kappa=0$ on $S_\e$.  

    To start, we have that $\zeta_\Phi$ is strongly languid with exponent $\kappa=0$ by Corollary~\ref{cor:languidSystem}. This guarantees that for any $\sigma_-$, there exists of a sequence of admissible heights $\set{T_n}_{n\in\ZZ}$ (i.e. with $T_{-n}\to-\infty$ and $T_n\to\infty$ as $n\to\infty$ and with $T_n>0>T_{-n}$ for each $n\geq 1$) and some screen $S_-$ with $\sup S_- < \sigma_-$ such that $\zeta_\Phi$ is uniformly bounded on horizontal contours of the form $[S_-(T_n)+i T_n,\sigma_++i T_n]$, where $\sigma_+>\max(\sigma_-,\dim_S(\Phi))$. In particular, $\zeta_\Phi$ will be uniformly bounded on the subsets $I_n:=[\sigma_0+\e+iT_n,\sigma_++iT_n]$. This establishes languidity hypothesis \textbf{L1} for $\zeta_\Phi$. 

    Next, we note that the estimate on $R$ implies that the function $\zeta_R$ is holomorphic in the open right half-plane $\HH_{\sigma_0}$ by Lemma~\ref{lem:MellinHolo}. By Corollary~\ref{cor:MellinBounds}, we also have that $\zeta_R$ is bounded on the screen $S_\e(\tau)\equiv \sigma_0+\e>\sigma_0$ as well as on any vertical strip of the form $\HH_{\sigma_0+\e}^{\sigma_+}$. This establishes both languidity hypotheses \textbf{L1} and \textbf{L2} with exponent $\kappa=0$ and on the screen $S_\e$ with respect to the intervals $I_n\subset \HH_{\sigma_0+\e}^{\sigma_+}$, with shared sequence of admissible heights.
    
    It remains to show that $\zeta_\Phi$ is bounded on the screen $S_\e$. To this end, we enforce that $\sigma_0+\e<D_\ell$, in which case for any $s$ on the screen, $\Re(s)=\sigma_0+\e<D_\ell$. Let $\set{r_k}_{k=1}^m$ denote the set of unique scaling ratios of $\Phi$, arranged in decreasing order $r_1\geq ...\geq r_M$, and let $m_k$ denote the multiplicity of $r_k$. Define the function 
    \begin{equation}
        p(t) = \frac1{m_M}(r_M^{-1})^t+\sum_{k=1}^{M-1}\frac{m_k}{m_M}\left(\frac{r_k}{r_M}\right)^t,
    \end{equation}
    which is readily seen to be strictly increasing with range $(0,\infty)$. Note that $p(D_\ell)=1$ by definition. We will obtain a bound for $\zeta_\Phi(s)$ when $\Re(s)=\sigma_0+\e<D_\ell$ using this function. 
    
    To that end, let $f(s)=\zeta_\Phi(s)^{-1}=1-\sum_{k=1}^M m_k r_k^s$ be the denominator. Then 
    \begin{align*}
        |r_M^{-s}/m_M f(s)+1| &= \left| r_M^{-s}/m_M - \sum_{k=1}^{M-1} m_k/m_M(r_k/r_M)^s \right| \\
            &\leq p(\sigma) < p(D_\ell) = 1.
    \end{align*}
    By the reverse triangle inequality, we have that 
    \begin{align*}
        1 > p(\sigma) \geq \left| |r_M^{-s}/m_M f(s)| - 1 \right| = \left| r_M^{-\sigma}/m_M|f(s)| - 1 \right|,
    \end{align*}
    which may be rewritten as $| |f(s)|-m_M r^{\sigma} | \leq m_M r^{\sigma}p(\sigma)<m_M r_M^\sigma$. Once again using the reverse triangle inequality, we find that the lower bound for $|f|$ is furnished by
    \begin{align*}
        |f(s)| \geq m_M r_M^\sigma - | |f(s)|-m_M r^{\sigma} | \geq m_M r_M^\sigma - m_M r^{\sigma}p(\sigma) >0.
    \end{align*}
    It follows that when $\Re(s)=\sigma_0+\e$ is fixed, $|\zeta_\Phi(s)|\leq C_\Phi$, with $C_\Phi = (m_M r_M^\sigma - m_M r^{\sigma_0+\e}p(\sigma_0+\e))^{-1}$. This establishes hypothesis \textbf{L2} for $\zeta_\Phi$, and thus $\zeta_\Phi$ and $\zeta_R$ are jointly languid on the screen $S_\e$.
\end{proof}

In order to accommodate $\alpha$-scaling laws with $\alpha>0$ (in the sense of Definition~\ref{def:scalingLaw}), we will consider scaling operators of the form $L_\Phi^\alpha:=\sum_{\ph\in\Phi}M_{\lambda_\ph^\alpha}$, where $\Phi$ is a self-similar system. In this case, the zeta function associated to $L_\Phi^\alpha$ is exactly the function $\zeta_\Phi(\alpha s)$, where $\zeta_\Phi(s)$ is the scaling zeta function associated to $\Phi$ itself. 

This change of variables amounts to a simple rescaling of the bound with respect to the lower similarity dimension of $\Phi$ a remainder must satisfy in order to be admissible by Theorem~\ref{thm:lowerDimAdmissibility}. Namely, when $R(t)=O(t^{-\sigma_R})$ as $t\to0^+$ for some $\sigma_R<\lowersimdim(\Phi)/\alpha$, we have that $\zeta_\Phi(\alpha s)$ has poles in the strip $\lowersimdim(\Phi)/\alpha \leq \Re(s) \leq \simdim(\Phi)/\alpha$ and $\sigma_R$ lies strictly to the left of this bound. The estimates for $\zeta_\Phi$ apply to its rescaled analogue, and thus proof of Theorem~\ref{thm:lowerDimAdmissibility} yields the following corollary.

\begin{corollary}[Rescaled Lower Dimension Criterion]
    \label{cor:lowerDimAdmissibilityRescaled}
    Let $\Phi$ be a self-similar system and let $D_\ell=\underline{\dim}_S(\Phi)$ be its lower similarity dimension (see Definition~\ref{def:lowerSimDim}). Given $\alpha>0$, define the scaling operator $L_\Phi^\alpha:=\sum_{\ph\in\Phi}M_{\lambda_\ph^\alpha}$. 
    \medskip     
    
    For any $R\in C^0(\RR^+)$, if there exists $\sigma_R<D_\ell/\alpha$ such that as $t\to0^+$, $R(t)=O(t^{-\sigma_R})$, then $R$ is an admissible remainder for $L_\Phi^\alpha$ and any screen of the form $S_\e(\tau)\equiv \sigma_R+\e$, with $0<\e<D_\ell/\alpha-\sigma_R$, is admissible. 
\end{corollary}

\begin{proof}
    By Lemma~\ref{lem:MellinHolo}, we have that $\zeta_R(s)$ is holomorphic in $\HH_{\sigma_R}$ and thus $\zeta_R(z/\alpha)$ is holomorphic when $\Re(z)>\alpha\,\sigma_R$. Further, $\zeta_R(s)$ is bounded on any vertical strip of the form $\HH_a^b$, $\sigma_R<a\leq b<\infty$, so $\zeta_R(z/\alpha)$ is bounded on the corresponding strips $\HH_{a'}^{b'}$, $\alpha\,\sigma_R<a'\leq b'<\infty$. 
    
    Since $\alpha\,\sigma_R<D_\ell$, by repeating the proof of Theorem~\ref{thm:lowerDimAdmissibility} (with $\sigma_0=\alpha\,\sigma_R$) we obtain that $\zeta_\Phi(z)$ and $\zeta_R(z/\alpha)$ are jointly languid on screens of the form $S_\e(\tau)\equiv \alpha\,\sigma_R+\e$ when $0<\e<D_\ell -\alpha\,\sigma_R$. Taking $z=\alpha s$, we see that the functions $\zeta_\Phi(\alpha s)$ and $\zeta_R(s)$ are jointly languid on screens of the form $S'_{\e'}(\tau)\equiv \sigma_R+\e'$ where $\e'=\e/\alpha\in(0,D_\ell/\alpha-\sigma_R)$. 
\end{proof}

The next criterion is related to when we have explicit knowledge of the locations of the singularities of $\zeta_\Phi$. In particular, in the lattice case (see Definition~\ref{def:latticeDichotomy}), we can explicitly show that all of the poles lie on one of finitely many vertical lines (and are distributed periodically along these lines with a shared period for each line). Thus, we can easily choose screens within this region which will never encounter singularities of $\zeta_\Phi$, with distance to any pole bounded by the distance of the real part of the screen to the real part of the closest of the finitely many exceptional points.

\begin{theorem}[Lattice Criterion for Admissiblility]
    \label{thm:latticeCaseAdmissibility}
    Let $\Phi$ be a self-similar system and suppose that its scaling ratios $\set{\lambda_\ph}_{\ph\in\Phi}$ are arithmetically related (see Definition~\ref{def:latticeDichotomy}). Let $R$ be a continuous function on $\RR^+$ with the estimate that $R(t)=O(t^{-\sigma_0})$ as $t\to0^+$ for some $\sigma_0\in\RR$. 
    \medskip
    
    Then for all but finitely many $\sigma>\sigma_0$, $S_\sigma(\tau)\equiv \sigma$ is an admissible screen. Consequently, there are admissible screens of the form $S_{\sigma_0+\e}$ for any $\e>0$ sufficiently small.
\end{theorem}
\begin{proof}
    By definition, in the lattice case there exists some $\lambda_0\in\RR^+$ such that for each $\ph\in\Phi$, there is a positive integer $k_\ph$ such that $\lambda_\ph=\lambda_0^{k_\ph}$. Under this assumption, we may explicitly write the denominator of $\zeta_\Phi$ as the Dirichlet polynomial 
    \[ P(s) = 1-\sum_{\ph\in\Phi} \lambda_\ph^s = 1- \sum_{\ph\in\Phi} \lambda_0^{k_\ph s}. \]
    Under the change of variables $s=\log_{\lambda_0}z$ (so that $\lambda_0^s=z$), we have that 
    \[ P(\log_{\lambda_0}z) = 1- \sum_{\ph\in\Phi} z^{k_\ph}. \]
    This is precisely a polynomial in the variable $z$ (since $|\Phi|=n<\infty$), so by the fundamental theorem of algebra it has finitely many roots (exactly $K=\max(\set{k_\ph}_{\ph\in\Phi})$) in $\CC$. 
    
    Denote these roots by $Z=\set{z_j}_{j=1}^K$. Any solution of $P(\omega)=0$ must then be of the form $\lambda_0^\omega=z_j$ for some $z_j\in Z$. We note that $0\notin Z$ since the polynomial equals one when $z=0$, so there exists a branch of the logarithm for which $\log z_j$ is well defined for each $z_j\in Z$ and without loss of generality we may find a single branch defined for each $z_j\in Z$ and for $\lambda_0$ since $Z$ is finite. The logarithm is multivalued, however, with $\lambda_0^\omega=e^{\omega\log\lambda_0}=e^{\omega\log\lambda_0+2\pi i m}$ for any $m\in\ZZ$. So, we will obtain as solutions to $P(\omega)=0$ exactly the points $\omega_{j,m}$ of the form $\omega_{j,m}= \log(z_j)/\log(\lambda_0) + 2\pi i m/\log(\lambda_0)$. 

    Observe that there are finitely many real parts, $\sigma_j=\Re(\log(z_j)/\log(\lambda_0))$, $j=1,...,K$, at which these poles occur and that the imaginary parts are all distributed with the same period, $2\pi/\log\lambda_0$. We will show that for any screen of the form $S_\sigma(\tau)\equiv \sigma$, $\sigma\neq\sigma_j$ for each $j=1,...,K$, $\zeta_\Phi$ is languid with exponent $\kappa=0$ with respect to $S_\sigma$, starting with hypothesis \textbf{L1}. 

    The next two steps of the proof are the same as in the proof of Theorem~\ref{thm:lowerDimAdmissibility}. In short, by the strong languidity of $\zeta_\Phi$, for any $\sigma_-$, there exists a sequence of admissible heights $\set{T_n}_{n\in\ZZ}$ on which $\zeta_\Phi$ is bounded on horizontal intervals of the form $[\sigma_-+iT_n,\sigma_++iT_n]$ where $\sigma_-$ is arbitrarily small and $\sigma_+>\max(\simdim(\Phi),\sigma_-)$. Choosing $\sigma_-=\sigma$ establishes hypothesis \textbf{L1} with respect to the screen $S_\sigma$ for $\zeta_\Phi$. Secondly, we have that $\zeta_R$ satisfies hypotheses \textbf{L1} and \textbf{L2} for the sequence of heights as above on any screen of the form $S_\sigma$, $\sigma>\sigma_0$, because it is a holomorphic Mellin transform. 

    It remains only to show that $\zeta_\Phi$ satisfies hypothesis \textbf{L2} on $S_\sigma$ when $\sigma$ is not one of the exceptional values $\sigma_j$, $j=1,...,K$. We will show that $f(\tau)$ is bounded from below by a strictly positive constant on all of $\RR$, and thus $\zeta_\Phi$ will be bounded on $S_\sigma$. As we have shown, $P(\sigma+i\tau)\neq0$ whenever $\sigma\notin\set{\sigma_j}_{j=1}^K$. Let $f(\tau)=|P(\sigma+i\tau)|$. Given any finite interval $[a,b]$, we have that $f(t)$ must be nonzero and bounded from below by a strictly positive constant. This follows by continuity and the intermediate value theorem, noting that if $f(\tau)=0$, then $P(\sigma+i\tau)=0$, which is a contradiction. 
    
    We will use (multiplicative) periodicity of $f$ to show that $f(t)=|P(\sigma+i\tau)|$ cannot become arbitrarily small as $|\tau|\to\infty$. Recall that  
    \[
        P(\sigma+i\tau) = 1 - \sum_{j=1}^K \lambda_0^{k_j(\sigma_0+i\tau)} = 1- \sum_{j=1}^K \lambda_0^{k_j\sigma}\cdot e^{i (k_j\log\lambda_0)\tau}.
    \]
    Under the transformation $\tau\mapsto (2\pi m/\log\lambda_0) \,\tau$, for any $m\in\ZZ$, we have that each exponential is invariant since $k_j$ is an integer. Thus, $P(\sigma+i\tau)=P(\sigma+i(2\pi m/\log\lambda_0)\tau)$. Choosing $m=-1$ shows us that the function $f(\tau):=|P(\sigma+i\tau)|$ is multiplicatively periodic with period $p=-2\pi/\log\lambda_0=2\pi/\log\lambda_0^{-1}$ or a rational multiple thereof, depending on the greatest common divisor of the integers $k_j$, $j=1,...,K$. However, if they share a greatest common divisor $\text{GCD}$, it is possible to redefine $\lambda_0$ by $\lambda_0'=\lambda_0^{\text{GCD}}$ and obtain a new set of integers without this property. Note that $\log(\lambda_0^{-1})>0$.

    By the previous argument, $f(\tau)=|P(\sigma+i\tau)|$ is bounded from below on any bounded interval. So, $f$ is nonzero on the interval $[-1,1]$. Now pick one full multiplicative period in $(0,\infty)$, say $[p^{m_0},p^{m_0+1}]$ (if $p>1$) or $[p^{m_0+1},p^{m_0}]$ (if $p<1$). The function $f$ must be bounded from below strictly away from zero as this is a finite interval. By periodicity, this same bound applies to any interval of the form $[p^{m},p^{m+1}]$ or $[p^{m+1},p^{m}]$, respectively, for any $m\in\ZZ$. Since the function $p\mapsto p^t$ is surjective onto $(0,\infty)$, this implies that $f$ is bounded uniformly from below by a strictly positive constant, the same as the first bound, when $\tau>0$. For $\tau<0$, we can use the starting interval $[-p^{m_0+1},-p^{m_0}]$ or $[-p^{m_0},-p^{m_0+1}]$, depending on whether $p>1$ or $p<1$, and the same argument to deduce that it is also bounded from below by a strictly positive constant on $(-\infty,0)$. Taking the minimum of these three bounds shows that $f(\tau)=|P(\sigma+i\tau)|>C>0$. 
    
    It follows that $\zeta_\Phi$ is bounded from above on $S_{\sigma_0}$, establishing hypothesis $\textbf{L2}$. Thus, we have established the joint languidity of $\zeta_\Phi$ and $\zeta_R$ on any screen of the form $S_\sigma$, where $\sigma>\sigma_0$ and where $\sigma\notin\set{\sigma_j}_{j=1}^K$. When choosing a screen $S_{\sigma_0+\e}$, choosing $\e$ with $0<\e<\min_{j=1,...,K}(|\sigma_j-\sigma_0|)$ is sufficient.
\end{proof}

\subsection{Zeta Functions from SFEs}
%
%

Now, we endeavor to solve scaling functional equations in a general framework. The main tool is the truncated Mellin transform (see Section~\ref{sec:Mellin}), $\Mellin^\delta$, where $\delta>0$ is fixed. In what follows, we make the convention that $\Mellin^\delta[f](s)$ is the zeta function $\zeta_f(s;\delta)$ of $f$. The first step of the solution is to apply a fixed Mellin transform, $\Mellin^\delta$, to both sides of the SFE. In doing so, we will be able to produce explicit formulae for the zeta function of a solution, prove that it is meromorphic, and most notably to describe explicitly the locations of its possible poles.  

The following functions will appear in our analysis, thus we introduce notation for them. Let $L:=\sum_{k=1}^m a_k M_{\lambda_k}$ be a scaling operator. We define 
\begin{equation}
    \label{eqn:defPartialZeta}
    \xilf(s;\delta) := \sum_{k=1}^m a_k\lambda_k^s\Mellin_\delta^{\delta/\lambda_k}[f](s)
\end{equation}
and if $\Phi$ is a self-similar system, then define $\xilfphi:=\xi_{L_\Phi,f}$, where $L_\Phi$ is the scaling operator associated to $\Phi$.  Additionally, given the functions $f$ and $R$, we define $\zeta_f(s;\delta)=\Mm^\delta[f](s)$ and $\zeta_R(s;\delta)=\Mm^\delta[R](s)$. Lastly, we recall that $\zeta_L$ is the scaling zeta function of $L$, with singular set $\Dd_L$ (see Definition~\ref{def:scalingZetaFunction}). 

\begin{theorem}[Zeta Functions of Solutions to General SFEs]
    \label{thm:zetaFormula}
    Let $f,R\in C^0(\RR^+)$ and suppose that there exist $\sigma_R,\sigma_1\in\RR$ such that $R(t)=O(t^{-\sigma_R})$ and $f(t)=O(t^{\sigma_1})$ as $t\to0^+$. Suppose that $L:=\sum_{k=1}^m a_k M_{\lambda_k}$ is a scaling operator for which $f$ satisfies the scaling functional equation $f=L[f]+R$ for all $t\in (0,\delta]$. 
    \medskip
 
    Then for all $s\in\HH_{\sigma_R}\setminus \Dd_L$, 
    \begin{equation}
        \label{eqn:ZetaFormula}
        \zeta_f(s;\delta)= \zeta_L(s)(\xilf(s;\delta)+\zeta_R(s;\delta)),
    \end{equation}
    where $\xilf$, defined as in Equation~\ref{eqn:defPartialZeta}, is an entire function. Further, $\zeta_f$ is holomorphic in the half-plane $\HH_{\max(D,\sigma_R)}$, where $D=\dim_S(\Phi)$, and admits a meromorphic continuation to $\HH_{\sigma_R}$ with any poles contained in $\Dd_L(\HH_{\sigma_R})$. Further, these poles are independent of the choice of $\delta>0$.
\end{theorem}
\begin{proof}
    We note that by Lemma~\ref{lem:MellinHolo}, $\zeta_f(s;\delta)=\Mm^\delta[f](s)$ and $\zeta_R(s;\delta)=\Mm^\delta[R](s)$ are well-defined and holomorphic in the half-plane $\HH_{\max(\sigma_R,\sigma_1)}$. Additionally, we have that by Lemma~\ref{lem:MellinScaling}, for each $M_{\lambda_k}$, 
    \[ \Mellin^\delta[M_{\lambda_k}[f]](s) = \lambda_k^s \Mellin^{\delta/\lambda_k}[f](s), \]
    noting that $M_{\lambda_k}$ is defined with a reciprocal scaling convention. It follows that $\Mellin^\delta[M_{\lambda_k}[f]]$ is well-defined and holomorphic in this same half-plane for each $k=1,...,m$. 

    So, we may apply the truncated scaling transform $\Mm^\delta$ to each term of the scaling functional equation $f=L[f]+R$, using linearity of the Mellin transform to commute with the summation. Thus, we compute that 
    \begin{align*}
        \Mm^\delta[f](s) &= \Mellin^\delta\left[\sum_{k=1}^m a_k M_{\lambda_k}[f]\right](s) + \Mm^\delta[R](s) \\
            &= \sum_{k=1}^m a_k \Mellin^\delta[M_{\lambda_k}[f]](s) + \Mm^\delta[R](s) \\
            &= \sum_{k=1}^m a_k \lambda_k^s \Mellin^{\delta/\lambda_k}[f](s) + \Mm^\delta[R](s) \\
            &= \sum_{k=1}^m a_k \lambda_k^s \Mellin^{\delta}[f](s) 
                + \sum_{k=1}^m a_k \lambda_k^s \Mellin_\delta^{\delta/\lambda_k}[f](s)+ \Mm^\delta[R](s). 
    \end{align*}
    In the last step, we have partitioned the integrals defining each Mellin transform, splitting each integral over $[0,\delta/\lambda_k]$ into two integrals, over $[0,\delta]$ and $(\delta,\delta/\lambda_k]$ respectively. Also, we note that the second sum is exactly the definition of $\xilf$. Rearranging this equation and using the definitions of $\zeta_f$ and $\zeta_R$, we find that 
    \begin{align*}
        \zeta_f(s;\delta)\Big(1-\sum_{k=1}^m a_k\lambda_k^s\Big) = \xilf(s;\delta)+\zeta_R(s;\delta).
    \end{align*}
    Provided that $s\in\HH_{\max(\sigma_R,\sigma_1)}\setminus\Dd_L$, it follows that $\Mm^\delta[f](s)=\zeta_L(s)(\xilf(s;\delta)+\zeta_R(s;\delta))$.

    From this equation, we can deduce some properties of $\zeta_f$. Firstly, by Lemma~\ref{lem:MellinHolo}, we have that each function $\Mellin_\delta^{\delta/\lambda_k}[f]$, $k=1,...,m$, is entire (using the fact that since $f$ is continuous, it is bounded on any compact interval of the form $[\alpha,\beta]$). Since $\xilf$ is a finite linear combination of these functions, each multiplied by the respective entire function $\lambda_k^s$, it must also be entire. Secondly, we have that the right-hand side is holomorphic in the half plane $\HH_{\max(D,\sigma_R)}$. Thus, $\zeta_f$ admits an analytic continuation to this region and is holomorphic therein (irrespective of the value of $\sigma_1$). Additionally, since $\zeta_L$ admits a meromorphic continuation to the whole complex plane, we find that $\zeta_f$ admits a meromorphic continuation to the half-plane $\HH_{\sigma_R}$ where $\zeta_R$ is holomorphic, thus having poles with a given multiplicity at a point only if the function $\zeta_L$ has a pole of that multiplicity of higher at that point. 

    Lastly, the independence of the poles of $\zeta_f$ on the choice of $\delta$ may be seen as a corollary of Lemma~\ref{lem:MellinHolo}. For any $\delta_1,\delta_2>0$, we have that $\Mellin^{\delta_2}[f](s)-\Mellin^{\delta_1}[f](s) = \Mellin_{\delta_1}^{\delta_2}[f](s)$, which is an entire function. Thus, these two functions must have the same poles. 

\end{proof}

We now specialize to operators of the form $L_\Phi^\alpha:=\sum_{\ph\in\Phi}M_{\lambda_\ph^\alpha}$, where $\Phi$ is a self-similar system and $\alpha>0$. We require that $\alpha>0$ so that $\lambda_\ph^\alpha\in(0,1)$ for each $\ph\in\Phi$. When $\alpha=1$, this is exactly the scaling operator $L_\Phi$ associated to $\Phi$. For other values of $\alpha$, the effect on the associated zeta function is simply a rescaling of the input relative to $\zeta_\Phi$, namely $\zeta_{L_\Phi^\alpha}(s)=\zeta_\Phi(\alpha s)$. (This is an immediate consequence of the definition of these functions and elementary properties of exponents.) Note that it follows that if $D=\dim_S(\Phi)$ is the abscissa of convergence of $\zeta_\Phi$, then $\zeta_{L_{\Phi}^\alpha}$ has the corresponding abscissa of convergence $D/\alpha$. Further, the sets of poles are directly related: any pole $\omega\in\Dd_\Phi$ exactly corresponds to the pole $\omega/\alpha\in\Dd_{L_\Phi^\alpha}$. 

These observations together yield the following corollary of Theorem~\ref{thm:zetaFormula}.
\begin{corollary}[Zeta Functions of Solutions to Self-Similar SFEs]
    \label{cor:structureOfPoles}
    Let $f,R\in C^0(\RR^+)$ and suppose that there exist $\sigma_R,\sigma_1\in\RR$ such that $R(t)=O(t^{-\sigma_R})$ and $f(t)=O(t^{\sigma_1})$ as $t\to0^+$. Let $\Phi$ be a self-similar system, let $\alpha>0$, and let ${L_\Phi^\alpha}:=\sum_{\ph\in\Phi}M_{\lambda_\ph^\alpha}$. If $f$ satisfies the scaling functional equation $f={L_\Phi^\alpha}[f]+R$ for all $t\in (0,\delta]$, then
    \begin{equation}
        \label{eqn:zetaFormulaAlpha}
        \zeta_f(s;\delta) = \zeta_\Phi(\alpha s)(\xilfphia(s;\delta)+\zeta_R(s;\delta))
    \end{equation}
    and is holomorphic in $\HH_{\max(D/\alpha,\sigma_R)}$, where $D=\dim_S(\Phi)$. Further, $\zeta_f(s;\delta)$ is meromorphic in $\HH_{\sigma_R}$ with poles contained in a subset of $\alpha^{-1}\Dd_\Phi(\HH_{\sigma_R})$ and are independent of $\delta$. 
\end{corollary}
\subsection{Mellin Inversion}
%
%

The next step of solving a scaling functional equation is to apply the Mellin inversion theorem. This result is analogous to the Fourier inversion theorem, and provides a means by which one may invert the Mellin transform; see for instance \cite{Tit86}. We shall make use of the following notation convention: 
\[ \int_{c-i\infty}^{c+i\infty} f(z)\,dz := \lim_{T\to\infty}\int_{c-iT}^{c+iT}f(z)\,dz, \]
which may also be considered a principle value integral such as in \cite{Gra10}. Combined with Theorem~\ref{thm:zetaFormula}, the function may be written in terms of $\zeta_\Phi$, $\zeta_R$, and $\xilfphia$. 
\begin{theorem}[Mellin Inversion Formula]
    \label{thm:MellinInversion}
    Let $f,R\in C^0(\RR^+)$. Let $\Phi$ be a self-similar system, let $\alpha>0$, and let ${L_\Phi^\alpha}:=\sum_{\ph\in\Phi}M_{\lambda_\ph^\alpha}$. Suppose that $f$ satisfies the scaling functional equation $f={L_\Phi^\alpha}[f]+R$ for all $t\in [0,\delta]$ and that $\zeta_R(s;\delta)=\Mellin^\delta[R](s)$ is holomorphic in $\HH_{\sigma_R}$. Lastly, let $D=\dim_S(\Phi)$ and $c>\max(D/\alpha,\sigma_R)$. Then for any $t\in(0,\delta)$, $f$ is given by:
    \begin{align*}
        f(t)    &= \Mm^{-1}[\zeta_\Phi(\alpha s)(\xilfphia(s;\delta)+\zeta_R(s;\delta))](t)\\
                &= \frac{1}{2\pi i} \int_{c-i\infty}^{c+i\infty} 
                    t^{-s} \zeta_\Phi(\alpha s)(\xilfphia(s;\delta)+\zeta_R(s;\delta))\,ds.
    \end{align*}
\end{theorem} 
\begin{proof}
    This result is an application of the standard Mellin inversion theorem applied to the result of Corollary~\ref{cor:structureOfPoles}. Technically, we modify the function $f$ slightly, namely by considering $f(t)\1_{[0,\delta]}(t)$ which is the same as $f$ when $t\leq \delta$. By Corollary~\ref{cor:structureOfPoles}, we have that $\zeta_\Phi(\alpha s)$ is holomorphic when $\Re(\alpha s)=\alpha \Re(s)>D$ and thus that $\Mellin[f\cdot\1_{[0,\delta]}]=\Mellin^\delta[f]$ is holomorphic in $\HH_{\max(D/\alpha,\sigma_R)}$. 
    
    Thus by application of the Mellin inversion theorem to Equation~\ref{eqn:zetaFormulaAlpha}, we have that for any $c>\max(D/\alpha,\sigma_R)$, 
    \[ 
        f(t)\cdot\1_{[0,\delta]}(x) 
            = \frac1{2\pi i}\int_{c-i\infty}^{c+i\infty} t^{-s} \zeta_\Phi(\alpha s)(\xilfphia(s;\delta)+\zeta_R(s;\delta))\,ds. 
    \]
    Since $t\in(0,\delta)$, we have that $\1_{[0,\delta]}(t)=1$, whence the result follows.
\end{proof}

It remains to evaluate this contour integral by means of the residue theorem. Strictly speaking, it requires a sort of unbounded residue theorem, using a sequence of regions for which the integrand has finitely many poles (on which the residue theorem applies) that converge to the region containing all of the countably many singularities. The conditions of languidity are necessary in this process, used to estimate contour integrals over pieces of the boundaries of regions involved in this process. This method developed to establish explicit formulae for fractal tube functions in \cite{LapvFr13_FGCD} and then, with refinements to the proof, used to prove explicit formulae for relative tube functions in \cite{LRZ17_FZF}. Applying this method in the general setting to obtain explicit formulae is the content Section~\ref{sec:SFEsolutions}. 

\section{Solutions of Scaling Functional Equations}
\label{sec:SFEsolutions}
%
%

In this section, we will obtain explicit formulae for functions which satisfy a scaling functional equation. Given a scaling operator $L$, a scaling functional equation is a relation of the form $f=L[f]+R$, where $R$ is a remainder term. The goal will be to obtain formulae which are valid up to an asymptotic order determined by estimates of an \textit{admissible remainder} term (in the sense of Definition~\ref{def:admissibleRem}). 

We focus on scaling functional equations which are induced by a self-similar system on functions which satisfy certain scaling laws (see Proposition~\ref{prop:inducedSFE}), as these are the main type of scaling functional equation that we will need for our applications in Chapters~\ref{chap:geometry} and \ref{chap:heat}. Explicitly, let $\Phi$ be a self-similar system, let $\alpha>0$, and define the scaling operator $L_\Phi^\alpha:=\sum_{\ph\in\Phi}M_{\lambda_\ph^\alpha}$. The scaling functional equations we study herein will be stated in terms of such scaling operators. Note, though, that any general scaling operator with scaling ratios in $(0,1)$ each having positive, integral multiplicity can be written in this form. 

There are two types of explicit formulae, namely pointwise and distributional. The former have the advantage of being simpler in their statements, but the disadvantage that they are derived only for antiderivatives of the function satisfying the scaling functional equation. The distributional explicit formulae we obtain do not have this restriction, but they will require additional terminology to state and interpret in a weak formulation. 

\subsection{Pointwise Explicit Formulae}
%
%

The pointwise explicit formulae are for antiderivatives of a function $f\in C^0(\RR^+)$ which satisfies a scaling functional equation. To that end, we introduce the following notation and convention. Firstly, define $f^{[0]}:=f$. For any $k>0$, we define $f^{[k]}$ recursively by integrating the previous antiderivative and imposing the convention that $f^{[k]}(0)=0$. Namely, for $k>0$,
\[
    f^{[k]}(t) := \int_0^t f^{[k-1]}(\tau)\,d\tau.
\] 
We will also denote $\zeta_f(s;\delta)=\Mellin^\delta[f](s)$. In our applications, typically we have some function $F(t)$ of interest which must be normalized in order to satisfy a scaling law in the sense of Definition~\ref{def:scalingLaw}, so we will have that $f(t) = t^{-\beta} F(t)$ for a given parameter $\beta$. For tube functions in Chapter~\ref{chap:geometry}, $\beta$ will be the dimension $\dimension$ and for heat content in Chapter~\ref{chap:heat} it will be $\dimension/2$. When we state the explicit formula, we will do so in terms of the parameter $\beta$, noting that typically $\beta=\dimension$. Let $F^{[k]}$ be defined in the same manner as $f^{[k]}$, $k\geq 0$. 

As an additional preliminary to stating the result, we define the Pochhammer symbol $(z)_w:=\Gamma(z+w)/\Gamma(w)$ for $z,w\in\CC$. Note that when $w=k$ is a positive integer, this simplifies to $(z)_k=z(z+1)\cdots(z+k-1)$ or, when $w=0$, to $(z)_0=1$.

Lastly, we must give meaning to the summations which will appear in what follows. Let $\Dd\subset\CC$ be a discrete subset with the property that for any $m\in\NN$, 
\[ 
    \Dd_m:=\set{\omega\in\Dd\suchthat |\Im(\omega)|\leq m}<\infty.
\] 
Then we define 
\[ 
    \sum_{\omega\in\Dd} a_\omega := \lim_{m\to\infty} \sum_{\omega\in\Dd_m}a_\omega.  
\]
Note that this summation is a \textit{symmetric limit} of finite partial sums. As such, the convergence of these sums is comparatively more delicate; it may be that the limit, when taken independently, does not exist. 

\begin{theorem}[Pointwise Explicit Formula]
    \label{thm:pointwiseFormula}
    Let $\Phi$ be a self-similar system, $\alpha>0$, $\beta\in\RR$, and $f(t)=t^{-\beta/\alpha}F(t)\in C^0(\RR^+)$. Suppose that $\Phi$ induces the scaling functional equation $f=L_\Phi^\alpha[f]+R$ on $[0,\delta]$ with admissible remainder term $R$ (in the sense of Definition~\ref{def:admissibleRem}) with corresponding screen $S$. Let $\sigma_R$ denote the abscissa of holomorphicity of $\zeta_R$, $D=\dim_S(\Phi)$ the (upper) similarity dimension of $\Phi$, and suppose $\beta/\alpha+1>\max(D/\alpha,\sigma_R)$. Lastly, suppose that $S$ is contained in the half-plane $\HH_{\sigma_R}$ and write $W_S$ for the window to the right of $S$.
    \medskip 

    Then for every $k\geq 2$ in $\ZZ$ and every $t\in(0,\delta)$, we have that
    \[ 
        F^{[k]}(t) = \sum_{\omega\in\Dd_\Phi(\alpha W_S)} 
            \Res\Bigg(\cfrac{t^{(\beta-s)/\alpha+k}}{((\beta-s)/\alpha+1)_k}\zeta_f(s/\alpha;\delta);\omega\Bigg)
            + \Rr^k(t),
    \]
    where $\zeta_f$ is given by Equation~\ref{eqn:zetaFormulaAlpha}. The error term satisfies $\Rr^k(t)=O(t^{\beta/\alpha-\sup(S)+k})$ as $t\to0^+$. 
\end{theorem}

Note, in particular, that the order of the remainder directly controls the error of the approximation. As we are considering small values of $t$, the larger the exponent, the better the remainder estimate. The restriction of $k\geq2$ is related to the exponent $\kappa=0$ of the languid growth of $\zeta_\Phi$, and consequently of $\zeta_f$. 

The leading order term occurs at the pole $D=\dim_S(\Phi)$, as this is the abscissa of convergence of $\zeta_f$. If $D$ is a simple pole, the only pole with real part $D$, and supposing $D>\sigma_R$, then the leading order term is explicitly 
\[ 
    \frac1{((\beta-D)/\alpha+1)_k}\Res(\zeta_f(s/\alpha;\delta);D)\,t^{(\beta-D)/\alpha+k}. 
\]
This occurs when the scaling ratios $\set{\lambda_\ph}_{\ph\in\Phi}$ are non-arithmetically related, called the  nonlattice case (see, for example, Definition~2.13 in \cite{LapvFr13_FGCD}). 

\subsection{Distributional Explicit Formulae}
%
%

Put simply, a \textit{distribution} is a functional on a prescribed space of functions called \textit{test functions}. A standard distribution is a functional on the space $C_c^\infty$ of smooth, compactly supported functions on a given domain. The type of distributions we consider here, called \textit{tempered distributions} will be a subspace of these distributions. 

We define as the space of test functions the class of Schwartz functions, or functions of rapid decrease. For a finite domain $(0,\delta)$, just as in Chapter~5 of \cite{LapvFr13_FGCD} and Chapter~5 of \cite{LRZ17_FZF}, the space is defined by 
\begin{equation}
    \Ss(0,\delta):=\left\{\testfn\in C^\infty(0,\delta)\suchthat 
    \begin{aligned}
        &\forall m\in\ZZ,\,\forall q\in\NN,\text{ as }t\to0^+,\\ 
        &t^m\testfn^{(q)}(t)\to0 \text{ and }(t-\delta)^m\testfn^{(q)}(t)\to0\,
    \end{aligned}
    \right\}.
\end{equation}
Note that $C_c^\infty(0,\delta)\subseteq\Ss(0,\delta)$, with a continuous embedding, whence by duality the space of functionals on $\Ss(0,\delta)$ is a subset of the usual space of distributions, the dual of $C_c^\infty(0,\delta)$. Typically, for a distribution $F$ and a test function $\testfn$, we write $F(\testfn)=\bracket{F,\testfn}$, denoting by $\bracket{\cdot,\cdot}$ the natural pairing of $\Ss(0,\delta)$ with its dual space. 

In this distributional setting, the restriction of $k\geq 2$ may be relaxed in Theorem~\ref{thm:pointwiseFormula}. However, equality must be interpreted \textit{in the sense of distributions}: we say that $F=G$ in the sense of distributions if for any test function $\testfn\in\Ss(0,\delta)$, $\bracket{F,\testfn}=\bracket{G,\testfn}$. In addition to the standard definitions for the definitions of differentiation and integration of distributions, we note the following relevant identity regarding residues at a point:
\begin{equation}
    \bracket{\Res(t^\beta G(s);\omega),\testfn} = \Res(\Mellin[\testfn](\beta+1)G(s);\omega).
\end{equation}
See Equation~5.2.5 in \cite{LRZ17_FZF}, as well as the preceding discussion regarding the extension of $\testfn$ to all of $\RR^+$. 

Secondly, we provide the definition for a distribution $\Rr$ to satisfy the error estimate $O(t^\alpha)$ as $t\to0^+$. Given a test function $\testfn\in\Ss(0,\delta)$ and $a>0$, define the new test function $\testfn_a(t):= a^{-1}\testfn(t/a)$. We say that $\Rr(t)=O(t^\alpha)$ as $t\to0^+$ if for all $a>0$ and for all $\testfn\in\Ss(0,\delta)$,
\begin{equation}
    \label{eqn:distRemEst}
    \bracket{\Rr,\testfn_a(t)} = O(a^{\alpha}),
\end{equation}
as $t\to0^+$ in the usual sense. 

\begin{theorem}[Distributional Explicit Formula]
    \label{thm:distFormula}
    Let $\Phi$ be a self-similar system, $\alpha>0$, $\beta\in\RR$, and $f(t)=t^{-\beta/\alpha}F(t)\in C^0(\RR^+)$. Suppose that $\Phi$ induces the scaling functional equation $f=L_\Phi^\alpha[f]+R$ on $[0,\delta]$ with admissible remainder term $R$ (in the sense of Definition~\ref{def:admissibleRem}) with corresponding screen $S$. Let $\sigma_R$ denote the abscissa of holomorphicity of $\zeta_R$, $D=\dim_S(\Phi)$ the (upper) similarity dimension of $\Phi$, and suppose $\beta/\alpha+1>\max(D/\alpha,\sigma_R)$. Lastly, suppose that $S$ is contained in the half-plane $\HH_{\sigma_R}$ and write $W_S$ for the window to the right of $S$.
    \medskip

    Then for any $k\in\ZZ$, we have that, in the sense of distributions, $F^{[k]}$ satisfies
    \begin{equation}
        \label{eqn:distributionalIdentity}
        \begin{split}
            F^{[k]}(t) &= \sum_{\omega\in\Dd_\Phi(\alpha W_S)} 
                \Res\Bigg(\cfrac{t^{(\beta-s)/\alpha+k}}{((\beta-s)/\alpha+1)_k}\zeta_f(s/\alpha;\delta);\omega\Bigg)
                + \Rr^{[k]}(t),
        \end{split} 
    \end{equation}
    as $t\to0^+$. See Equation~\ref{eqn:bracketIdentity} for the explicit identity of action on test functions. 
    \medskip

    Here, $\zeta_f$ is as in Corollary~\ref{cor:structureOfPoles}. The distributional remainder term satisfies the estimate $\Rr(t)=O(t^{\beta/\alpha-\sup(S)+k})$ as $t\to0^+$, in the sense of Equation~\ref{eqn:distRemEst}.
\end{theorem}
Most precisely, Equation~\ref{eqn:distributionalIdentity} means that for any $\testfn\in\Ss(0,\delta)$, 
\begin{equation}
    \label{eqn:bracketIdentity}
    \begin{split}
        \bracket{F^{[k]},\testfn} &= \sum_{\omega\in\Dd_\Phi(\alpha W_S)} 
        \Res\Bigg(\cfrac{\Mellin[\testfn]({(\beta-s)/\alpha+k+1})}{((\beta-s)/\alpha+1)_k}\zeta_f(s/\alpha);\omega\Bigg)
        + \bracket{\Rr^{[k]},\testfn}.
    \end{split}
\end{equation}
While this formulation is less direct compared to the pointwise expansion, it does have the advantage of requiring less regularity to leverage the expansion. Namely, it is valid when $k=0$, yielding a formula (in the sense of distributions) for the function $f$ itself. 
\begin{corollary}[Distributional Explicit Formula, $k=0$]
    \label{cor:distFormulaZero}
    Let $\Phi$ be a self-similar system, $\alpha>0$, $\beta\in\RR$, and $f(t)=t^{-\beta/\alpha}F(t)\in C^0(\RR^+)$. Suppose that $\Phi$ induces the scaling functional equation $f=L_\Phi^\alpha[f]+R$ on $[0,\delta]$ with admissible remainder term $R$ (in the sense of Definition~\ref{def:admissibleRem}) with corresponding screen $S$. Let $\sigma_R$ denote the abscissa of holomorphicity of $\zeta_R$, $D=\dim_S(\Phi)$ the (upper) similarity dimension of $\Phi$, and suppose $\beta/\alpha+1>\max(D/\alpha,\sigma_R)$. Lastly, suppose that $S$ is contained in the half-plane $\HH_{\sigma_R}$ and write $W_S$ for the window to the right of $S$.
    \medskip

    Then in the sense of distributions we have that for $t\in[0,\delta)$,
    \begin{equation}
        \label{eqn:distributionalIdentityZero}
            F(t) = \sum_{\omega\in\Dd_\Phi(\alpha W_S)} 
                \Res\big(t^{{\beta-s}/\alpha}\zeta_f(s/\alpha;\delta);\omega\big) + \Rr(t).
    \end{equation}
    The distributional remainder term satisfies the estimate $\Rr(t)=O(t^{\beta/\alpha-\sup(S)})$ as $t\to0^+$, in the sense of Equation~\ref{eqn:distRemEst}.
    \medskip

    If we further assume that the poles of $\zeta_\Phi$ in $W_S$ are simple, then Equation~\ref{eqn:distributionalIdentityZero} simplifies to 
    \[
        f(t) = \sum_{\omega\in\Dd_\Phi(\alpha W_S)} 
            \Res(\zeta_f(s/\alpha;\delta);\omega)\,t^{(\beta-\omega)/\alpha} + \Rr(t).
    \]
\end{corollary}
In this formulation, and in particular in the case when the self-similar system induces simple poles of its zeta function, the formula is a simple expansion with constant coefficients and powers determined by the poles $\omega$. The computation of the residues, however, is not a trivial matter and beyond the current scope of this work. Corollary~\ref{cor:distFormulaZero} is a direct application of Theorem~\ref{thm:distFormula}. 

\subsection{Languidity from Scaling Functional Equations}
%
%

Before we can prove the explicit formulae, we shall have need of estimating $\zeta_f(s;\delta)=\Mellin^\delta[f](s)$. In particular, using the formula in Corollary~\ref{cor:structureOfPoles}, we show that $\zeta_f$ is languid provided some control over the remainder term itself. 
\begin{theorem}[Languidity of Zeta Functions from SFEs]
    \label{thm:languidityFromSFEs}
    Let $\Phi$ be a self-similar system, let $\alpha>0$, and $f\in C^0(\RR^+)$. Suppose that $\Phi$ induces the scaling functional equation $f=L_\Phi^\alpha[f]+R$ on $[0,\delta]$ with admissible remainder term $R$ for $L_\Phi^\alpha$ (in the sense of Definition~\ref{def:admissibleRem}) and let $S$ be an admissible screen.
    \medskip 
    
    Then the zeta function $\zeta_f(s;\delta)=\Mellin^\delta[f](s)$ is also languid with exponent $\kappa=0$ with respect to $S$. 
\end{theorem}
\begin{proof}
    Because $R$ is an admissible remainder for $L_\Phi^\alpha$ with admissible screen $S$, we have that $\zeta_\Phi(\alpha s)$ and $\zeta_R(s;\delta)$ are jointly languid (with exponent $\kappa=0$) with respect to this screen $S$. We have that $\xilfphia$ is an entire function by Corollary~\ref{cor:structureOfPoles} and explicitly we show that $\xilfphia$ is bounded on any vertical strip of the form $\HH_a^b=\set{s\in\CC\suchthat a<\Re(s)<b}$.
    
    In order to do so, it suffices to find constants $C_\xi,A_\xi$ such that $|\xilfphia(s;\delta)|\leq C_\xi A_\xi^{|\Re(s)|}$, as this function is bounded when $|\Re(s)|$ is bounded as is the case with vertical strips. 
    
    To that end, let $\sigma=\Re(s)$ and define $\Lambda_+=\max(\set{\lambda_\ph^\alpha}_{\ph\in\Phi})$ and $\Lambda_-=\min(\set{\lambda_\ph^\alpha}_{\ph\in\Phi})$. Then we have that 
    \begin{align*}
        |\xilfphia(s;\delta)| 
            &\leq \sum_{\ph\in\Phi} \lambda_\ph^{\alpha\sigma}
                \left|\Mellin_\delta^{\delta/\lambda_\ph^\alpha}[f](s)\right| \\
            &\leq \max(\Lambda_+^\sigma,\Lambda_-^\sigma) \sum_{\ph\in\Phi}
                \left|\Mellin_\delta^{\delta/\lambda_\ph^\alpha}[f](s)\right|.
    \end{align*}
    By Corollary~\ref{cor:MellinBounds}, we have that each Mellin transform $\Mellin_\delta^{\delta/\lambda_\ph^\alpha}[f]$ is bounded on any vertical strip of the form $\HH_{a}^{b}$. It follows that a finite sum of such functions is bounded on any such vertical strip by the sum of these constants, which we denote by $C_\xi$. Choosing $A_\xi=\max(\Lambda_-,\Lambda_+,\Lambda_-^{-1},\Lambda_+^{-1})$ allows us to write that $A_\xi^{|\sigma|}\geq \max(\Lambda_+^\sigma,\Lambda_-^\sigma)$, and thus obtaining the desired bound. It follows that $|\xilfphia|$ is is bounded in any vertical strip.

    Now let $\set{\tau_n}_{n\in\ZZ}$ be a sequence of admissible heights shared jointly by $\zeta_\Phi(\alpha s)$ and $\zeta_R(s;\delta)$, which exists by fiat. We have that each of these respective functions is uniformly bounded on all of the horizontal contours $I_n=[S(\tau)+i\tau_n,c+i\tau_n]$, for some particular $c$ sufficiently large so that $\zeta_\Phi(\alpha s)$ and $\zeta_R(s;\delta)$ are holomorphic in $\HH_c$. Since $\xilfphia$ is bounded on the strip $\HH_{\inf(S)}^c$, it follows that it is uniformly bounded on all of these intervals $I_n\subset\HH_{\inf(S)}^c$. It follows that $\zeta_f(s;\delta)=\zeta_\Phi(\alpha s)(\xilfphi(s;\delta)+\zeta_R(s;\delta))$ is uniformly bounded on these intervals. This establishes hypothesis \textbf{L1} for the screen $S$. 

    For hypothesis \textbf{L2}, we note that $\zeta_R(s;\delta)$ and $\zeta_\Phi(\alpha s)$ are bounded on $S$ by assumption of their joint languidity. Because $S$ is contained in the strip $\HH_{\inf(S)}^{\sup(S)}$ (noting that screens are by definition bounded), we have that $\xilfphia$ is bounded on $S$. The boundedness of $\zeta_f$ on this screen thus follows. Thus, $\zeta_f$ is languid with respect to the same screen $S$ (with exponent $\kappa$). 
    
    

\end{proof}

\subsection{Proof of the Explicit Formulae}
%
%

With the languid growth conditions of $\zeta_f$ having been established, we may now complete the proofs of the main theorems, Theorems~\ref{thm:pointwiseFormula} and \ref{thm:distFormula}. We shall treat the proofs of these theorems together as they require similar setup and estimates. The main differences come in the application of different theorems to provide the relevant result. Starting from the result of Theorem~\ref{thm:MellinInversion}, we will compute the resulting contour integral by the methods in Theorem~5.1.11 and Theorem~5.2.2 of \cite{LRZ17_FZF} for the pointwise and distributional explicit formula, respectively. 

\begin{proof}
    Fix $c$ such that $\beta/\alpha+1>c>\max(D/\alpha,\sigma_R)$. By Theorem~\ref{thm:MellinInversion}, we have that for any $t\in(0,\delta)$,
    \begin{align*}
        f(t) &= \frac{1}{2\pi i} \int_{c-i\infty}^{c+i\infty} 
            t^{-s} \zeta_\Phi(\alpha s)(\xilfphia(s;\delta)+\zeta_R(s;\delta))\,ds,
    \end{align*}
    where $\zeta_f$, $\xilfphia$, and $\zeta_R$ are as in Corollary~\ref{cor:structureOfPoles}. By Theorem~\ref{thm:languidityFromSFEs}, we have that $\zeta_f$ is languid with exponent $\kappa=0$ on the admissible screen $S$.

    Noting that $F(t)=t^{\beta/\alpha}f(t)$, we have that 
    \begin{align*}
        F(t) &= \frac{1}{2\pi i} \int_{c-i\infty}^{c+i\infty} 
            t^{\beta/\alpha-s} \zeta_\Phi(\alpha s)(\xilfphia(s;\delta)+\zeta_R(s;\delta))\,ds.
    \end{align*}
    It is this expression that is analogous to Equation~5.1.18 in \cite{LRZ17_FZF}, but with $\beta$ in place of $\dimension$, $\alpha=1$, and $\zeta_f(s;\delta)$ in place of the relative tube zeta function $\tubezeta_{A,\Omega}(s;\delta)$. Under these conditions, Theorems~5.1.11 and Theorem~5.2.2 of \cite{LRZ17_FZF} are essentially applicable with only minor modifications and justifications required in the general case. 
    
    Firstly, Proposition~5.1.8 of \cite{LRZ17_FZF} is applicable, where the interchange of the order of integration may justified by the explicit formula for $\zeta_f$ (c.f. Corollary~\ref{cor:structureOfPoles}) and explicit bounds for the constituent functions. The uniform estimates for $\zeta_R$ and $\xilfphia$ follow from Corollary~\ref{cor:MellinBounds}. 
    
    For an estimate of $\zeta_\Phi$, note that when $\sigma=\Re(s)>D/\alpha$,
    \begin{align*}
        \left| \sum_{\ph\in\Phi} \lambda_\ph^{\alpha s}  \right| \leq p(\alpha \sigma):=\sum_{\ph\in\Phi} \lambda_\ph^{\alpha\sigma} < p(D) =1,
    \end{align*}
    using the technique in Proposition~\ref{prop:simDimExistence} (viz. $p$ is strictly decreasing). It follows that there is a uniform bound (based on $c>D/\alpha$) for $\zeta_\Phi(\alpha s)$ when $\Re(s)=c$. Put together, we obtain that the function $t^{\beta/\alpha-s}\zeta_f(s;\delta)$ is integrable for $(t,s)\in [0,\delta]\times (c-i\infty,c+i\infty)$ provided that $\beta/\alpha+1>c>D/\alpha$. The assumption that $\beta/\alpha+1>c$ ensures that $t^{\beta/\alpha-s}$ is integrable as $t\to0^+$. Next, Lemma~5.1.10 (with $\kappa=0$) is directly applicable based on the languidity of the function $\zeta_f$ (see Theorem~\ref{thm:languidityFromSFEs}) followed by Theorem~5.1.11 of \cite{LRZ17_FZF}. This establishes proving Theorem~\ref{thm:pointwiseFormula}. 

    This integrability argument also is required to justify the use of the Fubini--Tonelli theorem in Equation~5.2.10 (for the proof of Theorem~5.5.5) in \cite{LRZ17_FZF}. In this case, because $\testfn\in\Ss(0,\delta)$, we may write the corresponding integral in question as  
    \[
        \int_{c-i\infty}^{c+i\infty}\int_0^\infty 
            \testfn(t)\frac{t^{\beta/\alpha-s+k}}{(\beta/\alpha-s+1)_k}\zeta_f(s;\delta)  \,dt 
        = \int_{c-i\infty}^{c+i\infty}\int_0^\delta 
            \testfn(t)\frac{t^{\beta/\alpha-s+k}}{(\beta/\alpha-s+1)_k}\zeta_f(s;\delta)  \,dt.
    \]

    We claim that once again, the integrand is integrable over the product space, that is for $(t,s)\in[0,\delta]\times(c-i\infty,c+i\infty)$. Firstly, we have the same estimates for $\zeta_f$ as before. For the Pochhammer symbol, we may write that 
    \begin{align*}
        |(\beta/\alpha-s+1)_k| 
            &\geq |\beta/\alpha-c+1|\cdot |\beta/\alpha-c+2|\cdots |\beta/\alpha-c+k-1| \\
            &\geq |\beta/\alpha-c+1|^k,
    \end{align*}
    noting that on the vertical lines $z=\beta/\alpha-c+k+i\RR$, the closest point to the origin occurs exactly on the real axis. The second estimate follows from the fact that $\beta/\alpha-c+1>0$ is the smallest term in the product. This gives a uniform upper bound for $|(\beta/\alpha-s+1)_k^{-1}|$. Lastly, the power of $t$ may be arbitrary since $\testfn$ is a function of rapid decrease: for any $\gamma\in\RR$, $t^\gamma\testfn(t)\to0$ as $t\to0^+$.  
    
    The final step is a minor change of variables. Rather than indexing the sum by the poles of $\zeta_f$ in $W_S$, we note that as long as $\Re(s)>\sigma_R$, the abscissa of holomorphicity of $\zeta_R$, a pole of $\zeta_f$ only occurs at $s$ if $\omega=\alpha s$ is a pole of $\zeta_\Phi$. In other words, we have that 
    \[
        \sum_{s\in\Dd_f(\HH_{\sigma_R})} g(s) = \sum_{\omega/\alpha\in\Dd_f(\HH_{\sigma_R})} g(\omega/\alpha) = \sum_{\omega\in\Dd_\Phi(\HH_{\alpha \sigma_R})} g(\omega/\alpha).
    \]
    Since we assume that $S$ lies in the half-plane $\HH_{\sigma_R}$, this same argument applies with respect to the window $W_S$ and its scaled transformation $\alpha W_S$.
    
\end{proof}

\chapter{Geometry of Fractals}
\label{chap:geometry}
%
%

In this chapter, we apply our results in the general framework to establish explicit formulae for \textit{relative tube functions}. These functions describe the volume of a tubular neighborhood, obtained by thickening up a given set, as a function of the distance parameter describing the thickness. Strictly speaking, we concern ourselves with relative neighborhoods: after thickening up the set, we then intersect with some other set to which the neighborhood is considered relative. Given a \textit{relative fractal drum} $(X,\Omega)$ in $\RR^\dimension$, the relative tube function is $|X_t\cap\Omega|_\dimension$, where $|\cdot|_\dimension$ is the $\dimension$-dimensional Lebesgue measure. 

A key notion in the study of fractals is that of \textit{complex dimensions}. They may be defined explicitly starting from these tube functions (which produce tube zeta functions), our main object of study in this chapter. As we will see, the complex dimensions play a pivotal role in producing explicit formulae for the tube functions. We show that the possible complex dimensions of a self-similar set $X$, the invariant set of a self-similar system $\Phi$, are controlled exactly by the poles of the scaling zeta function $\zeta_\Phi$. This requires some appropriate separation conditions, including the open set condition. 


The results in this work determine the possible complex dimensions for a wide variety of self-similar fractals in dimensions two and larger. Most notably, our method allows the study of self-similar fractals which are not made up of a union of disjoint sets. Instead, the tubular neighborhoods may partition into scaled pieces only up to a given remainder term. We show that this induces a scaling functional equation for the normalized tube function, leading to the application of our results from Chapter~\ref{chap:SFEs}. We conclude with the explicit computation of the tube functions and possible complex dimensions for generalized von Koch fractals. 


\section{Relative Tube Functions}
%
%

In what follows, let $(X,\Omega)$ be a relative fractal drum in $\RR^\dimension$. We recall that per Definition~\ref{def:RFD}, this is a tuple of sets $X,\Omega\subset\RR^d$ with the condition that $\Omega$ is open, has finite Lebesgue measure, and that there exists some $\delta>0$ so that 
\[ \Omega_\delta = \set{x\in\RR^\dimension \suchthat \exists\,y\in\Omega,\,|x-y|<\delta} \supseteq X. \]
Let $|\cdot|_\dimension$ denote the $\dimension$-dimensional Lebesgue measure. A set of the form $\Omega_\delta$ is called a \textit{tubular neighborhood} of $\Omega$. If we consider a tubular neighborhood $X_t$ of the set $X$, a \textit{relative tubular neighborhood} is the set $X_t\cap\Omega$. We will concern ourselves with the following problem: given a relative fractal drum $(X,\Omega)$, find a formula for the volume of the relative tubular neighborhood $X_t\cap\Omega$. Note that we are using volume as a general term for $\dimension$-dimensional measure; it is understood to be area in $\RR^2$ and length in $\RR$. 

\subsection{Tubular Neighborhoods and Tube Functions}
Given a set $X\subset\RR^\dimension$ and $t\geq0$, we say that a \index{Tubular neighborhood}\textbf{tubular neighborhood} of $X$ is any set of the form 
\[ X_t := \{ y\in\RR^\dimension : \exists\,x\in X,\,|y-x|<t \}. \]
Typically, we care about these sets for small, positive values of $t=\e>0$. For this reason, tubular neighborhoods may also be called $\e$-neighborhoods, especially when one wishes to explicitly specify the value of the parameter. If $t=0$, then a $0$-neighborhood is simply the set itself. If the set $X$ is bounded, then it lies within some ball of radius $R$ about the origin. For large $t\gg R$, the $t$-neighborhood of $X$ will asymptotically behave simply like the neighborhood of this ball $B_R(0)$. So, the most interesting behavior occurs for small values of $t>0$. Generally speaking, we will fix $\delta>0$ and let $t\in(0,\delta)$, considering the limit as $t\to0^+$.  

Note that even if $X$ itself is not Lebesgue measurable, its tubular neighborhoods $X_t$ for any $t>0$ are automatically measurable. Note that the following is an equivalent definition for a tubular neighborhood:
\[ X_t = \{ y\in\RR^\dimension : \exists\,x\in X,\,|y-x|<t \} = \bigcup_{x\in X}B_t(x). \]
The latter set is an arbitrary union of open sets $B_t(x)$, the open ball of radius $t$ centered at $X$, and thus is also open. By the construction of Lebesgue measure, open sets are all measurable. 

Given a set $X\subset\RR^\dimension$, the function 
\[ V_X(t) := \measure{X_t}, \]
defined for any $t>0$, is called the \index{Tube function}\textbf{tube function} of $X$. Note that by the properties of Lebesgue measure, it follows that $V_X$ is continuous on $(0,\infty)$ and decreasing. If $X$ is bounded, we have that $V_X(t)$ is finite and strictly decreasing for any $t>0$. If $X$ is measurable, then $V_X(0)$ is well-defined. 

\subsection{Tube Formulae and Codimension}
\label{sub:tubeCodimension}
A formula for the tube function of a set $X$ is naturally connected to the dimensions of the set $X$ and its constituent pieces. Consider for a moment some simple shapes with clear dimension: a point $P=\set{(0,0)}$, a unit line segment $L=[0,1]$, a unit square $S=[0,1]\times[0,1]$ in $\RR^2$. Figure~\ref{fig:threeShapes} depicts the tubular neighborhoods of these shapes, which can be computed explicitly:
\begin{align}
    \label{eqn:simpleTubeFormulae}
    \begin{split}
        V_P(t)   &= \pi t^{2-0}, \\
        V_L(t)   &= \pi t^{2-0}+2t^{2-1}, \\
        V_S(t)   &= \pi t^{2-0}+4t^{2-1}+t^{2-2}.
    \end{split}
\end{align}
These formulae have been written in a particular fashion: each term is of the form $t^{\dimension-k}$, where $k$ is the dimension corresponding to a piece of the shape contributing to the neighborhood and $\dimension=2$. The point has a neighborhood from only a $0$-dimensional piece. The line segment has contributions from its endpoints (forming a full circle), and contributions from the segment itself. The square has contributions from its corners, from its edges, and from the area of the square itself. In other words, the exponents of the distance parameter $t$ in the tube formulae are the \textit{codimensions} of components of the set. 

\begin{figure}[t]
    \centering 
    \includegraphics[width=10cm]{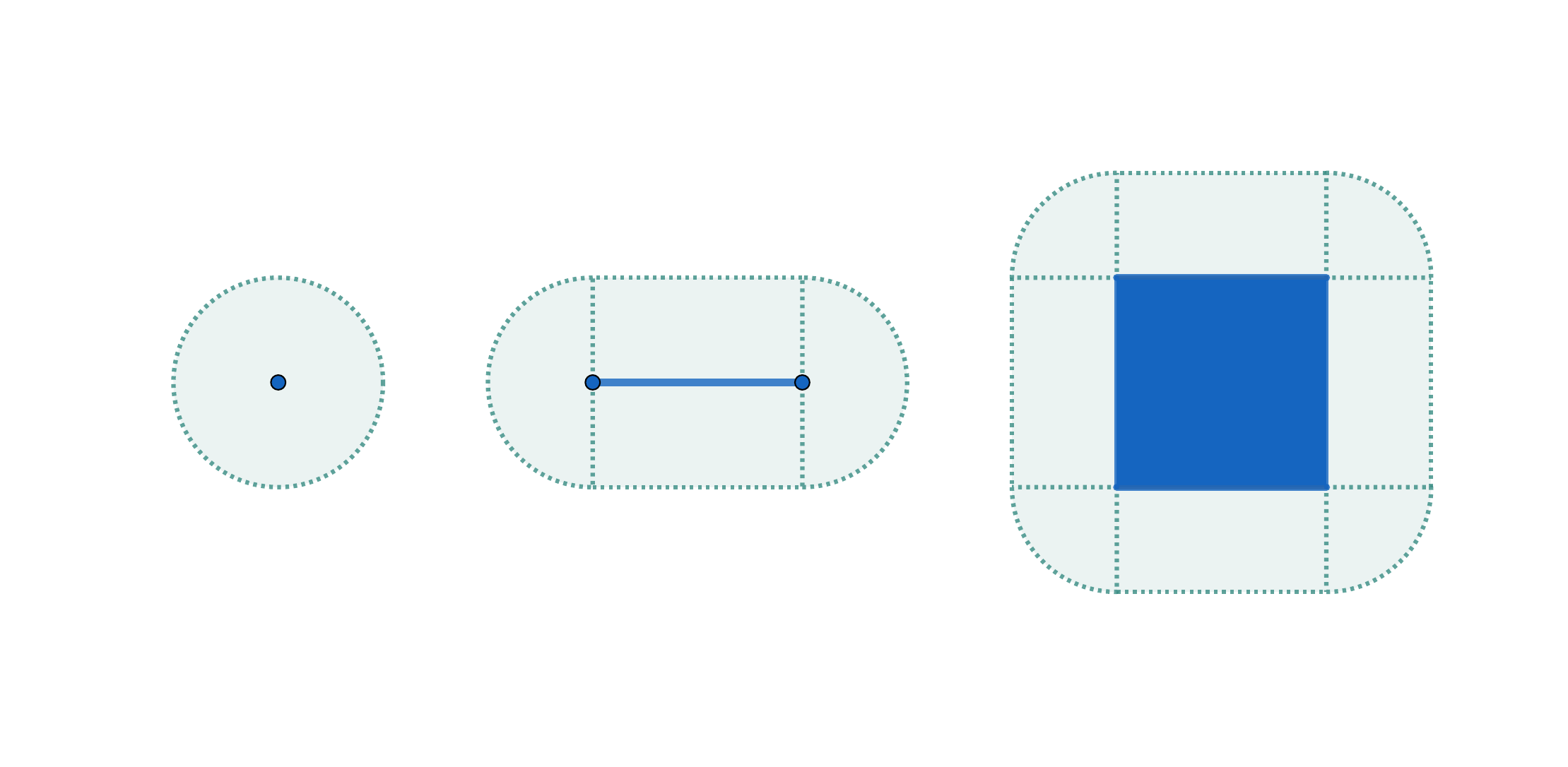}
    \caption[Tubular neighborhoods of some elementary planar shapes]{Tubular neighborhoods of a point, a line segment, and a unit square in $\RR^2$. The volume of these neighborhoods, given by the tube formulae in Equation~\ref{eqn:simpleTubeFormulae}, contain contributions of the form $t^{2-k}$ corresponding to pieces of the neighborhood arising from $k$-dimensional component of the set. Partition lines within the neighborhood have been added to delineate these pieces.}
    \label{fig:threeShapes}
\end{figure}

Now instead consider a fractal $X\subset\RR^\dimension$. If we obtain a tube formula for $\measure{X_t}$, it may contain terms of the form $t^{\dimension-D}$, where $D\in[0,\dimension]$ need not be an integer. In fact, if the formula contains periodic terms---corresponding to \textit{geometric oscillations}--- we can take the Fourier expansion of these periodic terms and the formula can contain terms of the form $t^{\dimension-\omega}$ where $\omega$ is a \textit{complex} number. Writing $\omega=\sigma+i\tau$, a single term can be rewritten as 
\[ 
    t^{\dimension-\omega} = t^{\dimension-\sigma}e^{-i\tau\log t} 
        = t^{\dimension-\sigma}(\cos(\tau\log t)-i\sin(\tau\log t)). 
\]
The real part of $\omega$ controls the amplitude of this term and the imaginary part controls the multiplicative period. For the latter, note that when $t'=e^{2\pi/\tau}t$, $\tau\log t'=\tau\log t+2\pi$ which is a shift by the period of the sine and cosine. 

It is through this means that we can make sense of the \textit{complex dimensions} of fractals. It also explains why a set will typically have multiple complex dimensions. We will take the approach of defining fractal zeta functions through \textit{relative tube zeta functions} which rely on relative tube functions. However, the complex dimensions of a set may be obtained by different types of zeta functions, yielding the exact same set of complex dimensions. For instance, a \textit{distance zeta function} has the same poles as a tube zeta function; see Corollary~2.2.10 in \cite{LRZ17_FZF}.

\subsection{Relative Tube Functions}
\label{sub:tubeFunctions}

Our main object of study is a relative fractal drum $(X,\Omega)$ in $\RR^\dimension$. Before computing the volume of a tubular neighborhood of $X$, we will first intersect with the relative set $\Omega$. The resulting set $X_t\cap\Omega$ is called a \index{Tubular neighborhood!Relative tubular neighborhood}\textbf{relative tubular neighborhood} of $X$. The condition that $\Omega$ is measurable with finite Lebesgue measure (see Definition~\ref{def:RFD}) implies that for any $t>0$, $X_t\cap\Omega$ is measurable and that $\measure{X_t\cap\Omega}<\infty$. 

Note that this allows for $X$ to be an unbounded set, so long as there is a finite set $\Omega$ such that some tubular neighborhood of $\Omega$ contains $X$. For example, take $X$ to be Gabriel's horn, the surface of revolution in $\RR^3$ formed by rotating the curve $y=1/x$, $x\in(0,\infty)$, about the $x$-axis. It is well-known that $X$ has infinite surface area but finite volume. Taking $\Omega$ to be the open region bounded by $X$ and the plane $x=0$, $(X,\Omega)$ is a relative (fractal) drum. Note that $X\subset\partial\Omega$, so $X$ is contained in any tubular neighborhood of $\Omega$. 

The tube function of $(X,\Omega)$ shall be the volume of a relative tubular neighborhood of $X$ as a function of the distance parameter. For this reason, we call them \textit{relative} tube functions to distinguish from the tube function of $X$ itself. We note that the definition works even if $(X,\Omega)$ is not a relative fractal drum provided that $X_t\cap\Omega$ is measurable with finite measure. 
\begin{definition}[Relative Tube Function]
    \label{def:tubeFunction}
    Let $X,\Omega\subset\RR^N$ and suppose that $\Omega$ is a Lebesgue measurable set. Suppose either that $\Omega$ is of finite measure or that $X$ is bounded. The \index{Tube function!Relative tube function}\textbf{relative tube function} $V_{X,\Omega}$ is the function defined by $V_{X,\Omega}(t)=\measure{X_t\cap\Omega}$ for any $t>0$ and at $t=0$ if $X$ is measurable. 
\end{definition}
We say that $V_{X,\Omega}$ is the tube function of $X$ relative to $\Omega$, and if $(X,\Omega)$ is a relative fractal drum, then it is the tube function of $(X,\Omega)$. In a relative fractal drum, $\Omega$ is assumed to be open (hence measurable) and having finite Lebesgue measure, so this function is well-defined on $\RR^+$. We do not need to assume that $X$ is measurable (as for $t>0$, its tubular neighborhoods are open). Note that if we choose $\Omega=\RR^\dimension$, then we obtain a standard tube function for $X$, which is finite when $X$ is bounded. 

\subsection{Scaling Properties}
The Lebesgue measure has the scaling property that, for any $\lambda\in\RR^+$, $\measure{\lambda X} = \lambda^\dimension \measure{X}$. When a tubular neighborhood $X_t$ is scaled, the parameter $t$ defining the new scaled neighborhood will scale linearly. In other words, $\lambda\cdot X_t=(\lambda X)_{\lambda t}$. Combining these properties, we may deduce how ordinary tube functions change under scaling: for any $\lambda>0$, $V_{\lambda X}(\lambda t) = \lambda^\dimension V_{X}(t)$. 

In general, we wish to consider the behavior of a set under the image of a similitude $\ph:\RR^\dimension\to\RR^\dimension$. But such maps are exactly compositions of isometries, which preserve the measure, and uniform scaling maps of the form $x\mapsto\lambda x$. When considering a relative fractal drum, we will also be scaling the corresponding relative set $\Omega$. 
\begin{lemma}[Scaling Property of Relative Tube Functions] 
    \label{lem:volumeScaling}
    \index{Scaling law!of relative tube functions}
    Let $(X,\Omega)$ be a relative fractal drum in $\RR^\dimension$. For any $t>0$ and $\lambda>0$, we have that
    \[ V_{\lambda X,\lambda\Omega}(t) = \lambda^\dimension V_{X,\Omega}(t/\lambda). \]
    More generally, if $\ph:\RR^\dimension\to\RR^\dimension$ is a similitude with scaling ratio $\lambda$, then we have the analogous identity:
    \begin{equation}
        \label{eqn:volumeScaling}
        V_{\ph(X),\ph(\Omega)}(t) = \lambda^\dimension V_{X,\Omega}(t/\lambda).
    \end{equation}
\end{lemma}
Note that this scaling property is Lemma~3.3 of \cite{Hof25}. We reproduce the proof with slight tweaks to the notation for consistency with our present conventions.
\begin{proof}
    Let $\ph=\psi_1\circ\ph_\lambda\circ\psi_2$, where $\psi_2,\psi_1$ are isometries (equivalently, compositions of rotations, reflections, and translations) and where $\ph_\lambda$ is a uniform scaling transformation, viz. $\ph_\lambda(x)=\lambda x$. Any similitude of scaling ratio $\lambda$, by definition, may be written this way. Note that the first identity is a special case of the second property; merely set $\psi_1=\psi_2=\text{Id}$. For convenience, let us denote $U'=\psi_2(U)$ in what follows, and note that $\lambda U = \ph_\lambda(U)$.
    
    The maps in question have lots of nice properties. First, we note that $\psi_2$ and $\psi_1$ are isometries, which is to say that $\measure{U'}=\measure{U}$ and $\measure{\psi_1(U)}=\measure{U}$ for any $U\subset\RR^\dimension$. Since $\psi_1,\psi_2,\ph_\lambda,$ and their composition $\ph$ are injective, they each have the property that the intersection of the images of two sets is the image of the intersection of those sets.
    The same property is true (unconditionally) for the union of images being the image of the union. 

    This lattermost fact is one way to see that $\psi_1(X_t)=(\psi_1(X))_t$. One may write the neighborhood $X_t$ as a union of $t$-balls, and then apply the preservation of unions under images to see that this set is exactly the neighborhood of the transformed set $\psi_1(X)$. 
    
    Using this fact and the above properties, we find that the tube functions are unaffected by an injective isometry such as $\psi_1$:
    \begin{align*}
        V_{\psi_1(\lambda X'),\psi_1(\lambda \Omega')}(t) 
            &= \measure{\psi_1((\lambda X')_t)\cap\psi_1(\lambda\Omega')} \\
            &= \measure{\psi_1[(\lambda X')_t\cap\lambda\Omega']}
                =V_{\lambda X',\lambda\Omega'}(t).
    \end{align*}
    As mentioned before, we have that $\ph_\lambda(X_t)=(\lambda X)_{\lambda t}$. Writing $t=\lambda t$, we obtain that:
    \begin{align}
        \label{eqn:scalingFormula}
        V_{\lambda X',\lambda\Omega'}(\lambda t) = \measure{\ph_\lambda(X'_t\cap\Omega')} =\lambda^N V_{X',\Omega'}(t),
    \end{align}
    where the last step is the scaling property of the Lebesgue measure. We have already shown that for an injective isometry like $\psi_2$, $V_{\psi_2(X),\psi_2(\Omega)}(t)=V_{X,\Omega}(t)$. Thus, the result follows by replacing $t$ in Equation~\ref{eqn:scalingFormula} by $t/\lambda$ and using this isometry invariance. 
\end{proof}

This scaling property leads to a choice of normalizing factor $t^{-\dimension}$. Since $t^{-\dimension}\lambda^\dimension=(t/\lambda)^{-\dimension}$, if we multiply both sides of Equation~\ref{eqn:volumeScaling} by this factor, then the normalized function $t^{-\dimension}V_{X,\Omega}(t)$ satisfies a homogenized scaling law. 
\begin{corollary}[Volume Scaling Law]
    \label{cor:volumeScalingLaw}
    \index{Scaling law!of relative tube functions}
    Let $V_{X,\Omega}$ be the (relative) tube function of an RFD $(X,\Omega)$ in $\RR^\dimension$. Then the normalized (relative) tube function, $ f(t;(X,\Omega)) := t^{-\dimension}V_{X,\Omega}(t)$, satisfies the $1$-scaling law 
    \[ f(t;(\ph[X],\ph[\Omega])) = f(t/\lambda;(X,\Omega))  \]
    in the sense of Definition~\ref{def:scalingLaw}.
\end{corollary}
\section{Tube Zeta Functions and Complex Dimensions}
%
%

Tube zeta functions were introduced in \cite{LRZ17_FZF} to study fractals in higher dimensions and to extend the theory of \textit{complex dimensions}, introduced in \cite{LapvFr13_FGCD} for the theory of one dimensional fractal harps, to these fractals. A (relative) tube zeta function is constructed from a (relative) tube function, which is the volume of tubular neighborhoods of the given set (relative to another) as a function of the parameter. To be precise, a (relative) tube zeta function is a truncated Mellin transforms of the corresponding tube function. Consequently, these zeta functions possess important scaling properties (c.f. Lemma~\ref{lem:zetaScaling}) which will be important for studying self-similar fractals which are related to scaled copies of itself. 

The complex dimensions of a set correspond to the singularities of its associated fractal zeta function. We note that there are other types of fractal zeta functions, such as distance zeta functions also introduced in \cite{LRZ17_FZF}. However, the key feature of a fractal zeta function is the collection of its singularities (and the residues of the function at these values), and it turns out that these two types of zeta functions are equivalent in the sense of having identical poles (with corresponding multiplicities); see Corollary~2.2.10 of \cite{LRZ17_FZF}. The singularities of a fractal zeta function of a set $X$ are called the complex dimensions of $X$. In the case of relative tube functions, and their corresponding relative tube zeta functions, we say that these are the complex dimensions of $X$ relative to the given set. 

\subsection{Relative Tube Zeta Functions}

Tube zeta functions come from normalized tube functions. Note that here, we will strictly consider relative tube zeta functions, which come from relative tube functions. One may define a standard tube zeta function simply by replacement of relative tube functions with standard tube functions, or equivalently by replacement of the set $\Omega$ with the ambient space $\RR^\dimension$. Thus, for convenience we will sometimes simply call them tube zeta functions. 

Let $(X,\Omega)$ be a relative fractal drum in $\RR^\dimension$ and $V_{X,\Omega}$ its relative tube function. We recall from Corollary~\ref{cor:volumeScalingLaw} that the function $t^{-\dimension}V_{X,\Omega}(t)$ satisfies a homogenized scaling law. Fixing some $\delta>0$, we will take the truncated Mellin transform (see Definition~\ref{def:truncatedMellin}) of this normalized quantity to obtain the zeta function. We note that this construction suggests how to define zeta functions for quantities unrelated to volume as we shall see in Chapter~\ref{chap:heat} in the context of heat content. 
\begin{definition}[Relative Tube Zeta Function]
    \label{def:tubeZeta}
    \index{Tube zeta function!Relative tube zeta function}
    Let $(X,\Omega)$ be a relative fractal drum\footnote{Alternatively, let $X,\Omega\subset\RR^\dimension$ and suppose that $\Omega$ has finite Lebesgue measure.} in $\RR^\dimension$, let $V_{X,\Omega}$ its relative tube function, and let $\delta>0$. The \textbf{relative tube zeta function} $\tubezeta_{X,\Omega}$ of $X$ relative to $\Omega$ is given by
    \[ \tubezeta_{X,\Omega}(s;\delta) := \Mellin^\delta[t^{-N}V_{X,\Omega}(t)](s) = \int_0^\delta t^{s-N-1}V_{X,\Omega}(t)\,dt \]
    for $s\in\CC$ with sufficiently large real part and analytically continued elsewhere in the complex plane. 
\end{definition}
Generally speaking, the analytic continuation of a (zeta) function may encounter poles, essentially singularities, or even worse (logarithmic singularities, natural boundaries, and so forth). In this work, however, we shall only have need of meromorphic functions and thus will specialize to domains in which the zeta functions only have poles for singularities. We note that in fact $\Re(s)>\dimension$ is a sufficiently large real part, as the following shows. 
\begin{lemma}[Tube Function Asymptotics]
    \label{lem:tubeAsymptotics}
    Let $X,\Omega\subset\RR^\dimension$. Suppose that $\Omega$ is measurable and that either $\Omega$ has finite measure or $X$ is bounded. Then the relative tube function $\tubefn$ is continuous and decreasing on $\RR^+$. Consequently, $\tubezeta_{X,\Omega}(s;\delta)=\Mellin^\delta[t^{-\dimension}\tubefn(t)](s)$ is holomorphic in $\HH_{\dimension}$. 
\end{lemma}
\begin{proof}
    The continuity and monotonicity of $\tubefn$ follow from properties of the Lebesgue measure and the containment $X_t\subset X_s$ whenever $t<s$. The key element here is that $\tubefn$ is bounded as $t\to0^+$, which implies that $t^{-\dimension}\tubefn(t)=O(t^{-\dimension})$ as $t\to0^+$. This estimate not generally optimal, but is enough to ensure that tube zeta functions are well-defined. By Lemma~\ref{lem:MellinHolo}, we have that $\tubezeta_{X,\Omega}$ is holomorphic in $\HH_\dimension$.
\end{proof}

Note that the choice of $\delta>0$ is essentially irrelevant. For any two $\delta_2>\delta_1>0$, we have that 
\begin{equation}
    \label{eqn:zetaDiff}
    \tubezeta_{X,\Omega}(s;\delta_2)-\tubezeta_{X,\Omega}(s;\delta_1)
        = \int_{\delta_1}^{\delta_2}t^{s-N-1}V_{X,\Omega}(t)\,dt 
        = \Mellin_{\delta_1}^{\delta_2}[t^{-N}V_{X,\Omega}(t)](s). 
\end{equation}
By Lemma~\ref{lem:MellinHolo}, this is an entire function. Thus, these two functions must have the same set of singularities and respective residues because their difference is entire. In other words, the poles of a relative tube zeta function are not dependent on the choice of $\delta$. 

Tube zeta functions are well-suited to studying self-similar or nearly self-similar objects due to their scaling properties. The scaling property may be deduced directly from Lemma~\ref{lem:MellinScaling} and Corollary~\ref{cor:volumeScalingLaw}. 
\begin{lemma}[Scaling Property of $\tubezeta_{X,\Omega}$]
    \label{lem:zetaScaling}
    Let $(X,\Omega)$ be an RFD in $\RR^\dimension$, let $\tubezeta_{X,\Omega}$ be its relative tube zeta function, and fix $\delta>0$. Then for any similitude $\ph$ with scaling ratio $\lambda>0$ and for all $s$ in the domain of $\tubezeta_{X,\Omega}$,
    \[ 
        \tubezeta_{\phi(X),\phi(\Omega)}(s;\delta) = \lambda^s \tubezeta_{X,\Omega}(s;\delta/\lambda) 
        = \lambda^s\tubezeta_{X,\Omega}(s;\delta) + \lambda^s\Mellin_{\delta}^{\delta/\lambda}[t^{-N}V_{X,\Omega}(t)](s). 
    \] 
\end{lemma}

We will see these doubly truncated transforms often enough that we given them a name, \textit{partial tube zeta functions}, and some notation. The term partial is owing to the fact that their domain of integration does not reach zero. It is precisely the behavior near zero that matters for the development of singularities, so these functions do not detect scaling behavior in the same way that tube zeta functions do. Strictly speaking, one might call them partial \textit{relative} tube zeta functions but we will not need this distinction.
\begin{definition}[Partial (Relative) Tube Zeta Function]
    \label{def:partialTubeZeta}
    \index{Tube zeta function!Partial tube zeta function}
    Let $(X,\Omega)$ be a relative fractal drum in $\RR^\dimension$, $V_{X,\Omega}$ its relative tube function, and let $0<\delta_1<\delta_2$. A \textbf{partial (relative) tube zeta function} is the function
    \[ 
        \partialzeta_{X,\Omega}(s;\delta_1,\delta_2) 
            := \int_{\delta_1}^{\delta_2} t^{s-\dimension-1}V_{X,\Omega}(t)\,dt.
    \] 
    By Lemma~\ref{lem:MellinHolo}, it is an entire function with respect to $s$.
\end{definition}

\subsection{Complex Dimensions}

At last we may formally define complex dimensions, which lie at the heart of our analysis of self-similar fractals. First, we specify a set $W\subseteq\CC$ called a window. What we call complex dimensions are dubbed \textit{visible}, or accessible, complex dimension in the window $W$. Secondly, we note that the definition we provide is, in a sense, only a special case. Namely, we have restricted our attention to the study of meromorphic zeta functions, having only poles for singularities. Essential singularities, or even other types of singular behavior, may yet deserve to be incorporated into the general theory. However, that is beyond the scope of the present work. 

Lastly, we note that our definition of complex dimensions is for the complex dimensions of a set $X$ relative to another, $\Omega$. So, these may be called the relative complex dimensions of a given set $X$ (relative to $\Omega$) or more simply as the complex dimensions of the relative fractal drum $(X,\Omega)$ itself. We will use the terminology interchangeably. 
\begin{definition}[Complex Dimensions of a Relative Fractal Drum]
    \label{def:cDims}
    \index{Complex dimensions}
    Let $(X,\Omega)$ be a relative fractal drum and let $\tubezeta_{X,\Omega}$ be its relative tube zeta function.
    
    Let $W\subset \CC$ such that $\tubezeta_{X,\Omega}$ is meromorphic within $W$. Then the \textbf{complex dimensions} of $X$ relative to $\Omega$ contained in the window $W$ are the poles of $\tubezeta_{X,\Omega}$ contained within $W$. We denote them by $\Dd_{X,\Omega}(W)$, or $\Dd_X(W)$ when the set $\Omega$ is clear from context. 
\end{definition}

As alluded to earlier, the complex dimensions of a set control the exponents which appear in explicit formulae for tube functions. Namely, the exponents are codimensions of the form $\dimension-\omega$, where $\omega$ is in the set $\Dd_X$, and the sums are indexed by these complex dimensions. See for instance the general explicit formulae for tube functions of (generalized) fractal harps (in Chapter~5 of \cite{LapvFr13_FGCD}) or for relative tube functions (in Chapter~5 of \cite{LRZ17_FZF}). 

In the present context, we will explicitly compute the possible complex dimensions of self-similar fractals using our framework of scaling functional equations, with the results having appeared in \cite{Hof25}. In particular, Corollary~\ref{cor:structureOfPoles} will produce a formula for the relative tube zeta function and a description of its possible poles, namely being governed by the scaling zeta function $\zeta_\Phi$ for the self-similar system $\Phi$ having $X$ as its invariant set.

\section{General Results on Self-Similar Sets}
\label{sec:fractalSFEs}
%
%

In this section, we apply the results of Chapter~\ref{chap:SFEs} to study self-similar fractals. The method of establishing explicit formulae for relative tube functions was established in \cite{LRZ17_FZF}, and the approach of explicitly computing the tube zeta functions via scaling functional equations was established in \cite{Hof25}. Here, the main results lie in establishing formulae for the relevant fractal zeta functions and the possible complex dimensions of self-similar fractals. In particular, we will see that the self-similar system $\Phi$, whose invariant set is the given self-similar fractal $X$, controls the possible complex dimensions. Essentially, the possible complex dimensions may be deduced from only knowledge of the scaling ratios (and respective multiplicities) that define the self-similarity of the given fractal. 

We note that this framework requires certain separation conditions to be well-defined, namely the introduction of an \textit{osculant fractal drum}. Suppose that $(X,\Omega)$ is a fractal drum where $X$ is a self-similar set. Since an IFS with a given invariant set is not unique, the first condition (the open set condition, Definition~\ref{def:OSC}) ensures that the particular self-similar system $\Phi$ carries authentic geometric information about its invariant set $X$. The second condition is that $\Omega$ osculates $X$ in the sense that its iterates $\ph[\Omega]$ stay closest to the corresponding iterate $\ph[X]$ of $X$ itself. 

This constraint is vital for establishing an induced decomposition of the relative tube functions for the fractal drum $(X,\Omega)$ in terms of the tube functions for its corresponding rescaled drums, $\ph[(X,\Omega)] := (\ph[X],\ph[\Omega])$. Together with the scaling properties of these volume functions, we obtain an induced scaling functional equation which may be solved to obtain explicit formulae. 

\subsection{Osculant Fractal Drums}
%
%

When $X$ is a self-similar set, its tubular neighborhoods $X_t$ need not themselves be self-similar in the same sense. However, with some appropriate separation conditions imposed on a fractal drum $(X,\Omega)$, we can partition the relative tubular neighborhood for $(X,\Omega)$ in terms of relative tubular neighborhoods of the similar RFDs $(\ph[X],\ph[\Omega])$, $\ph\in\Phi$, and a residual RFD $(X,R)$, where $R=\Omega\setminus\cup_{\ph\in\Phi}\ph[\Omega]$. Ultimately, such a decomposition will lead to a scaling function equation for the relative tube function induced by $\Phi$, in the sense of Definition~\ref{def:inducedDecomp}. Ideally, only a small portion of this neighborhood should lie in the set $R$ as this is precisely what yields the remainder term. 

In order for this construction to make sense, we need the images $\ph[\Omega]$ and $\psi[\Omega]$ to be disjoint whenever $\ph\neq\psi$, where $\ph,\psi\in\Phi$, as well as the containment $\ph[\Omega]\subseteq\Omega$ for each $\ph\in\Phi$. Note that these are exactly the conditions so that $\Omega$ is a feasible set for the open set condition for $\Phi$ (see Definition~\ref{def:OSC}). This is our first condition.

Next we note a small technical point: $R$ need not be open, so strictly speaking $(X,R)$ is not an RFD. If $\partial R$ has measure zero, we may safely replace $R$ by its interior and $V_{X,R}$ shall be unaffected. Even if not, however, the definition of $R$ ensures that it is measurable with $\measure{R}\leq \measure{\Omega}<\infty$, and thus $V_{X,R}$ (and consequently $\tubezeta_{X,R}$) will still be well-defined as in Definition~\ref{def:tubeFunction} (resp. Definition~\ref{def:tubeZeta}).

Next, we will need that $X_t\cap \ph(\Omega)$ is in fact $(\ph(X))_t\cap\ph(\Omega)$ for each $\ph\in\Phi$. Equivalently, for each map $\ph\in\Phi$ and every point $y\in\ph(\Omega)$, $d(y,X)=d(y,\ph(X))$. This condition will be necessary to relate $V_{X,\ph(\Omega)}$ to $V_{\ph(X),\ph(\Omega)}$. This we call the osculating condition, introduced in \cite{Hof25}: the points in iterates $\ph(\Omega)$ of $\Omega$ ($\ph\in\Phi$) stay closest to the corresponding iterate $\ph(X)$ of $X$, rather than to a different iterate $\psi(X)$, where $\psi\in\Phi$ and $\psi\neq\ph$.

\begin{definition}[Osculating Sets and Osculant Fractal Drums]
    \label{def:oscRFD}
    \index{Osculating set}\index{Relative fractal drum (RFD)!Osculant Fractal Drum}
    Let $\Phi$ be an iterated function system on $\RR^\dimension$ and let $X$ be its attractor. A nonempty, open set $\Omega\subset\RR^\dimension$ is said to be an \textbf{osculating set} for $\Phi$ if the following hold:
    \begin{enumerate}
        \item $\Phi$ satisfies the open set condition with respect to $\Omega$;
        \item For each $\ph\in\Phi$, if $y\in\ph[\Omega]$, then $d(y,X)=d(y,\ph[X])$. 
    \end{enumerate}
    A relative fractal drum $(X,\Omega)$ is called an \textbf{osculant fractal drum} if there exists an IFS $\Phi$ for which $X$ is its attractor and for which $\Omega$ is an osculating set thereof. 
\end{definition}

In practice, it is generally easiest to ensure that a relative fractal drum is osculant in the construction of the set $\Omega$. In Section~\ref{sec:appToGKFs}, we will see that the interior of an $(n,r)$-von Koch snowflake 
(contained within a sector of angle $2\pi/n$) serves as an osculating set when $r$ is chosen to be sufficiently small. 

In the general setting, the following algorithm may be attempted to construct an osculating set. Note, however, that this may produce an empty set. Suppose that $X$ is the attractor of an iterated function system $\Phi$. For each map $\ph\in\Phi$, 
define the set
\begin{equation} 
    \label{eqn:defUphi}
    U_\ph := \set{ x\in\RR^\dimension \suchthat 
        d(x,\ph[X])<\min_{\psi\in\Phi\setminus\set{\ph}}d(x,\psi[X]) }. 
\end{equation}
By construction, $U_\ph$ consists of all the points in $\RR^\dimension$ which are closest to some point of $\ph[X]$ and strictly further from any other respective image. By construction, $U_\ph$ is open (it may be written as the preimage of a continuous function, the difference of the distance functions, of an open interval). Note that if $\ph[X]\subset\psi[X]$ for some other map $\psi$---which means that $\Phi$ will fail to satisfy the open set condition---this set is empty. 

Now consider the new open set $U'$ defined by 
\[ U' = \bigcap_{\ph\in\Phi} \ph^{-1}[U_\ph]. \]
By construction, if $x\in U'$ then $\ph(x)\in U_\ph$. Note, however, that $U'$ may or may not be nonempty in general (even if each individual $U_\ph$ is nonempty). The osculating condition is satisfied by $U'$: for any $x\in U'$ and $\ph\in\Phi$, we have that $\ph(x)\in U_\ph$, whence $d(x,\ph[X])=d(x,X)$ by definition. The set $U'$, or some nonempty subset of it, is a candidate for an osculating set $\Omega$ if it is also a feasible set for which $\Phi$ satisfies the open set condition. In fact, if $\Phi$ consists only of injective maps (such as is the case for self-similar systems), then it follows by construction that the pairwise images of $U'$ are disjoint. In this case, one must then check that $U'$ is nonempty and that it contains its images, i.e. for each $\ph\in\Phi$, $\ph[U']\subseteq U'$.

\subsection{Induced Decompositions of Osculant Fractal Drums}
%
%

The conditions of an osculant fractal drum $(X,\Omega)$ ensure that its relative tube function $V_{X,\Omega}$ has an induced decomposition in the sense of Definition~\ref{def:inducedDecomp}. This, together with the volume scaling law (Corollary~\ref{cor:volumeScalingLaw}), will allow us to apply the results of Chapter~\ref{chap:SFEs} to set up and solve scaling functional equations for the tube function. 

\begin{theorem}[Induced Decomposition of Relative Tube Functions]
    \label{thm:inducedDecompTubes}
    \index{Decomposition!Decomposition of tube functions}
    Let $\Phi$ be a self-similar system, let $(X,\Omega)$ be an osculant fractal drum with respect to $\Phi$, and let $V_{X,\Omega}$ be the relative tube function of $(X,\Omega)$. Let $R=\Omega\setminus\cup_{\ph\in\Phi}\ph[\Omega]$, which we call the \index{Decomposition!Residual set}\textbf{residual set}. 
    \medskip 

    Then $\Phi$ induces a decomposition of $V_{X,\Omega}$ on $\RR^+$ with remainder term $V_{X,R}(t)=\measure{X_t\cap R}$ in the sense of Definition~\ref{def:inducedDecomp}.
\end{theorem}
We note that this result is essentially the proof of Theorem~5.3 in \cite{Hof25}. The final step to obtain the statement of this theorem is the application of Lemma~\ref{lem:volumeScaling}, which we state here as a separate corollary. 

\begin{proof}    
    We begin by partitioning $\Omega$ according to the iterates and the residual set: 
    \[ \Omega = \enclose{\biguplus_{\ph\in\Phi} \ph[\Omega]} \uplus R.\] 
    Observe that this is a partition since $\Omega$ is an osculating set for $\Phi$ (and by the definition of $R$). Next, we intersect each of the sets of the partition with $X_t$ and take the Lebesgue measure. By disjoint additivity, we obtain that
    \[ V_{X,\Omega}(t) = \sum_{\ph\in\Phi} V_{X,\ph[\Omega]}(t) + V_{X,R}(t), \]
    where $V_{X,R}(t)=\measure{X_t\cap R}$. It is not strictly speaking a relative tube function; however, $R$ is measurable, so it is well-defined. To see this, note that $\Omega$ is measurable and so too is each set $\ph[\Omega]$. The easiest way to see that $\ph[\Omega]$ is measurable is to note that a similitude is injective, whence its inverse is well-defined and continuous, in which case $\ph[\Omega]$ is the preimage of $\Omega$ under a continuous (hence measurable) function.

    By the osculant condition, we have that for each $\ph\in\Phi$, $d(y,X)=d(y,\ph[X])$ for all $y\in\ph[\Omega]$. Consequently, $X_t\cap\ph[\Omega]=(\ph[X])_t\cap\ph[\Omega]$, as each of these neighborhoods $X_t$ is the preimage of the interval $[0,t)$ under the respective distance functions, $d(\cdot,X)\equiv d(\cdot,\ph[X])$, equivalent in $\ph[\Omega]$ by assumption. Thus, $V_{X,\ph[\Omega]}=V_{\ph[X],\ph[\Omega]}$ on $\RR^+$. Putting everything together, we obtain the identity that 
    \[ V_{X,\Omega}(t) = \sum_{\ph\in\Phi} V_{\ph[X],\ph[\Omega]}(t) + V_{X,R}(t), \]
    which is exactly an induced decomposition.
\end{proof}

Given an induced decomposition and the scaling law from Lemma~\ref{lem:volumeScaling}, we obtain an induced scaling functional equation for $V_{X,\Omega}$ in the sense of Proposition~\ref{prop:inducedSFE}.
\begin{corollary}[Induced Scaling Functional Equation]
    \label{cor:inducedTubeSFE}
    \index{Scaling functional equation (SFE)!Induced SFE of tube functions}
    Let $\Phi$ be a self-similar system, let $(X,\Omega)$ be an osculant fractal drum with respect to $\Phi$, and let $V_{X,\Omega}$ be the relative tube function of $(X,\Omega)$. Let $R=\Omega\setminus\cup_{\ph\in\Phi}\ph[\Omega]$ be the residual set and let $V_{X,R}(t)=\measure{X_t\cap R}$. 
    \medskip 

    Then $\Phi$ induces the scaling functional equation $V_{X,\Omega}=L_\Phi[V_{X,\Omega}]+V_{X,R}$ on $\RR^+$, where $L_\Phi$ is the scaling operator associated to $\Phi$ as in Equation~\ref{eqn:defSystemScalingOp}.
\end{corollary}
\subsection{Complex Dimensions}
%
%

In order to apply the results of Chapter~\ref{chap:SFEs}, we must have some control over the remainder term, $V_{X,R}$, that appears in the induced functional equation of Corollary~\ref{cor:inducedTubeSFE}. Let $(X,\Omega)$ be an osculant fractal drum with respect to a self-similar system $\Phi$, and let $R=\Omega\setminus\cup_{\ph\in\Phi}\ph[\Omega]$ and $V_{X,R}$ be defined accordingly. 

Our main type of estimate is when there exists some $\sigma_0<\underline{\dim}_S$ such that at $t\to0^+$, we have that $R(t)=O(t^{\dimension-\sigma_0})$. Then $t^{-\dimension}V_{x,R}(t)=O(t^{-\sigma_0})$ as $t\to0^+$ and thus by Theorem~\ref{thm:lowerDimAdmissibility} it is an admissible remainder for the scaling functional equation $V_{X,\Omega}=L_\Phi[V_{X,\Omega}]+V_{X,R}$ for any screen of the form $S(\tau)\equiv \sigma_0+\e$, $\e>0$ arbitrarily small. 

Given such an estimate, we can deduce an explicit formula for the relative tube zeta function of $(X,\Omega)$, and thus the possible complex dimensions as well as explicit tube formulae in either a pointwise or distributional sense. 

\begin{theorem}[Tube Zeta Function Formula]
    \label{thm:tubeZetaFormula}
    \index{Tube zeta function!Explicit formula}
    Let $(X,\Omega)$ be an osculant fractal drum in $\RR^\dimension$ with corresponding self-similar system $\Phi$. Let $V_{X,\Omega}$ be its relative tube function and let $V_{X,R}$ be the relative tube function of the residual set (as in Theorem~\ref{thm:inducedDecompTubes}). Suppose that as $t\to0^+$, $V_{X,R}(t)=O(t^{\dimension-\sigma_0})$. 
    \medskip 
    
    Then the relative tube zeta function of $(X,\Omega)$ is given by 
    \begin{equation}
        \label{eqn:tubeZetaFormula}
        \tubezeta_{X,\Omega}(s;\delta) = \zeta_\Phi(s)(\xilfphi(s;\delta)+\zeta_{X,R}(s;\delta)),
    \end{equation} 
    where $\xilfphi$ is defined as in Equation~\ref{eqn:defPartialZeta} with $f(t)=t^{-\dimension}\tubefn(t)$. Letting $D=\simdim(\Phi)$, we have that $\tubezeta_{X,\Omega}(s;\delta)$ is holomorphic in $\HH_{D}$. Further, $\xilfphi$ is an entire function and $\zeta_R$ is holomorphic in $\HH_{\sigma_0}$, whereby $\tubezeta_{X,\Omega}$ admits a meromorphic continuation to $\HH_{\sigma_0}$, having poles contained in the set $\Dd_\Phi(\HH_{\sigma_0})$. If $\sigma_0<\lowersimdim(\Phi)$, then $\Dd_\Phi(\HH_{\sigma_0})=\Dd_\Phi(\CC)$.
\end{theorem}
\begin{proof}
    First, we note that by Theorem~\ref{thm:inducedDecompTubes}, $\Phi$ induces a decomposition of $V_{X,\Omega}$ on $\RR^+$ with remainder term $R$. By Corollary~\ref{cor:volumeScalingLaw}, we have that $V_{X,\Omega}$ satisfies a $1$-scaling law, so by Proposition~\ref{prop:inducedSFE}, $\Phi$ induces the scaling functional equation $V_{X,\Omega}=L_\Phi[V_{X,\Omega}]+R$ on $\RR^+$. Thus, by Corollary~\ref{cor:structureOfPoles} (with $\alpha=1$), we have that Equation~\ref{eqn:tubeZetaFormula} holds and that $\tubezeta_{X,\Omega}$ is holomorphic in $\HH_D$ and meromorphic in $\HH_{\sigma_0}$, having poles contained in $\Dd_\Phi(\HH_{\sigma_0})$. If we have that $\sigma_0<D_\ell=\lowersimdim(\Phi)$, we have that $\Dd_\Phi(\HH_{\sigma_0})=\Dd_\Phi(\CC)$ because all the poles $\omega$ of $\zeta_\Phi$ have $D_\ell\leq \Re(\omega)\leq D$.
\end{proof}

From this theorem, we can describe the possible complex dimensions of the fractal drum explicitly, as they are governed by the explicit function $\zeta_\Phi$. Note, however, that it is possible for the residue of the tube function to vanish at a pole $\omega$ of $\zeta_\Phi$, in which case that it is not a complex dimension of $(X,\Omega)$. This occurs if $\xilfphi(\omega;\delta)+\zeta_{X,R}(\omega;\delta)=0$ and vanishes to a degree equal to or larger than than the multiplicity of the pole $\omega\in\Dd_\Phi$.
\begin{corollary}[Possible Complex Dimensions]
    \label{cor:possibleCdims}
    \index{Relative fractal drum (RFD)!Possible complex dimensions}
    Let $(X,\Omega)$ be an osculant fractal drum in $\RR^\dimension$ with corresponding self-similar system $\Phi$ and suppose that the hypotheses of Theorem~\ref{thm:tubeZetaFormula} hold.
    \medskip 
    
    Then the set $\Dd_X(\HH_{\sigma_0})$ of complex dimensions of $X$ relative to $\Omega$ in the window $\HH_{\sigma_0}$ satisfies the containment
    \begin{equation}
        \label{eqn:cDimsFullDim}
        \Dd_X(\HH_{\sigma_0}) \subset \Dd_{\Phi}(\HH_{\sigma_0}) = \Big\{\omega \in \CC : 1 = \sum_{\ph\in\Phi} \lambda_\ph^\omega \Big\}.
    \end{equation}
    A point $\omega\in\Dd_\Phi(\HH_{\sigma_0})$ is a complex dimension of $X$ (relative to $\Omega$) exactly when $\Res(\tubezeta_{X,\Omega}(s;\delta);\omega)\neq 0$. If $\sigma_0<\lowersimdim(\Phi)$, then $\Dd_\Phi(\HH_{\sigma_0})=\Dd_\Phi({\CC})$.
\end{corollary}

Of note, Theorem~\ref{cor:possibleCdims} requires a suitable estimate for the remainder to produce any results. For example, the knowledge that $V_{X,R}$ is bounded yields no information, as the estimate corresponds to $\sigma_0=\dimension$ and in $\HH_\dimension$, $\zeta_\Phi$ is holomorphic since $\simdim(\Phi)\leq\dimension$. The result in general requires that there exists a screen on which $\zeta_\Phi$ and $\zeta_R$ are both jointly languid. Imposing that $\sigma_0<\lowersimdim(\Phi)$ is tantamount to imposing that the remainder is holomorphic in a half-plane which contains all of the poles of $\zeta_\Phi$, which are contained in the strip defined by $\lowersimdim(\Phi)\leq\Re(s)\leq\simdim(\Phi)$. 

It is possible to extend the result even if this is not possible, but it requires knowledge of the locations of the poles of $\zeta_\Phi$ in the form of dimension-free regions. This is easily done in the lattice case (see Definition~\ref{def:latticeDichotomy}), but is in general a much more challenging problem in the generic nonlattice case.

The best results occur if the parameter $\sigma_0$ is small; in Section~\ref{sec:appToGKFs}, our results shall apply for any order $\e>0$ sufficiently close to $0$. Informally, $\sigma_0$ corresponds to the (Minkowski) dimension of the remainder quantity: when $V_{X,R}$ scales like the volume of an epsilon neighborhood of a point, the exponent is $\dimension$ and thus $\sigma_0=0$; when $V_{X,R}$ scales like the volume of a neighborhood of a line, so the exponent is $\dimension-1$ and $\sigma_0=1$; and so forth. (Recall the discussion in Subsection~\ref{sub:tubeCodimension}.) In other words, our work determines the complex dimensions with real parts (or amplitudes) up to the real part corresponding to the dimension of (the estimate of) the remainder term.

\subsection{Explicit Formulae}
\label{sub:generalTubeFormulae}
%
%

Having established the identities for the tube zeta functions and deduced the possible complex dimensions, we may now proceed to proving explicit formulae for the relative tube functions themselves. First, though, we recall some notation from Section~\ref{sec:SFEsolutions} adapted to the present context. Let $(X,\Omega)$ be an osculant fractal drum in $\RR^\dimension$ with associated self-similar system $\Phi$. 

Firstly, define $\tubefn^{[0]}:=\tubefn$. For any $k>0$, we define the $k\th$ antiderivative $\tubefn^{[k]}$ recursively by integrating the previous antiderivative and imposing the convention that $f^{[k]}(0)=0$. That is, for $k>0$,
\[
    \tubefn^{[k]}(t) := \int_0^t \tubefn^{[k-1]}(\tau)\,d\tau.
\] 
Next, we recall that the Pochhammer symbol is defined by $(z)_w:=\Gamma(z+w)/\Gamma(w)$ for $z,w\in\CC$. In the special case that $w=k$ is a positive integer, this simplifies to $(z)_k=z(z+1)\cdots(z+k-1)$. If $w=0$, then we have that $(z)_0=1$.

Lastly, we recall an important caveat regarding the nature of the sums which will appear forthwith. Namely, they are symmetric limits and thus may have delicate convergence. Letting $\Dd\subset\CC$ denote a discrete subset with the property that for any $m\in\NN$, 
\[ 
    \Dd_m:=\set{\omega\in\Dd\suchthat |\Im(\omega)|\leq m}<\infty.
\] 
We define 
\[ 
    \sum_{\omega\in\Dd} a_\omega := \lim_{m\to\infty} \sum_{\omega\in\Dd_m}a_\omega.  
\]
In what follows, $\Dd$ will be the set of poles of the zeta function $\zeta_\Phi$ contained in a vertical strip. When the imaginary parts are bounded, the region is thus also finite. It is an elementary fact (following from the identity theorem) that a meromorphic function can have only finitely many poles in a bounded, connected region, so the sums over the set of poles of $\Dd_\Phi$ are well-defined. 

\begin{theorem}[Pointwise Explicit Tube Formula]
    \label{thm:pointwiseTubeFormula}
    \index{Tube function!Pointwise formulae}
    Let $(X,\Omega)$ be an osculant fractal drum in $\RR^\dimension$ with corresponding self-similar system $\Phi$. Let $V_{X,\Omega}$ be its relative tube function and let $V_{X,R}$ be the relative tube function of the residual set (as in Theorem~\ref{thm:inducedDecompTubes}). Suppose that as $t\to0^+$, $V_{X,R}(t)=O(t^{\dimension-\sigma_0})$. Suppose either that $\sigma_0<\lowersimdim(\Phi)$ (see Definition~\ref{def:lowerSimDim}) or that the scaling ratios of $\Phi$ are arithmetically related (see Definition~\ref{def:latticeDichotomy}). 
    \medskip 

    Then for every $k\geq 2$ in $\ZZ$ and every $t\in(0,\delta)$, we have that
    \[ 
        \tubefn^{[k]}(t) = \sum_{\omega\in\Dd_\Phi(\HH_{\sigma_0})} 
            \Res\Bigg(\cfrac{t^{\dimension-s+k}}{(\dimension-s+1)_k}\tubezeta_{X,\Omega}(s;\delta);\omega\Bigg)
            + \Rr^k(t),
    \]
    where $\tubezeta_{X,\Omega}$ is given explicitly by Theorem~\ref{thm:tubeZetaFormula}. For any $\e>0$ sufficiently small, the error term satisfies $\Rr^k(t)=O(t^{\dimension-\sigma_0-\e+k})$ as $t\to0^+$. 
\end{theorem}
\begin{proof}
    This theorem is a corollary of Theorem~\ref{thm:pointwiseFormula} from Chapter~\ref{chap:SFEs}. Note that $\beta=\dimension$, $\alpha=1$, and $\dimension+1>\simdim(\Phi)\geq\lowersimdim(\Phi)>\sigma_0$.

    The induced scaling functional equation, with the scaling operator $L_\Phi$, is the result of Corollary~\ref{cor:inducedTubeSFE}. The condition that $\sigma_0<\lowersimdim(\Phi)$ is a sufficient condition for the admissibility of the remainder term and any screen of the form $S_\e(\tau)\equiv \sigma_0+\e$ per Theorem~\ref{thm:lowerDimAdmissibility}, together with the presupposed estimate for $V_{X,R}$. In the event that the scaling ratios are arithmetically related, instead we may apply Theorem~\ref{thm:latticeCaseAdmissibility} to deduce admissibility of the remainder and screens $S_\e$. Note that $W_{S_\e}=\HH_{\sigma_0+\e}$ is simply a half-plane.

    In both cases, $\e$ is chosen so that $|\Re(\omega)-\sigma_0-\e|>0$ for any $\omega\in\Dd_\Phi(\HH_{\sigma_0})$. (Either we choose $\e<D_\ell-\sigma_0$ or $\e$ to be smaller than the minimum distance to any of the finitely many exceptional real parts.) By the estimate for $V_{X,R}$, we have that $\tubezeta_R$ is holomorphic in $\HH_{\sigma_0}$ and $\sup(S)=\sigma_0+\e>\sigma_0$, so $S_\e$ is contained in the half plane $\HH_{\sigma_R}=\HH_{\sigma_0}$. 

    Lastly, we note that in these specific cases, we may reindex by the set $\HH_{\sigma_0}$ rather than $\HH_{\sigma_0+\e}$ for uniformity. This follows because $\zeta_\Phi$, and thus $\tubezeta_{X,\Omega}$, has no poles with real part $\Re(s)\in(\sigma_0,\sigma_0+\e]$ by our choice of $\e$ sufficiently small.
\end{proof}

Next, we recall the theory of Schwartz distributions. For our test functions, let $\Ss(0,\delta)$ be the space of 
functions of rapid decrease, given by 
\begin{equation}
    \label{eqn:defSchwartzFns}
    \Ss(0,\delta):=\left\{\testfn\in C^\infty(0,\delta)\suchthat 
    \begin{aligned}
        &\forall m\in\ZZ,\,\forall q\in\NN,\text{ as }t\to0^+,\\ 
        &t^m\testfn^{(q)}(t)\to0 \text{ and }(t-\delta)^m\testfn^{(q)}(t)\to0\,
    \end{aligned}
    \right\}.
\end{equation}
The space of Schwartz distributions is the dual space, $\Ss'(0,\delta)$. For a distribution $F\in\Ss'(0,\delta)$ and a test function $\testfn\in\Ss(0,\delta)$, we write $F(\testfn)=\bracket{F,\testfn}$. Two distributions $F,G$ are said to be equal if for any $\testfn\in\Ss(0,\delta)$, $\bracket{F,\testfn}=\bracket{G,\testfn}$. 

Given a test function $\testfn\in\Ss(0,\delta)$ and $a>0$, define the new test function $\testfn_a(t):= a^{-1}\testfn(t/a)$. A distribution $\Rr$ is said to be $\Rr(t)=O(t^\alpha)$ as $t\to0^+$ if for all $a>0$ and for all $\testfn\in\Ss(0,\delta)$,
\begin{equation}
    \label{eqn:distRemEst2}
    \bracket{\Rr,\testfn_a(t)} = O(a^{\alpha}),
\end{equation}
as $t\to0^+$ in the usual sense. 

\begin{theorem}[Distributional Explicit Tube Formula]
    \label{thm:distTubeFormula}
    \index{Tube function!Distributional formulae}
    Let $(X,\Omega)$ be an osculant fractal drum in $\RR^\dimension$ with corresponding self-similar system $\Phi$. Let $V_{X,\Omega}$ be its relative tube function and let $V_{X,R}$ be the relative tube function of the residual set (as in Theorem~\ref{thm:inducedDecompTubes}). Suppose that as $t\to0^+$, $V_{X,R}(t)=O(t^{\dimension-\sigma_0})$. Suppose either that $\sigma_0<\lowersimdim(\Phi)$ (see Definition~\ref{def:lowerSimDim}) or that the scaling ratios of $\Phi$ are arithmetically related (see Definition~\ref{def:latticeDichotomy}). 
    \medskip

    Then for any $k\in\ZZ$, we have that for $t\in(0,\delta)$, $\tubefn^{[k]}$ satisfies
    \begin{equation}
        \label{eqn:distTubeFormula}
        \begin{split}
            \tubefn^{[k]}(t) &= \sum_{\omega\in\Dd_\Phi(\HH_{\sigma_0})} 
                \Res\Bigg(\cfrac{t^{\dimension-s+k}}{(\dimension-s+1)_k}\tubezeta_{X,\Omega}(s;\delta);\omega\Bigg)
                + \Rr^{[k]}(t),
        \end{split} 
    \end{equation}
    in the sense of distributions. See Equation~\ref{eqn:bracketIdentity2} for the explicit identity of action on test functions. Here, $\tubezeta_{X,\Omega}$ is given explicitly by Theorem~\ref{thm:tubeZetaFormula} and the distributional remainder term satisfies the estimate $\Rr(t)=O(t^{\dimension-\sigma_0-\e+k})$ as $t\to0^+$, in the sense of Equation~\ref{eqn:distRemEst2}, for any $\e>0$ sufficiently small. 
\end{theorem}
Most precisely, Equation~\ref{eqn:distTubeFormula} means that for any $\testfn\in\Ss(0,\delta)$, 
\begin{equation}
    \label{eqn:bracketIdentity2}
    \begin{split}
        \bracket{\tubefn^{[k]},\testfn} &= \sum_{\omega\in\Dd_\Phi(\HH_{\sigma_0})} 
        \Res\Bigg(\cfrac{\Mellin[\testfn]({\dimension-s+k+1})}{(\dimension-s+1)_k}\tubezeta_{X,\Omega}(s);\omega\Bigg)
        + \bracket{\Rr^{[k]},\testfn}.
    \end{split}
\end{equation}

\begin{proof}
    This theorem is a corollary of Theorem~\ref{thm:distFormula} from Chapter~\ref{chap:SFEs}. The application of the theorem follows the same comments and justifications as in the proof of Theorem~\ref{thm:pointwiseTubeFormula} in regards to the induced scaling functional equation, values of the parameters, admissibility of the remainder, and so forth. 
\end{proof}

\begin{corollary}[Distributional Explicit Tube Formula, $k=0$]
    \label{cor:distTubeFormulaZero}
    \index{Tube function!Distributional formulae}
    Let $(X,\Omega)$ be an osculant fractal drum in $\RR^\dimension$ with corresponding self-similar system $\Phi$ be such that the hypotheses of Theorem~\ref{thm:distTubeFormula} holds. 
    \medskip

    Then we have that for $t\in(0,\delta)$, $\tubefn$ satisfies
    \begin{equation}
        \label{eqn:distTubeFormulaZero}
        \begin{split}
            \tubefn(t) &= \sum_{\omega\in\Dd_\Phi(\HH_{\sigma_0})} 
                \Res\Bigg({t^{\dimension-s}}\tubezeta_{X,\Omega}(s;\delta);\omega\Bigg)
                + \Rr^{[k]}(t),
        \end{split} 
    \end{equation}
    in the sense of distributions, and the remainder satisfies the estimate $\Rr(t)=O(t^{\dimension-\sigma_0-\e+k})$ as $t\to0^+$, in the sense of Equation~\ref{eqn:distRemEst2} with $k=0$, for any $\e>0$ sufficiently small.
    \medskip

    If we further assume that the poles of $\zeta_\Phi$ in $\HH_{\sigma_0}$ are simple, then Equation~\ref{eqn:distTubeFormulaZero} simplifies to 
    \[
        \tubefn(t) = \sum_{\omega\in\Dd_\Phi(\HH_{\sigma_0})} 
            \Res(\tubezeta_{X,\Omega}(s;\delta);\omega)\,t^{\dimension-\omega} + \Rr(t).
    \]
\end{corollary}
\section{Examples: Generalized von~Koch Fractals}
\label{sec:appToGKFs}
%
%

In this section, we shall apply the results of Section~\ref{sec:fractalSFEs} to generalized von Koch fractals (or GKFs.) See Section~\ref{sec:fractalExamples} for the definitions and properties concerning these fractals. In particular, we provide a scaling functional equation for the tubular neighborhood volume and estimates of the error term that allow us to describe the possible complex dimensions of these GKFs. 

A tube formula (see Definition~\ref{def:tubeFunction}) for the standard von Koch snowflake was established in the work of Lapidus and Pearse \cite{LP06}. They computed that the volume of an inner epsilon neighborhood of the von Koch snowflake takes the form
\[ V(\e) = \sum_{n\in\ZZ}\phi_n \e^{2-D-in P} + \sum_{n\in\ZZ}\psi_n \e^{2-in P}, \]
where $D=\log_3 4$ is the Minkowski dimension of the snowflake, $P=2\pi/\log 3$ is the multiplicative period of the oscillations, and with $\psi_n,\phi_n$ as constants depending only on $n$ \cite{LP06}. As a consequence, they deduced the possible complex (fractal) dimensions of the von Koch snowflake. These complex dimensions encode both the amplitude and period of geometric oscillations in a fractal, and are pivotal to establishing such tube formulae in the work of Lapidus and collaborators; see \cite{LapvFr13_FGCD,LRZ17_FZF} and the references therein. 

In \cite{Hof25}, we extended this analysis to general $(n,r)$-von Koch snowflakes, where $r$ is chosen to be sufficiently small, using the methods which have been generalized here in Chapter~\ref{chap:SFEs}. In particular, we deduced the possible complex dimensions of these GKFs with positive real part. In this setting, the lattice/nonlattice dichotomy arises depending on whether the scaling ratio $r$ and its conjugate factor $\ell=\frac12(1-r)$ are arithmetically related or not (see Definition~\ref{def:latticeDichotomy}). For the standard $(3,\frac13)$-von Koch snowflake, the ratios are the same: $r=\ell$, and thus the fractal falls under the lattice case. The emergence of the nonlattice case is unique to the generalizations of the von Koch snowflake. 

\subsection{Geometric Preliminaries}
%
%

Let $\Knr$ be a $(n,r)$-von Koch snowflake which is a simple, closed curve.\footnote{See Proposition~\ref{prop:selfAvoid} for a sufficient condition based on $r$ and $n$.} By the Jordan curve theorem, a simple closed curve bisects the plane into two connected components, one of which is bounded and the other unbounded. Let us denote by $\Omega$ the bounded component; then $\partial\Omega=\Knr$. By Proposition~\ref{prop:selfAvoid}, we have an explicit upper bound on $r$ which is a sufficient condition for this to hold.

\begin{figure}[t]
    \centering
    \subfloat{\includegraphics[width=4cm]{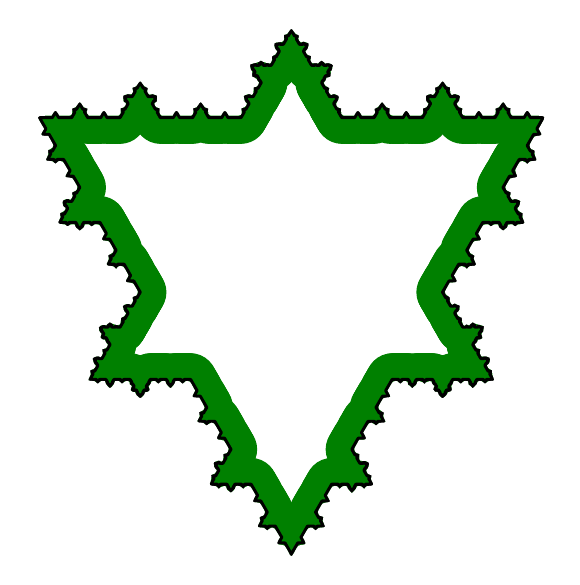}}
    \qquad
    \subfloat{\includegraphics[width=4cm]{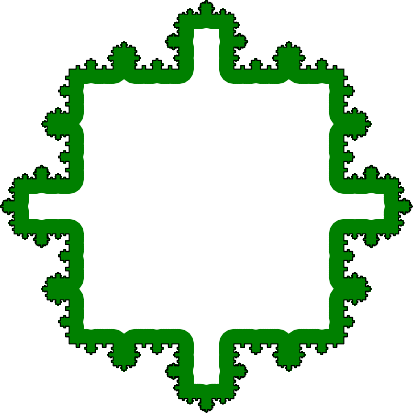}}
    \caption[Inner tubular neighborhoods of some fractal snowflakes]{Inner tubular neighborhoods of the fractals $K_{3,\frac15}$ (left) and $K_{4,\frac15}$ (right).}
    \label{fig:innerTubeNbds}
\end{figure}

Note that $(\Knr,\Omega)$ is readily seen to be a relative fractal drum since the sets are both bounded. We shall be interested in describing the inner tubular neighborhood of $\Knr$, which is the tube function of $\Knr$ relative to the set $\Omega$. See Figure~\ref{fig:innerTubeNbds} for two such examples.

We note that by the $n$-fold symmetry of these regions, it is enough to compute a formula for the portion of the neighborhood contained in a sector of angle $2\pi/n$, chosen so that the rays defining the section pass through the edges of the regular $n$-gon used to define the generalized von Koch fractal. So, let $S$ be any one of the $n$ congruent (open) sectors defined by the center of $\Knr$ and two adjacent vertices of this underlying regular $n$-gon. For simplicity of notation, we define the new relative tube function by 
\[ V_{K,U}(t) := V_{\Knr,\Omega\cap S}(t) = \frac1n V_{\Knr,\Omega}(t),  \]
defined for the RFD $(K,U)=(\Knr,\Omega\cap S)$. Note that $K$ is precisely an $(n,r)$-von Koch curve. Let $\Phinr$ be the explicit self-similar system defined in Equation~\ref{eqn:defGKCsystem} and, up to application of an isometry, we may without loss of generality assume that $K$ is the attractor of $\Phinr$. 

\begin{figure}[t]
    \centering
    \includegraphics[width=6cm]{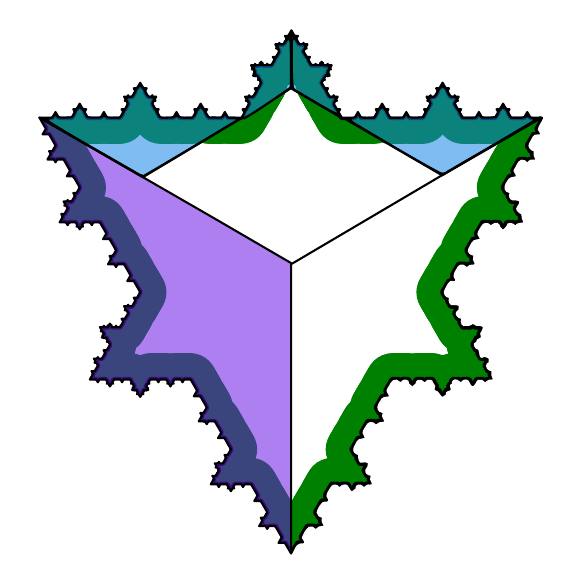}
    \caption[A partitioned inner tubular neighborhood of a generalized von Koch fractal]{An inner tubular neighborhood of a $(3,\frac15)$-von Koch snowflake. A relative neighborhood in one sector of angle $2\pi/3$ is shaded in the bottom left and images of the relative neighborhood in the top third, under mappings of the self-similar system, are shaded in the top. Note the leftover portion of the inner tubular neighborhood contained in the residual set.}
    \label{fig:partitionGKF}
\end{figure}

If we can show that this is an osculant fractal drum and appropriately estimate the remainder, then we will be able to apply the results of Section~\ref{sec:fractalSFEs}. We may partition $K_t\cap U$ into $2+(n-1)$ self-similar copies of itself as well as $4$ pieces contained within congruent triangles and $2$ pieces which are exactly circular sectors. See Figure~\ref{fig:partitionGKF} for a depiction of such a decomposition. This partitioning leads to the following theorem.

\begin{theorem}[SFE for an $(n,r)$-von Koch Curve]
    \label{thm:vonKochSFE}
    Let $\Knr$ be an $(n,r)$-von Koch snowflake boundary satisfying the self-avoidance criterion in Proposition~\ref{prop:selfAvoid} and let $\Omega$ be the interior region defined by $\Knr$ (i.e. the bounded component of $\RR^2\setminus\Knr$). Let $S$ be a sector of angle $2\pi/n$ centered in the regular $n$-gon used to define $\Knr$ and with edges passing through two adjacent vertices of this polygon. 
    \medskip

    Then $(K,U)=(\Knr,\Omega\cap S)$ is an osculant fractal drum with respect to the self-similar system $\Phinr$. For any $t\geq 0$, we that $V_{K,U}$ satisfies the scaling functional equation 
    \begin{equation}
        \label{eqn:volumeFunctionalEq}
        t^{-2}V_{K,U}(t) = 2(t/\ell)^{-2} V_{K,U}(t/\ell) + (n-1)(t/r)^{-2}V_{K,U}(t/r) + t^{-2}V_{K,R}(t),
    \end{equation}
    where $\ell=\frac12(1-r)$ is the conjugate scaling factor to $r$. We have that $V_{K,R}(t)=O(t^2)$ as $t\to 0^+$ with the explicit estimate that $V_{K,R}(t)(t)\leq (2\cot(\theta_n/2)+\theta_n)t^2$ for any $t\geq0$. Here, $\theta_n=\frac{2\pi}n$ is the central angle of a regular $n$-gon.   
\end{theorem}

\begin{proof}
    Without loss of generality, we may suppose that $K:=\Knr\cap S$ is exactly the $(n,r)$-von Koch curve $\Cnr$; in general, $K$ is isometric to this set. In this setting, we have explicitly the similitudes for the self-similar system $\Phinr=\set{\phi_L,\phi_R,\psi_k, k=1,...,n-1}$ which defines $\Cnr$ as in Definition~\ref{def:vkCurve}. (In the general setting, simply compose these maps $\ph\in\Phinr$ with the isometry between the sets.) 

    By definition, $K$ is the attractor of the self-similar system $\Phi$ in Definition~\ref{def:vkCurve}. There are two similitudes $\phi_L$ and $\phi_R$ with scaling ratios $\ell:=\frac12(1-r)\in(0,1)$ and $n-1$ similitudes $\psi_k$ with scaling ratios $r\in(0,1)$. So, we define the scaling operator $L=L_{n,r}$ by $L:= 2M_\ell + (n-1)M_r.$ 

    Let $U=\Omega\cap S$. Then we have that $U$ is an osculating set for $\Phi$. Indeed, the collection of sets $\set{\phi_L[U],\phi_R[U],\psi_k[U],k=1,...,n-1}$ is readily seen to be pairwise disjoint and each set is contained within $U$ itself. $(K,U)$ is a readily seen to be a relative fractal drum since $K$ is contained in the bound of $U$ and $U$ is a bounded, measurable set. 
    
    The osculating condition may be seen from the symmetry involved in the construction of $K$. In particular, we shall use the property that if $X$ and $X'$ are two sets in $\RR^2$ related by reflection about the line $L$, then the set of points which are equidistant to $X$ and $X'$ must lie on $L$. We note that by choice of $U$, its images $\psi_k[U]$ partition a regular $n$-gon scaled by a factor of $r$ about the lines of symmetry passing from the center of the shape to its vertices. Because the adjacent images $\psi_j[K]$ and $\psi_{j+1}[K]$ are reflections of each other about these lines, it follows that $d(\cdot,K)=\min_{\phi\in\Phi}d(\cdot,\phi[K])$ can only change form across these lines by continuity of the distance functions. In other words, it follows that $d(\cdot,K)\equiv d(\cdot,\psi_j[K])$ in each $\psi_j[U]$, $j=1,...,n-1$. 

    Now, the other two iterates on the edges are simpler. Indeed, we have that $d(\cdot,K)\equiv d(\cdot,\phi_L[K])$ in $\phi_L[U]$ (resp. $d(\cdot,K)\equiv d(\cdot,\phi_R[K])$ in $\phi_R[U]$) is clear from the geometry of $K$ since a ball $B_\e(y)$ for $y\in\phi_L[U]$ (resp. $y\in\phi_R[U]$) will intersect with $\phi_L[K]$ (resp. $\phi_R[K])$ before ever reaching $\psi_1[K]$ (resp. $\psi_{n-1}[K]$). 

    It follows by Corollary~\ref{cor:inducedTubeSFE} that $t^{-2}V_{K,U}(t)$ satisfies the scaling functional equation 
    \[ t^{-2}V_{K,U}(t) = 2(t/\ell)^{-2} V_{K,U}(t/\ell) + (n-1)(t/r)^{-2}V_{K,U}(t/r) + t^{-2}V_{K,R}(t), \]
    where $R=U\setminus(\cup_{\phi\in\Phi} \phi[U])$ and $V_{K,R}(t) = \measure{K_t\cap R}$. We note that $K_t\cap R$ may be partitioned into two sectors of angle $\theta_n$ and length $\e$ and four congruent pieces contained in triangles of height $t$ and width $t\cot(\theta_n/2)$. See Figure~\ref{fig:trianglet} for a depiction of two such triangles. Therefore, we may write that
    \[ R(\e) := V_{K,R}(t) \leq \theta_nt^2+2t^2\cot(\theta_n/2). \]
\end{proof}

\begin{figure}[t]
    \centering
    \includegraphics[width=0.6\textwidth,trim=10 20 10 0, clip]{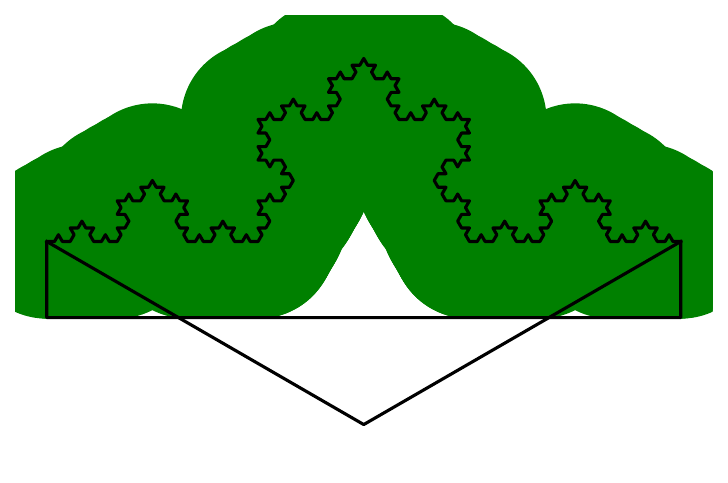}
    \caption[A portion of the tubular neighborhood of the von Koch curve]{A portion of the tubular neighborhood of the von Koch curve $C_{3,1/3}$ is depicted. In particular, note the regions contained in the triangles of height $t$ and width $t\cot(\theta_n/2)$, where $\theta_n=2\pi/n$, on the left and the right of the figure. These are used as part of an estimate the volume of the relative tubular neighborhood of a von Koch snowflake contained in the residual set.}
    \label{fig:trianglet}
\end{figure} 

Because we have that $V_{\Knr,\Omega}(t)=nV_{K,U}(t)$, we obtain the following scaling functional equation for the relative fractal drum of the boundary of the fractal snowflake relative to its interior as an immediate corollary. Note that the remainder term is also scaled by a factor of $n$ but otherwise satisfies the same estimate. 

\begin{corollary}[SFE for an $(n,r)$-von Koch Snowflake]
    \label{cor:fullVonKochSFE}
    Let $\Knr$ be an $(n,r)$-von Koch Snowflake boundary satisfying the self-avoidance criterion in Proposition~\ref{prop:selfAvoid} and let $\Omega$ be the interior region defined by $\Knr$ (i.e. the bounded component of $\RR^2\setminus\Knr$). 
    \medskip 

    Then the tube function $V_{\Knr,\Omega}$ of the relative fractal drum $(\Knr,\Omega)$ satisfies the scaling functional equation
    \[ t^{-2}V_{\Knr,\Omega}(t) = L_\Phinr[t^{-2}V_{\Knr,\Omega}(t)](t) + t^{-2}V_{\Knr,R}(t) \]
    where $\Phinr$ is as in Equation~\ref{eqn:defGKCsystem}. The remainder term is continuous and satisfies the estimate $|V_{\Knr,R}(t)| \leq (2n\cot(\pi/n)+2\pi)t^2$ as $t\to0^+$. 
\end{corollary}

Next, we consider the lower similarity dimension of generalized von Koch snowflakes. These depend only on $n$ and $r$, since $\ell=(1-r)/2$. Since our remainder estimate is of the form $t^{-2}V_{K,R}(t)=O(t^{0})$ (as $t\to0^+$), in order to apply the results of Section~\ref{sec:fractalSFEs} without further assumption, we need information about the lower similarity dimension, $D_\ell=\lowersimdim(\Phi)$. Because $\sigma_0=0$ and $r<\ell$ (or in the case when $r=\ell=\frac13$), it turns out that whether or not $D_\ell>0$ is independent of the value of $r$, depending only one the multiplicities, and thus only on $n$. We have that when $n\geq5$, $D\ell>0$, when $n=4$, $D_\ell=0$, and when $n=3$, $D_\ell<0$. 

In the latter two cases, we will need to assume that the real parts of the poles of $\zeta_\Phi$ are not dense in a right hand neighborhood of the origin so that screens of the form $S_\e(\tau)=\e$ exist on which $\zeta_\Phi$ is jointly languid with the remainder zeta function. Assuming that the ratios are arithmetically related, viz. assuming the lattice case, is sufficient. Notably, in the case of GKFs, it is actually possible to show that the remainder admits a meromorphic extension to the entire complex plane with a (jointly) strongly languid remainder zeta function, in which case this assumption may be dispensed with. This will be the subject of a future work. 

\begin{lemma}[Sufficient Lower Dimension Bounds]
    \label{lem:lowerDimBoundsGKF}
    Let $n\geq3$, $r\in(0,\frac13]$, and define $\ell=(1-r)/2$. Let $\Phi$ a self similar system having scaling ratios $r$ with multiplicity $n-1$ and $\ell$ with multiplicity $2$, such as $\Phinr$ as in Equation~\ref{eqn:defGKCsystem}. Denote by $D_\ell=\lowersimdim(\Phi)$. Then we have the following.
    \begin{align*}
        \text{If }n=3,\,D_\ell<0. && \text{If }n=4,\,D_\ell=0. && \text{If }n\geq5,\,D_\ell>0.
    \end{align*}
\end{lemma}
\begin{proof}
    Let $\zeta_\Phi$ the zeta function of $\Phi$ and let $P(s)$ denote its denominator. Because $r\leq\frac13$, we have that $r\leq \ell$. Suppose first that $r=\ell=\frac13$. Then we have that 
    \[ P(s) = 1 - 2\ell^s-(n-1)r^s = 1-(n+1)3^{-s}. \]
    We may explicitly solve to find that $P(\omega)=0$ when $(n+1)=3^\omega$, or $\omega = \log_3(n+1)+2\pi i k/\log 3$ for $k\in\ZZ$. Since $n\geq3$, the real part of these poles are always positive and exactly equal to $\log_3(n+1)$. Thus, $0<\lowersimdim(\Phi)=\log_3(n+1)$. 

    Now suppose that $r<\frac13$, in which case $r<\ell$. Consider the polynomial 
    \[ p(t) = \frac1{n-1} \enclose{r^{-1}}^t + \frac2{n-1} \enclose{\frac{\ell}{r}}^t; \]
    by Definition~\ref{def:lowerSimDim}, $D_\ell=\lowersimdim(\Phi)$ is the unique real solution to $p(D_\ell)=1$. Note that $p$ is strictly increasing, so if $p(t)<1$, then $t<D_\ell$. We have that $p(0)=\frac{3}{n-1}$. If $n\geq5$, then $p(0)<1$. If $n=4$, $p(0)=1$ and thus $D_\ell=0$. If $n=3$, $p(0)>1$ whence $D_\ell<0$.  
\end{proof}

\subsection{Complex Dimensions of GKFs}
%
%

Now we will deduce the possible complex dimensions of generalized von Koch fractals. In this work, we will assume that either $n\geq5$ or that the scaling ratios $r$ and $\ell=(1-r)/2$ are arithmetically related. In the case of two scaling ratios, this condition explicitly means that the ratio of their logarithms must be rational:
\[ \frac{\log\ell}{\log r}=\frac{\log{(\frac{1-r}2)}}{\log{r}} = \frac{p}q \in\QQ. \]
This ratio, viewed as a function of $r$, is continuous and strictly decreasing, taking every value in $(0,2]$ over the interval $(0,\frac12]$; see Figure~\ref{fig:LogScalingRatioPlot} for a plot of this function. Thus, by the density of the rational numbers, for any desired $r\in(0,\frac12]$ we may find a choice of $r_0$ arbitrarily close to $r$ for which the corresponding $(n,r_0)$-von Koch snowflake is in the lattice case. 

\begin{figure}[t]
    \centering
    \includegraphics[width=6cm]{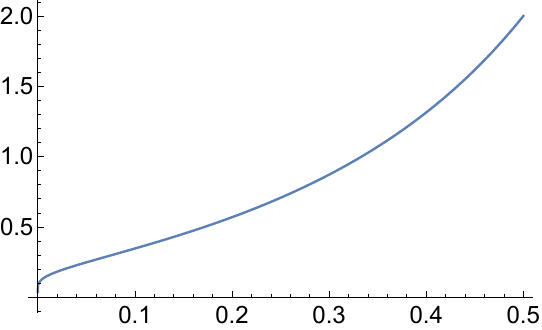}
    \caption[A plot of the ratios of logarithms of scaling ratios]{A plot of $\log\ell/\log r$, the ratio of logarithms of the scaling ratios, where $r\in(0,\frac12]$ and $\ell=(1-r)/2$. It is strictly increasing with range $(0,2]$, taking the value $1$ when $r=\ell=\frac13$ and $2$ when $r=\frac12$. As $r\to0^+$, it converges to zero.}
    \label{fig:LogScalingRatioPlot}
\end{figure}

Under these assumptions, we may apply the general results for self-similar fractals to deduce the possible complex dimensions of an $(n,r)$-von Koch snowflake and to compute explicit formulae for their relative tube functions. First, however, we observe a property in the lattice case regarding the poles of the scaling zeta function $\zeta_\Phinr$ which will have consequences for the complex dimensions and the explicit formulae of the corresponding $(n,r)$-von Koch snowflakes. 

\begin{proposition}[Simplicity of Poles in the Lattice Case]
    \label{prop:latticeSimplePolesGKF}
    Let $r\in(0,\frac13]$, let $\ell=(1-r)/2$, and let $n\geq 3$. If $n=4$, suppose further that $r<\frac13$.\footnote{Note that this is automatically satisfied when we impose the self-avoidance condition in Proposition~\ref{prop:selfAvoid}.} If $r$ and $\ell$ are arithmetically related, then the poles of the function 
    \[
        \zeta_{\Phi_{n,r}}(s) = \cfrac1{1-(n-1)r^{s}-2\ell^{s}}
    \]
    are all simple.  
\end{proposition}
\begin{proof}
    This follows from the proof of Proposition~1.3 (i) in \cite{vdB00_generalGKF} with the following modifications regarding the notation and the definition of the polynomial used therein. 

    Let $P(s)=\zeta_\Phinr(s)^{-1}$. A pole $\omega$ of $\zeta_\Phi$ occurs with multiplicity two or greater if we have that both $P(\omega)=0$ and $P'(\omega)=0$. In the lattice case, let $p$ and $q$ be such that $r=r_0^q$ and $\ell=r_0^p$ (so that $\log\ell/\log r=p/q$) and define $z=r_0^\omega$. When $p=q$, we have that $r=\ell=\frac13$ and thus $P(s)=1-(n+1)3^{-s}$ and $P'(s)=(n+1)3^{-s}\log 3$. It follows immediately that $P'(s)\neq 0$ by properties of the complex exponential, so we may assume that $p\neq q$.
    
    Solving now for when $P(\omega)=0$ and $P'(\omega)=0$, we obtain the system of equations:
    \begin{align*}
        1 &= (n-1)z^q+2z^p, \\
        0 &= (n-1)qz^q+2pz^p.
    \end{align*}
    Solving the system yields the relations that 
    \begin{align*}
        z^q &= \frac1{(n-1)(1-q/p)}, \\
        z^p &= \frac1{2(1-p/q)}.
    \end{align*}
    By considering the quantity $z^{-pq}$, we obtain the equations
    \begin{align*}
        (n-1)^p(1-q/p)^p &= 2^q(1-p/q)^q, \\
        (-1)^p q^q(n-1)^p(q-p)^p &= 2^q p^p (q-p)^q.
    \end{align*}
    Note that this last identity is equivalent to Equation~3.6 in \cite{vdB00_generalGKF}. There are no solutions when $q>p$. If $p<q$, then the only solution occurs when $n=4$, $p=3$, and $q=2$. Because we enforce that $r<\frac13$, we have that $r<\ell$ and thus that $\log\ell/\log r=p/q<1$, reaching the same contradiction as in \cite{vdB00_generalGKF}. 
\end{proof}

We now apply our results to deduce the possible complex dimensions of generalized von Koch fractals. We observe that since the relative tube function $V_{\Knr,\Omega}(t)=nV_{K,U}(t)$, its normalization satisfies the same scaling functional equation (with remainder also scaled by $n$). Consequently, it will have the same possible complex dimensions, so we state our result for the $(n,r)$-von Koch snowflake itself.

\begin{theorem}[Complex Dimensions of $(n,r)$-von Koch Snowflakes]
    \label{thm:cDimsOfGKFs}
    \index{Generalized von Koch fractal!Complex dimensions}
    Let $\Knr$ be an $(n,r)$-von Koch Snowflake boundary satisfying the self-avoidance criterion in Proposition~\ref{prop:selfAvoid} and let $\Omega$ be the interior region defined by $\Knr$ (i.e. the bounded component of $\RR^2\setminus\Knr$).
     \medskip 

    Then the complex dimensions $\Dd_\Knr(\HH_0)$ satisfy the containment 
    \[ \Dd_\Knr(\HH_0)\subseteq \Dd_\Phinr(\HH_0)=\set{\omega\in\CC \suchthat 1 = 2((1-r)/2)^\omega + (n-1)r^\omega}, \]
    where $\Dd_\Phinr(\HH_0)$ is the set of possible complex dimensions.
\end{theorem}
\begin{proof}
    By Corollary~\ref{cor:fullVonKochSFE}, we have that $(\Knr,\Omega)$ is an osculant fractal drum and that its relative tube function $\tubezeta_{\Knr,\Omega}$ satisfies a scaling functional equation of the form 
    \[ t^{-2}V_{\Knr,\Omega}(t) = L_\Phinr[t^{-2}V_{\Knr,\Omega}(t)](t) + t^{-2}V_{\Knr,R}(t), \]
    where $\Phinr$ is as in Equation~\ref{eqn:defGKCsystem} and where $t^{-2}V_{\Knr,R}(t)$ is continuous with estimate $t^{-2}V_{\Knr,R}(t)=O(t^0)$ as $t\to0^+$. 

    We may follow the same argument as in the proof of Theorem~\ref{thm:tubeZetaFormula}, namely applying Corollary~\ref{cor:structureOfPoles}, to deduce that the corresponding tube function satisfies the identity 
    \begin{equation}
        \label{eqn:zetaFormulaGKF}
        \tubezeta_{\Knr,\Omega}(s;\delta) = \zeta_\Phinr(s) (\xi_{n,r}(s;\delta)+\tubezeta_{R}(s;\delta))
    \end{equation}
    where $\tubezeta_R=\Mellin^\delta[t^{-2}V_{\Knr,R}(t)](s)$ and where $\xi_{n,r}:= \xi_{L_\Phinr,t^{-2}V_{\Knr,R}}$. 

    The result follows just as in Corollary~\ref{cor:possibleCdims}, with respect to the self-similar system $\Phinr$ and $\sigma_0=0$. 
\end{proof}

\begin{figure}[t]
    \centering
    \begin{tabular}{ccc}
        \includegraphics[width=3cm]{figures/GKFs/Knr/K_n=3,r=0.33,lv=6.pdf}
        &
        \includegraphics[width=3cm]{figures/GKFs/Knr/K_n=4,r=0.25,lv=6.pdf}
        &
        \includegraphics[width=3cm]{figures/GKFs/Knr/K_n=5,r=0.20,lv=6.pdf}
        \\
        \includegraphics[width=2.5cm]{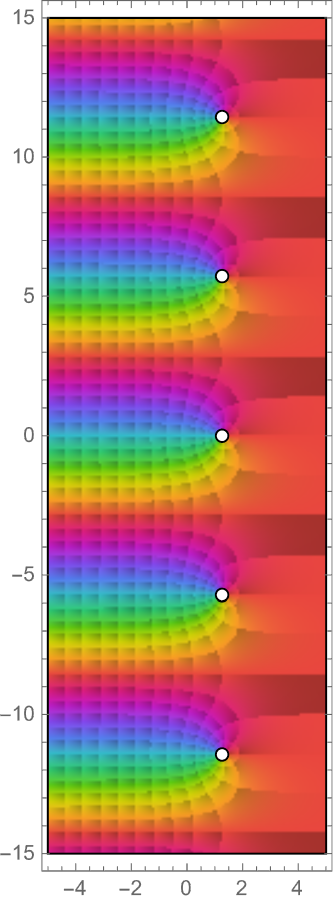}
        &
        \includegraphics[width=2.5cm]{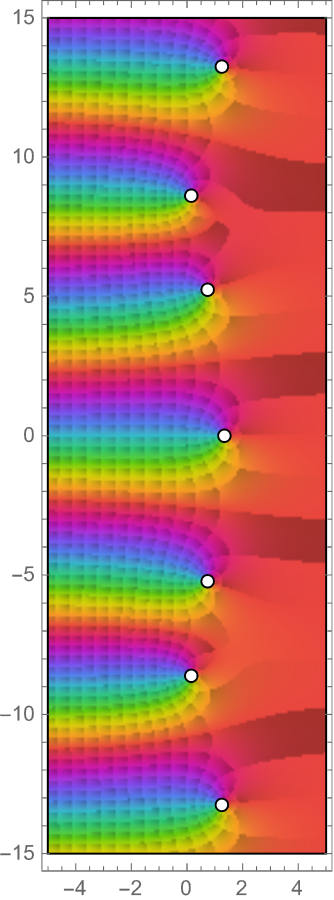}
        &
        \includegraphics[width=2.5cm]{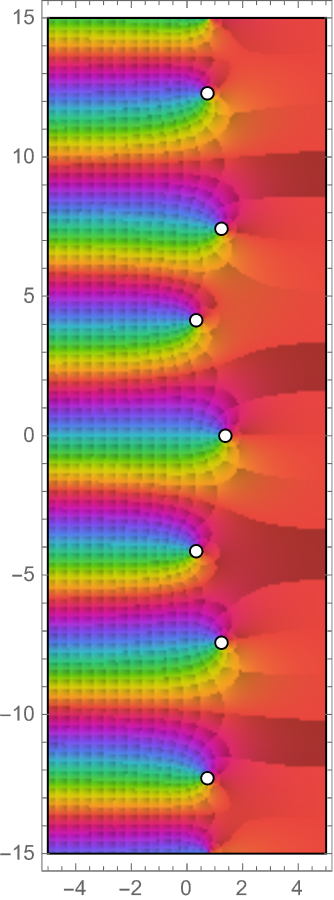} 
    \end{tabular}
    \caption[Complex plots of scaling zeta functions of generalized von Koch fractals]{Two dimensional complex argument plots of associated scaling zeta functions of the fractals $K_{3,\frac13}$ (left), $K_{4,\frac14}$ (middle), and $K_{5,\frac15}$ (right). The possible complex dimensions of the fractals occur at the poles of these functions, indicated in the plots by the white circles.}
    \label{fig:cDimPlots2D}
\end{figure}

Using this result, we may compute explicitly the possible complex dimensions of an $(n,r)$-von Koch snowflake using the corresponding (complexified) Moran equation alone. See Figure~\ref{fig:cDimPlots2D} for some plots of the scaling zeta functions $\zeta_\Phinr$, whose poles determine the possible complex dimensions, as well as the corresponding $(n,r)$-von Koch snowflakes.

\subsection{Tube Formulae for GKFs}
\label{sub:tubeFormulaeGKF}
%
%

In this subsection, we deduce explicit tube formulae for $(n,r)$-von Koch snowflakes. In what follows, we will assume that either $n\geq5$ or that the scaling ratios $r$ and $\ell=(1-r)/2$ are arithmetically related. We also impose the self-avoidance criterion (Proposition~\ref{prop:selfAvoid}) for $r$ as well. The latter case ensures that $\Knr$ has a connected interior region $\Omega$ which we use to construct its relative fractal drum $(\Knr,\Omega)$. The former restriction ensures that the scaling functional equations satisfied by the relative tube function $V_{\Knr,\Omega}$ have admissible remainders. 

For the preliminaries regarding the notation and conventions defining the antiderivatives $V_{\Knr,\Omega}^{[k]}$, the Pochhammer symbol $(z)_w$, $z,w\in\CC$, and the definition of the sums over complex dimensions as symmetric limits, we refer the reader to Section~\ref{sec:fractalSFEs}. 

\begin{theorem}[Pointwise Tube Formulae for GKFs]
    \label{thm:pointwiseTubeFormulaGKF}
    Let $\Knr$ be an $(n,r)$-von Koch Snowflake boundary satisfying the self-avoidance criterion in Proposition~\ref{prop:selfAvoid} and let $\Omega$ be the interior region defined by $\Knr$ (i.e. the bounded component of $\RR^2\setminus\Knr$). Suppose that either $n\geq5$ or that the scaling ratios $r$ and $\ell=(1-r)/2$ are arithmetically related (viz. the ratio of their logarithms is rational). 
    \medskip

    Then for every $k\geq 2$ in $\ZZ$, any $\delta>0$, and every $t\in(0,\delta)$, we have that
    \[ 
        V_{\Knr,\Omega}^{[k]}(t) = \sum_{\omega\in\Dd_\Phinr(\HH_{0})} 
            \Res\Bigg(\cfrac{t^{2-s+k}}{(2-s+1)_k}\tubezeta_{\Knr,\Omega}(s;\delta);\omega\Bigg)
            + \Rr^k(t),
    \]
    where $\tubezeta_{\Knr,\Omega}$ is given explicitly by Equation~\ref{eqn:zetaFormulaGKF}. For any $\e>0$ sufficiently small, the error term satisfies $\Rr^k(t)=O(t^{2-\e+k})$ as $t\to0^+$. 
\end{theorem}

\begin{proof}
    We have by Corollary~\ref{cor:fullVonKochSFE} that the tube function $V_{\Knr,\Omega}$ of the relative fractal drum $(\Knr,\Omega)$ satisfies the scaling functional equation
    \[ t^{-2}V_{\Knr,\Omega}(t) = L_\Phinr[t^{-2}V_{\Knr,\Omega}(t)](t) + t^{-2}V_{\Knr,R}(t) \]
    where $\Phinr$ is as in Equation~\ref{eqn:defGKCsystem}. The remainder term is continuous and satisfies the estimate $|V_{\Knr,R}(t)| \leq (2n\cot(\pi/n)+2\pi)\,t^2$ as $t\to0^+$.

    We note that $(\Knr,\Omega)$ is strictly speaking not an osculant fractal drum with respect to the self-similar system $\Phinr$ because this requires that $\Knr$ is the invariant set of $\Phinr$. Instead, each of the $n$ congruent \textit{components} of $\Knr$ is isometric to its attractor, an $(n,r)$-von Koch curve. However, because we have a scaling functional equation (established by linearity), the proof of Theorem~\ref{thm:pointwiseTubeFormula} is essentially unchanged as we apply Theorem~\ref{thm:pointwiseFormula} in the same manner. 

    As to the admissibility of the remainder with respect to the screens $S_\e(\tau)\equiv \e>0$, we apply Lemma~\ref{lem:lowerDimBoundsGKF} when $n\geq5$ to ensure that $\sigma_0=0<\lowersimdim(\Phinr)$ which is a sufficient condition per Theorem~\ref{thm:lowerDimAdmissibility}. In the lattice case, we use Theorem~\ref{thm:latticeCaseAdmissibility}.
\end{proof}

Note that in the lattice case, we have by Proposition~\ref{prop:latticeSimplePolesGKF} that the poles of $\zeta_\Phinr$ are simple. Thus, we may simplify the residues to obtain an expansion in powers of $t^{2-\omega+k}$. 
\begin{corollary}[Pointwise Tube Formulae for GKFs, Lattice Case]
    \label{cor:pointwiseTubeFormulaGKFLattice}
    Let $\Knr$ be an $(n,r)$-von Koch Snowflake boundary satisfying the self-avoidance criterion in Proposition~\ref{prop:selfAvoid} and let $\Omega$ be the interior region defined by $\Knr$ (i.e. the bounded component of $\RR^2\setminus\Knr$). Suppose that its scaling ratios $r$ and $\ell=(1-r)/2$ are arithmetically related (i.e. the ratio of their logarithms is rational).
    \medskip 

    Then for every $k\geq 2$ in $\ZZ$, any $\delta>0$, and every $t\in(0,\delta)$, we have that
    \[ 
        V_{\Knr,\Omega}^{[k]}(t) = \sum_{\omega\in\Dd_\Phinr(\HH_{0})} 
            \cfrac{r_\omega}{(2-\omega+1)_k}\,t^{2-\omega+k}
            + \Rr^k(t),
    \]
    where $r_\omega$ is a constant given by $r_\omega:=\Res\enclose{\tubezeta_{\Knr,\Omega}(s;\delta);\omega}$, with $\tubezeta_{\Knr,\Omega}$ given explicitly by Equation~\ref{eqn:zetaFormulaGKF}. For any $\e>0$ sufficiently small, the error term satisfies $\Rr^k(t)=O(t^{2-\e+k})$ as $t\to0^+$.
\end{corollary}

We note that the shift by one in the Pochhammer symbol $(z)_k$ can be understood in a straightforward manner. If we use its expansion as a product when $k$ is a positive integer and write its terms in reverse order, we have that 
\[ (2-\omega+1)_k = (2-\omega+k)(2-\omega+k-1)\cdots(2-\omega+1). \]
So, this first term corresponds to the power $t^{2-\omega+k}$. 

Lastly, we stress that these sums over the complex dimensions $\Dd_\Phinr(\HH_0)$ may have delicate convergence properties, as they must be defined as symmetric limits. The use of integration, leading to pointwise formulae for the \text{antiderivatives} of the relevant tube functions, leads to increased regularity. If we wish to view the formula for the tube function itself, when $k=0$, we may use the distributional formulation in order to do so. We will state the formula when $k=0$ and refer the reader to Theorem~\ref{thm:distTubeFormula} for the formulae which are valid for any $k\in\ZZ$. 

\begin{theorem}[Distributional Tube Formula for GKFs, $k=0$]
    \label{thm:distTubeFormulaGKF}
    Let $\Knr$ be an $(n,r)$-von Koch Snowflake boundary satisfying the self-avoidance criterion in Proposition~\ref{prop:selfAvoid} and let $\Omega$ be the interior region defined by $\Knr$ (i.e. the bounded component of $\RR^2\setminus\Knr$). Suppose that either $n\geq5$ or that the scaling ratios $r$ and $\ell=(1-r)/2$ are arithmetically related (i.e. the ratio of their logarithms is rational). 
    \medskip

    Then for any $\delta>0$ and every $t\in(0,\delta)$, we have that as an equality of distributions in the Schwartz space $\Ss'(0,\delta)$ (the dual of the space defined in Equation~\ref{eqn:defSchwartzFns}),
    \[ 
        V_{\Knr,\Omega}(t) = \sum_{\omega\in\Dd_\Phinr(\HH_{0})} 
            \Res\Bigg(t^{2-s}\tubezeta_{\Knr,\Omega}(s;\delta);\omega\Bigg)
            + \Rr(t),
    \]
    where $\tubezeta_{\Knr,\Omega}$ is given explicitly by Equation~\ref{eqn:zetaFormulaGKF}. For any $\e>0$ sufficiently small, the error term satisfies $\Rr(t)=O(t^{2-\e})$ as $t\to0^+$ in the sense of Equation~\ref{eqn:distRemEst2}. The action of $V_{\Knr,\Omega}$ on a test function is given explicitly by Equation~\ref{eqn:bracketIdentity2} with $k=0$.
\end{theorem}
\begin{proof}
    This is a direct corollary of Theorem~\ref{thm:distTubeFormula} with the same minor modification as in the proof of Theorem~\ref{thm:pointwiseTubeFormulaGKF}. Alternatively, one may first apply Theorem~\ref{thm:distTubeFormula} to the osculant fractal drum $(K,U)$ with corresponding self-similar system $\Phinr$ as in Theorem~\ref{thm:vonKochSFE}. Then, one deduces the formula here using the relation that $V_{\Knr,\omega}(t)=nV_{K,U}(t)$. By linearity of Mellin transforms, it follows that $\tubezeta_{\Knr,\Omega}(s;\delta)=n\tubezeta_{K,U}(s;\delta)$. Similar holds for the remainder term and its Mellin transform. Lastly, the result follows by the linearity of limits (namely, the symmetric limit defining the summations), of the finite summations within this limit, and of residues (which themselves may be seen as either limits or equivalently as contour integrals).  
\end{proof}

We conclude with the same corollary which is true in the lattice case, namely that the residues are constants computed directly from the tube zeta function of the relative fractal drum. 
\begin{corollary}[Distributional Tube Formula for GKFs, $k=0$]
    \label{cor:distTubeFormulaGKFLattice}
    Let $\Knr$ be an $(n,r)$-von Koch Snowflake boundary satisfying the self-avoidance criterion in Proposition~\ref{prop:selfAvoid} and let $\Omega$ be the interior region defined by $\Knr$ (i.e. the bounded component of $\RR^2\setminus\Knr$). Suppose that the scaling ratios $r$ and $\ell=(1-r)/2$ are arithmetically related (i.e. the ratio of their logarithms is rational). 
    \medskip

    Then for any $\delta>0$ and every $t\in(0,\delta)$, we have that as an equality of distributions in the Schwartz space $\Ss'(0,\delta)$ (the dual of the space defined in Equation~\ref{eqn:defSchwartzFns}),
    \[ 
        V_{\Knr,\Omega}(t) = \sum_{\omega\in\Dd_\Phinr(\HH_{0})} 
            r_\omega\,t^{2-\omega}
            + \Rr(t),
    \]
    where $r_\omega$ is a constant given by $r_\omega:=\Res\enclose{\tubezeta_{\Knr,\Omega}(s;\delta);\omega}$, with $\tubezeta_{\Knr,\Omega}$ given explicitly by Equation~\ref{eqn:zetaFormulaGKF}. For any $\e>0$ sufficiently small, the error term satisfies $\Rr(t)=O(t^{2-\e})$ as $t\to0^+$ in the sense of Equation~\ref{eqn:distRemEst2}. The action of $V_{\Knr,\Omega}$ on a test function is given explicitly by Equation~\ref{eqn:bracketIdentity2} with $k=0$.
\end{corollary}

\chapter{Heat Content of Fractals}
\label{chap:heat}
%
%


In this chapter, we study the \textit{heat content} of regions with self-similar fractal boundaries. On such regions, we consider the \textit{Dirichlet problem for the heat equation} and the total heat within the region as a function of time. The boundary conditions we consider are a model for heat flow into a region, with zero initial temperature inside the region and constant temperature on the boundary.

Our main goal is to demonstrate the role of complex dimensions, which describe the geometric oscillations within a fractal, in this spectral-adjacent context. Our main result is to produce explicit formulae for the heat content, the integral of the temperature function, for small values of the time parameter. This expansion is indexed by the possible complex dimensions of the fractal boundary, with these dimensions controlling the exponents which appear in the expansion. These results constitute partial progress toward understanding the relationship between geometry and spectrum of self-similar fractals in higher dimensions.

Our setting will be the heat equation on a region $\Region$ whose boundary $\Regionbd$ is the attractor of a self-similar system (or the union of such regions placed symmetrically, in the case of generalized von Koch fractals). Given a decomposition of the heat content into the images of $\Region$ under the mappings of the self-similar system, as well as the scaling property of the heat content, we apply the results of Chapter~\ref{chap:SFEs} to deduce scaling functional equations and explicit formulae for the heat content. Using the results of \cite{vdB00_generalGKF}, we produce explicit results for generalized von Koch fractals. 

\section{Perron-Wiener-Brelot Solutions to the Heat Equation}
\label{sec:heatPWBsoln}
%
%

Let $\Region\subset\RR^\dimension$ be a bounded open set. The heat equation with Dirichlet boundary conditions on $\Region$ is the second order, parabolic differential equation given by 
\index{Heat equation}
\begin{equation}
    \label{prob:generalHeatProblem}
    \left\{
    \begin{aligned}
        \partial_t u - C\Delta u &= 0   &&\text{in }\Region\times[0,\infty) \\
            u &= f                      &&\text{on }\Region\times\set{0} \\
            u &= g                      &&\text{on }\partial\Region\times (0,\infty).
    \end{aligned}
    \right.
\end{equation}
Here, $C$ is a positive constant called the diffusivity constant. Here, we consider the Dirichlet problem on a cylinder set, $\Region\times[0,\infty)$, where the spatial region $\Region$ is fixed with respect to the time parameter, for simplicity. However, in order to understand solutions to the heat equation when the boundary of $\Region$ is arbitrary, most notably fractal, we will study the theory of the heat equation on an arbitrary open set $\Regionplus$ and then specialize to this case to establish scaling functional equations. 

To understand the properties of solutions to Problem~\ref{prob:generalHeatProblem}, we must discuss the nature of what a solution means first. In particular, we will discuss the notion of a Perron-Wiener-Brelot (PWB) solution, which is suitable for extending the classical theory to arbitrary open sets $\Regionplus\subset\RRNplus$ with boundary conditions imposed on $\partial\Regionplus$. While we will later restrict our attention to cylindrical sets when imposing specific boundary conditions, the general theory of PWB solutions does not have this constraint. The advantage of the PWB solution (over, say, a weak solution existing in some corresponding Sobolev space) is that it yields a function which is pointwise defined everywhere in the interior, which is once (in time) or twice (in spatial coordinates) continuously differentiable, and which converges to the prescribed boundary value where possible. 

\subsection{History of the PWB Method for the Heat Equation}
%
%

Perron introduced his famous method for solving the Dirichlet problem for Laplace's equation \cite{Per23}. The basic premise is to obtain the solution, a harmonic function, as a supremum over a family of subharmonic functions. This method was studied by Wiener, who in particular defined a generalized solution as one obtained in this way with limits converging to the prescribed boundary values at regular boundary points and proved a criterion for the convergence of such solutions to the prescribed boundary value conditions \cite{Wie24_Dirichlet,Wie24_Potential}. 

Wiener's criterion established existence of such generalized solutions, and the uniqueness of such solutions (namely, that such a solution is identically zero when all its limits approaching regular boundary points are zero) was established by Kellogg \cite{Kel28} and Evans \cite{Eva33}. Brelot's contributions include simplifying this proof \cite{Bre55} as well as developing a general theory and extending the method (c.f. \cite{Bre67}). His general study in potential theory was already inspired by others, notably including the work of Doob \cite{Doo66}. The potential theory would be further developed and axiomatized by Bauer \cite{Bau66} and Constantinescu and Cornea \cite{CC72}.

The application of the PWB method to the heat equation relates to the probabilistic approach, owing to the connection of this equation to Brownian motion. This connection was discovered independently by Einstein \cite{Ein56} and Smoluchowski \cite{Smo06}. However, some of these same ideas were discovered as early as the work of Bachelier, who studied mathematical finance \cite{Bac00}. Einstein's work especially was a large inspiration for Wiener's mathematical development of the theory of Brownian motion \cite{Wie23,Wie24_Average,Wie76}. It is for this reason that Brownian motion may be instead called a Wiener process. 

The analysis of the heat equation on an arbitrary bounded open set in Euclidean space, using the probabilistic approach, is due to Doob \cite{Doo55}. He obtained the solution in the sense of Perron and Wiener, using the Martingale analogues of subharmonic functions, though he did not provide details for the classification of regular and irregular boundary points. Evans and Gariepy proved an analogue of Wiener's criterion for the classification of regular and irregular boundary points for solutions of the heat equation \cite{EG82}. Watson would later show that Doob's approach is equivalent to the later developments in the theory, including filling in details not originally provided by Doob \cite{Wat15}. For a modern reference, we recommend Watson's book on the subject \cite{Wat12}. 

Lastly, we note that the Dirichlet problem for the heat equation requires some extra assumptions in order to be well-posed. (A Dirichlet problem is well-posed if there exists a unique solution which depends continuously on initial conditions. We shall only need concern ourselves with existence and uniqueness in this work, though, having fixed constant boundary conditions.) As was discovered by Tychonoff, functions with growth larger than quadratic exponentials can yield nonunique solutions, but under the restriction to this growth condition solutions are unique \cite{Tyc35}. Widder further analyzed the problem and showed that nonnegative solutions are unique \cite{Wid44}. For more information about the (non)uniqueness, we recommend Widder's paper. 

\subsection{Classes of Temperature Functions}
%
%

In order to define PWB solutions, we must define a class of \textit{hypertemperatures} (and respectively \textit{hypotemperatures}) to be considered. Let $\Regionplus\subseteq\RRNplus$ be an arbitrary open set and we write $p=(x,t)\in\RR^\dimension\times\RR$ to distinguish the \textit{spatial coordinates} $x$ and the \textit{time coordinate} $t$. Let $C^{2,1}(\Regionplus)$ denote the class of continuously differentiable functions on $\Regionplus$ with continuous second partial derivatives with respect to the spatial coordinates (each $x_k$, $k=1,...,\dimension$) and one continuous derivative with respect to the time coordinate $t$. 

In the simplest case, if we consider $u\in C^{2,1}(\Regionplus)$, then we may define the heat operator\footnote{Note that we take the opposite convention as that in \cite{Wat12}, where the heat operator therein is $\Theta:=-\heatop$.} $\heatop$ by $\heatop u:= \partial_t u -\Delta u$, in which case the heat equation is simply $\heatop u =0$. A \index{Temperature}\textbf{temperature} is a function which is in $C^{2,1}(\Regionplus)$ and which solves the heat equation, i.e. $\heatop u=0$. 

A \index{Temperature!Subtemperature}\textbf{subtemperature} is then a function $u\in C^{2,1}(\Regionplus)$ such that $\heatop u\leq 0$. A \index{Temperature!Supertemperature}\textbf{supertemperature} is a function such that $\heatop u\geq 0$. Note that a function $u$ is a supertemperature if and only if $-u$ is a subtemperature, and in general we need only define one class of functions and we may use this relation to define the other. We see already the theme of the PWB approach: a temperature function is both a subtemperature and a supertemperature. 

However, for the study of the heat problem on an arbitrary open set, we will need to consider functions which are not necessarily continuously differentiable. Thus, we will need an analogue of the heat operator that does not require differentiability, as we cannot define $\heatop$ for such functions to consider sub/supertemperatures. There are several equivalent approaches. For instance, one may consider a mean value operator over an appropriate type of set (so-called \textit{heat balls}) and look for functions which agree with their mean values. This is a fundamental property of solutions to the classical heat equation, and can be used to show that such functions obey the heat equation. One may also use a classification in terms of a type of integral called a \textit{Poisson integral} which involves types of measures called \textit{parabolic measures}. 

Let us presently take the mean value approach. Let $p_0=(x_0,t_0)$ and $r>0$. A \index{Ball!Heat ball}\textbf{heat ball} is a set of the form 
\[ 
    E_{p_0,r} := \set{p=(x,t)\in\Regionplus \suchthat 
        |x-x_0| \leq \sqrt{2\dimension(t_0-t)\log(r/(t_0-t))},\, 0<t_0-r<t<t_0}.
\]
(One may represent this set in terms of a lower bound for the \textit{fundamental solution} to the heat equation, hence the name.) We say that a \index{Heat sphere}\textbf{heat sphere} is the boundary $\partial E_{p_0,r}$ of a heat ball.

In order to avoid the singular behavior which occurs near the center $p_0$ of a heat ball, we will introduce the function 
\begin{equation}
    \label{eqn:cutoffKernel}
    Q(p) = Q(x,t) := \begin{cases}
        \cfrac{|x|^2}{\strut \sqrt{ 4|x|^2|t|^2 + (|x|^2 - 2\dimension t)^2 } } & (x,t)\neq 0, \\
        1 & (x,t)=(0,0).
    \end{cases}
\end{equation}
Note that $Q$ has the property that for any sequence of points $p\in\partial E_{p_0,r}$ converging to $p_0$, 
$Q(p) \to 1 $, i.e. it is continuous on $\partial E_{p_0,r}$. 

\begin{definition}[Heat Sphere Mean Value Operator]
    \label{def:heatMeanValueOp}
    \index{Mean value operator!on a heat sphere}
    Letting $p_0=(x_0,t_0)\in\RR^\dimension\times\RR$ and $r>0$, consider the heat ball $E_{p_0,r}$. Let $dS$ denote the surface area measure on $\partial E_{p_0,r}$. Then for any function $u$ for which the integral converges, we define the \textbf{mean value operator} $\heatmean_{p_0,r}$ on $\partial E_{p_0,r}$ to be 
    \begin{equation}
        \label{eqn:defHeatMean}
        \heatmean_{p_0,r}[u]=\heatmean[u;\partial E_{p_0,r}]:= \frac1{(4\pi)^{\dimension/2}} \int_{\partial E_{p_0,r}} Q(p-p_0) u(p) dS(p).
    \end{equation}
\end{definition}

The last preliminary is to define the properties that our functions must satisfy in lieu of continuous differentiability. We first define the notions of upper and lower \textbf{semicontinuity}: a function $f:\RRNplus\to[-\infty,\infty]$ is \index{Semicontinuity!Upper semicontinuity}\textbf{upper semicontinuous} on a set $\Regionplus\subseteq \RRNplus$ if for any $r\in\RR$, the set $\set{p\in\Regionplus \suchthat f(p)<a}$ is a relatively open subset of $S$ (i.e. open in the subspace topology it inherits from $\RRNplus$). We say that an extended real-valued function $f$ is \index{Semicontinuity!Lower semicontinuity}\textbf{lower semicontinuous} if $-f$ is upper semicontinuous. 

Next, we say that an extended real-valued function $f:\RRNplus\to[-\infty,\infty]$ is \index{Finite function!Upper finite function}\textbf{upper finite} on a set $\Regionplus\subseteq\RRNplus$ if $f[\Regionplus]\subseteq [-\infty,\infty)$ and \index{Finite function!Lower finite function}\textbf{lower finite} on $\Regionplus$ if $f[\Regionplus]\subset(-\infty,\infty]$. Note that when $f$ is both upper and lower finite, it is simply a real-valued function. 

We may now define hypertemperatures and hypertemperatures. 
\begin{definition}[Hypertemperatures]
    \label{def:hypertemperature}
    \index{Temperature!Hypertemperature}\index{Temperature!Hypotemperature}
    Let $\Regionplus\subseteq\RRNplus$. We say that an extended real-valued function $f$ is a \textbf{hypertemperature} on $\Regionplus$ if $f$ is \textit{lower finite}, \textit{lower semicontinuous}, and with the following property. For any $p\in\Regionplus$ and any $\e>0$, there exists $r<\e$ such that 
    \[ f(p)\geq \heatmean[f;\partial E_{p,r}]. \]
    Let $\hyperspace(\Regionplus)$ denote the set of all hypertemperatures on $\Regionplus$.
\end{definition}

The efficient way to define when a function $f$ is a hypotemperature is to require that $-f$ be a hypertemperature. However, for the sake of completeness and establishing notation, we will give an independent definition and let this fact be an immediate corollary. 
\begin{definition}[Hypotemperatures]
    \label{def:hypotemperature}
    \index{Temperature!Hypertemperature}\index{Temperature!Hypotemperature}
    Let $\Regionplus\subseteq\RRNplus$. We say that an extended real-valued function $f$ is a \textbf{hypotemperature} on $\Regionplus$ if $f$ is \textit{upper finite}, \textit{upper semicontinuous}, and with the following property. For any $p\in\Regionplus$ and any $\e>0$, there exists $r<\e$ such that 
    \[ f(p)\leq \heatmean[f;\partial E_{p,r}]. \]
    Let $\hypospace(\Regionplus)$ denote the space of all hypotemperatures on $\Regionplus$. 
\end{definition}

Ultimately, when we look for solutions of the heat equation, we will consider the collection of all hypertemperatures and hypotemperatures that satisfy the prescribed boundary conditions. By taking an infimum of hypertemperatures or a supremum of hypotemperatures, then ideally we arrive at the same function which now satisfies the boundary conditions and will be a \text{temperature}, satisfying the heat equation proper. 

\subsection{Classification of Boundary Points}
%
%

In what follows, let $\Regionplus\subseteq\RRNplus$ be an arbitrary open set. We consider the boundary of $\Regionplus$ relative to the one-point compactification of $\RRNplus$, so that $\infty\in\partial\Regionplus$ exactly when $\Regionplus$ is unbounded. With this proviso, $\partial\Regionplus$ is otherwise the standard Euclidean boundary of $\Regionplus$. For an arbitrary open set $\Regionplus$, it turns out that we cannot prescribe boundary conditions on all of $\partial\Regionplus$. This is to be \textit{reasonably expected}, however. 

Consider the simple problem of a one-dimensional rod given by the interval $[a,b]$ and suppose that we wish to consider the problem for time $t\in[0,1]$. Suppose now to look for solutions of the heat equation $\partial_x^2u=\partial_tu$ in the set $(a,b)\times(0,1)$. Prescribing conditions on the boundary of $[a,b]\times[0,1]$ amounts to three types of constraints. On the set $(a,b)\times\set{0}$, we prescribe the initial amount of heat on the inside of the rod at time $t=0$. On each of the sets $\set{a}\times[0,1]$ and $\set{b}\times[0,1]$, we are prescribing the temperature values of the endpoints of the rod. Lastly, on the boundary piece $(a,b)\times\set{1}$, we are prescribing the {future} temperature of the inside of rod at time one. However, in this deterministic model, the conditions at time $t=0$ and on the boundary for all time $t\in[0,1]$ will determine the future temperature; it cannot be arbitrarily chosen. 

This leads us to a general classification of the types of boundary points of $\partial\Regionplus$. In order to do so, let us establish some notation regarding half-balls with respect to the time coordinate. Letting $p_0=(x_0,t_0)\in\RRNplus$, we write 
\[ H_*(p_0,r) = \set{(x,t)\suchthat |x-x_0|^2+|t-t_0|^2<r^2,\, t<t_0} \] 
for the open \index{Ball!Lower half-ball}\textbf{lower half-ball}. Similarly, we write $H^*(p_0,r)$ for the open \index{Ball!Upper half-ball}\textbf{upper half-ball}, defined analogously but with $t>t_0$. Note that these are standard balls cut in half, not heat balls. 

The boundary $\partial\Regionplus$ is classified into three types of points: \textit{normal}, \textit{semi-singular} (or \textit{semi-normal}), and \textit{(fully) singular points}. If $q\in\partial\Regionplus$, then $q$ is said to be \index{Boundary point!Normal boundary point}\textbf{normal} if either $q=\infty$ is the point at infinity or if every lower half-ball $H_*(q,r)$ centered at $q$ intersects the complement of $\Regionplus$, viz. $H_*(q,r)\setminus\Regionplus\neq\emptyset$ for all $r>0$, namely as $r\to0^+$. For abnormal boundary points $q$, eventually there exists an $r_0>0$ so that $H_*(q,r)\subseteq\Regionplus$. If there exists $r_1<r_0$ such that $H^*(q,r_1)\cap \Regionplus=\emptyset$, then $q$ is a \index{Boundary point!Fully singular boundary point}{(fully) singular} boundary point. (Note that this implies that any smaller upper half-ball $H^*(p,r_2)$, $r_2<r_1$, will also be contained in the complement of $Y$.) Otherwise, for every $r<r_0$ we have that $H^*(q,r_1)\cap \Regionplus\neq\emptyset$ and $q$ is said to be \index{Boundary point!Semi-singular boundary point}\textbf{semi-singular}.

Returning to the example from earlier, we can see now how these criteria describe the different parts of the boundary. The initial condition set, $(a,b)\times\set{0}$, clearly contains only normal points since the lower half-balls $H_*(p,r)$ are contained in the complement for any $r$, let alone merely intersecting the complement. On the edges of the rod, we also have that arbitrarily small lower half-balls always intersect the complement to the left or right, respectively. Thus this part of the boundary is also normal. 

The top part of the boundary, however, is fully singular: the lower half-balls are strictly contained in the set $\Regionplus$ and the upper half-balls are strictly contained in the complement. For these fully singular points, they can only be approached from within the set $\Regionplus$ and only from past times, not future times. The conditions above amount to requiring that fully singular boundary points look locally like such points. 

This example does not illustrate a semi-singular point, but for a simple example puncture the square $(a,b)\times(0,1)$ at an interior point, say $(x_0,t_0)$. This new point in the boundary is then semi-singular. Imposing a boundary condition here amounts to instantaneously changing the temperature on the rod at $x_0$ at the time $t_0$. In this case, one could expect a solution for times $t>t_0$ which are compatible with the current values of the temperature at $t_0$ for $x\neq x_0$ and the new prescribed temperature at $x_0$. However, the solution in the larger set will not approach this arbitrarily changed value from times $t<t_0$. 

So, for the Perron-Wiener-Brelot solutions to follow, we will impose boundary conditions only on the set of normal and semi-singular boundary points. For this reason, we say that the \index{Boundary!Essential boundary}\textbf{essential boundary} $\partial_e\Regionplus$ of a set $\Regionplus\subseteq\RRNplus$ is the union of the set of all \textit{normal} boundary points, the \index{Boundary!Normal boundary}\textbf{normal boundary} $\partial_n\Regionplus$, and the set of all \textit{semi-singular} boundary points, the \index{Boundary!Semi-singular boundary}\textbf{semi-singular boundary} $\partial_{ss}\Regionplus$. If we denote the set of \textit{fully singular} boundary points by $\partial_s\Regionplus$, called the \index{Boundary!Singular boundary}\textbf{singular boundary}, then we have that $\partial_e\Regionplus=\partial\Regionplus\setminus\partial_s\Regionplus$. Note also that the normal boundary and the semi-singular boundary are disjoint by the definitions of normal and semi-singular boundary points. Recall that $\partial\Regionplus$ is considered with respect to the one-point compactification of $\RRNplus$ and that $\infty\in\partial_n\Regionplus$ is always a normal boundary point when $\Regionplus$ is unbounded.

We now may define the sense in which boundary conditions will be satisfied. We will require convergence of the function to the boundary function at normal points to be for any sequence in $\Regionplus$ converging to the boundary points. Meanwhile, for convergence to the value of the function at semi-singular boundary points, we will only impose convergence when the sequence approaches from future times, not past times.
\begin{definition}[Regularity and Convergence on the Essential Boundary]
    \label{def:convergenceOnEssenBd}
    \index{Boundary!Essential boundary convergence}
    \index{Boundary!Regular (essential) boundary}
    Let $\Regionplus\subseteq\RRNplus$ be an open set and let $\partial_e\Regionplus=\partial_n\Regionplus\uplus\partial_{ss}\Regionplus$ be its essential boundary. Let $f$ be a function on $\partial\Regionplus$ and $u$ a function on $\Regionplus$ itself. 
    \medskip 

    We say that $u$ \textbf{converges to} $f$ at a point $q\in\partial_e\Regionplus$ if we have the following.
    \begin{itemize}
        \item If $q\in\partial_n\Regionplus$, then we have that for any sequence $\set{p_n}_{n\in\NN}\subset\Regionplus$ with $p_n\to q$, $u(p_n)\to f(q)$. 
        \item If $q=(y,s)\in\partial_{ss}\Regionplus$, then we have that for any sequence $\set{p_n=(x_n,t_n)}_{n\in\NN}\subset\Regionplus$ with $x_n\to y$ and $t_n\to s^+$, $u(p_n)\to f(q)$.
    \end{itemize}
    When $q\in \partial_e\Regionplus$ is such that these limits exist, $q$ is called a \index{Boundary point!Regular boundary point}\textbf{regular boundary point}. We say that $u$ \textbf{converges to} $f$ \textbf{on the regular (essential) boundary} if $u$ converges to $f$ at every regular point $q$ in the essential boundary.
\end{definition}

The last caveat is that given a set $\Regionplus$, it can have both \textit{regular} and \textit{irregular} boundary points in its essential boundary. (A boundary point is called \index{Boundary point!Irregular boundary point}\textbf{irregular} if it is not regular in the sense above.) So, the final task is to determine which points in the essential boundary are regular. In the case that the entire essential boundary is not regular, i.e. there exist irregular boundary points, then one must seek to understand which points are regular and which are irregular. We will not explore this issue further here, but refer the reader to the work of Evans and Gariepy \cite{EG82} or to Section~8.6 of \cite{Wat12} for more information, including sufficient conditions to ensure that an entire boundary or a specific point in the boundary is regular. 

\subsection{Perron-Wiener-Brelot Solutions}
%
%

We now have the means to define a Perron-Wiener-Brelot solution to the heat equation on an arbitrary open set. Consider the following Dirichlet problem for the heat equation on an arbitrary open set $\Regionplus\subseteq\RRNplus$. 
\index{Heat equation}
\begin{equation}
    \label{prob:arbitraryHeatProb}
    \left\{
    \begin{aligned}
        \partial_t u - C\Delta u &= 0   &&\text{in }\Regionplus, \\
            u &= f                      &&\text{on }\partial\Regionplus, \\
    \end{aligned}
    \right.
\end{equation}
Here, $C>0$ is a positive constant. 

We consider now two collections of functions which we define based on our boundary function $f$. These are either hypertemperatures bounded from below by $f$ or hypotemperatures bounded from above by $f$.
\begin{align*}
    \hyperspace(\Regionplus;f) &:= \set{w\in\hyperspace(\Regionplus)\suchthat 
        \begin{aligned}
            \liminf_{(x,t)\to (y,s)}  w(x,t) &\geq f(y,s) 
                & (x,t)\in\Regionplus,\, (y,s)\in\partial_n\Regionplus \\
            \liminf_{(x,t)\to(y,s^+)} w(x,t) &\geq f(y,s) 
                & (x,t)\in\Regionplus,\, (y,s)\in\partial_{ss}\Regionplus
        \end{aligned}
    }, \\
    \hypospace(\Regionplus;f) &:= \set{w\in\hyperspace(\Regionplus)\suchthat 
        \begin{aligned}
            \limsup_{(x,t)\to (y,s)}  w(x,t) &\leq f(y,s) 
                & (x,t)\in\Regionplus,\, (y,s)\in\partial_n\Regionplus \\
            \limsup_{(x,t)\to(y,s^+)} w(x,t) &\leq f(y,s) 
                & (x,t)\in\Regionplus,\, (y,s)\in\partial_{ss}\Regionplus
        \end{aligned}
    }.
\end{align*}
These spaces of functions are each nonempty since $w\equiv\infty$ is in the first and $w\equiv-\infty$ is in the second.

We next define the upper and lower candidates for the solution to the heat equation. 
\begin{align}
    \label{eqn:PWBcandidates}
    \begin{split}
        \wupper(p)  &:= \inf\set{w(p) \suchthat w\in\hyperspace(\Regionplus;f)}, \\
        \wlower(p)  &:= \sup\set{w(p) \suchthat w\in\hypospace(\Regionplus;f)}.
    \end{split}
\end{align}
These define two functions $\wupper$ and $\wlower$ on $\Regionplus$. If they agree (and are sufficiently regular), then the function $f$ is said to be \textit{resolutive}.

\begin{definition}[Resolutive Function]
    \label{def:resolutiveFn}
    \index{Resolutive function}
    \index{Heat equation!Resolutive boundary conditions}
    A function $f$ on $\partial_e\Regionplus$, the essential boundary of $\Regionplus\subseteq\RRNplus$, is said to be \textbf{resolutive} if the following hold. Firstly, the functions defined by Equation~\ref{eqn:PWBcandidates} are the same, i.e. $\wupper\equiv\wlower$, and in this case, we define $\wsoln:=\wupper=\wlower$. Secondly, we require that $\wsoln\in C^{2,1}(\Regionplus)$. 
\end{definition}

When the boundary function $f$ is resolutive, we obtain a \textit{temperature} function $\wsoln$, as it follows that since $\wsoln\in C^{2,1}(\Regionplus)$, $\heatop\wsoln=0$. From the definitions of the function spaces $\hyperspace(\Regionplus,f)$ and $\hypospace(\Regionplus,f)$, note that whenever the limit of $\wsoln$ as $p$ approaches a point $q$ in the essential boundary $\partial_e\Regionplus$ exists, the value of the limit must necessarily equal $f(q)$ since $\wsoln\in\hyperspace(\Regionplus,f)\cap \hypospace(\Regionplus,f)$. Therefore, $\wsoln$ converges to $f$ on the regular essential boundary of $\Regionplus$. This is exactly what it means for $\wsoln$ to be a solution to a Dirichlet problem for the heat equation in the sense of in the sense of Perron-Wiener-Brelot. 
\begin{proposition}[Perron-Wiener-Brelot (PWB) Solution]
    \label{def:PWBsolution}
    \index{Heat equation!Perron-Wiener-Brelot solution}
    Let $\Regionplus\subseteq\RRNplus$ and suppose that $f$ is a \textit{resolutive} function in the sense of Definition~\ref{def:resolutiveFn}. Then we say that $\wsoln$ is the \textbf{Perron-Wiener-Brelot solution} to Problem~\ref{prob:arbitraryHeatProb} and we have that $\wsoln$ converges to $f$ on the regular essential boundary of $\Regionplus$ in the sense of Definition~\ref{def:convergenceOnEssenBd}.
\end{proposition}
\subsection{Parabolic Measure and Resolutive Functions}
%
%

The final ingredient in our study of PWB solutions is the last but most important remaining question. Given an open set $\Regionplus\subseteq\RRNplus$, when is a function $f$ on $\partial_e\Regionplus$ resolutive? Per Proposition~\ref{def:PWBsolution}, this condition is exactly what is required for a PWB solution to exist with the Dirichlet boundary conditions specified by $f$. 

Luckily, such functions are plentiful. In fact, any continuous function $f\in C^0(\partial_e\Regionplus)$ on the essential boundary is resolutive; see for instance Theorem~8.26 in \cite{Wat12}. We can give a characterization of the class of resolutive functions in terms of an integrability condition relating to a family of measures called \textit{parabolic measures}. 

Given the open set $\Regionplus\subseteq\RRNplus$ and a point $p\in\Regionplus$, it can be shown (see Theorem~8.27 in \cite{Wat12}) that there exists a unique nonnegative Borel probability measure $\parameasure$ on $\partial_e\Regionplus$ such that for any continuous function $f\in C^0(\partial_e\Regionplus)$, 
\begin{equation}
    \label{eqn:defPoissonInt}
    \wsoln(p) = \int_{\partial_e\Regionplus} f\,d\parameasure.
\end{equation}
An integral of this form may be called a \textit{Poisson integral} owing to their connection to the classical theory of Poisson kernels and integral representations (in the setting of harmonic functions which can be extended to the heat equation). 

\begin{definition}[Parabolic Measure and Parabolic Integrability]
    \label{def:paraMeasure}
    \index{Parabolic measure}
    \index{Measure!Parabolic measure}
    \index{Parabolic measure!Parabolical integrability}
    The completion of the measure defined by Equation~\ref{eqn:defPoissonInt}, which we shall also denote by $\parameasure$, is called the \textbf{parabolic measure} relative to $\Regionplus\subseteq\RRNplus$ and $p\in\Regionplus$. 
    \medskip 

    A function $f$ on $\partial_e\Regionplus$ is said to be \textbf{parabolically integrable} if it is $\parameasure$-integrable for any $p\in\Regionplus$. 
\end{definition}

The set of resolutive functions $f$ for a given set $\Regionplus$ is exactly the class of \textit{parabolically integrable} functions on $\partial_e(\Regionplus)$ (see for instance Corollary~8.34 of \cite{Wat12}). The integral representation given by Equation~\ref{eqn:defPoissonInt} means that the PWB solution $\wsoln$ to Problem~\ref{prob:arbitraryHeatProb} at a point $p$ is exactly the parabolic average of the boundary value $f$ with respect to the associated parabolic measure $\parameasure$. 

\section{General Heat Content Results for Self-Similar Sets}
\label{sec:heatContent}
%
%

The heat content of a set is the amount of heat within that spatial region, given by an integral of the heat solution over that region as a function of the time variable. In what follows, we will standardize the boundary conditions from Problem~\ref{prob:generalHeatProblem} in the following way. First, consider a bounded open set $\Region\subset\RR^\dimension$. At the initial time, suppose the total heat inside of $\Region$ is zero and for all time in the future, suppose that the boundary $\partial\Region$ is held at a constant temperature, which we suppose to be equal to one. Explicitly,
\index{Heat equation}
\begin{equation}
    \label{prob:specificHeatProblem}
    \left\{
    \begin{aligned}
        \partial_t u - C\Delta u &= 0   &&\text{in }\Region\times[0,\infty), \\
            u &= 0                      &&\text{on }\Region\times\set{0},    \\
            u &= 1                      &&\text{on }\partial\Region\times (0,\infty).
    \end{aligned}
    \right.
\end{equation}
In particular, we will obtain an explicit formula for this quantity for small values of $t$. These boundary conditions are a model for heat flow into a region. 

We note that for an arbitrary open set $\Regionplus\subseteq\RRNplus$, we may use the PWB solution to Problem~\ref{prob:specificHeatProblem} as established in Section~\ref{sec:heatPWBsoln}. The choice of piecewise constant boundary conditions is sufficient to ensure that the boundary function is resolutive. Explicitly, we may define the function 
\[ F(x,t) = \1_{\partial\Region}(x), \]
which is a measurable function for every $\parameasure$, with $\Regionplus=\Region\times(0,\infty)$ and $p\in\Regionplus$, since they are Borel measures and $\partial\Region$ is closed. Further, it is clearly integrable since $\parameasure(\partial_e\Regionplus)=1$ since they are probability measures. Thus $F$ is parabolically integrable and admits a PWB solution $w_{\Regionplus,F}$. 

Further, we also have that the solution to Problem~\ref{prob:specificHeatProblem} is unique since it is nonnegative, which can be seen from the minimum principle for the heat equation (originally shown in \cite{Wid44}), or else because it is bounded by the maximum principle, hence certainly of sub-quadratic exponential growth (originally shown in \cite{Tyc35}).

The choice of constant boundary conditions has another important property: it is \textit{scale invariant}. This is important for establishing the relationship between Problem~\ref{prob:specificHeatProblem} and the analogous problem for the set $\ph[X]$, where $\ph$ is a similitude. In general, the results of this section may be generalized to parabolically integrable, scale-invariant functions. 

\subsection{Heat Content and its Properties}
%
%

Let $u_\Region$ be the PWB solution to Problem~\ref{prob:specificHeatProblem} on the open bounded region $\Region$. The heat content in $F\subseteq\Region$, where $F$ is a measurable subset of $\Region$, is simply the integral of $u_\Region$ over $F$ as a function of the time parameter. 
\begin{definition}[Heat Content]
    \label{def:heatContent}
    \index{Heat content}
    Let $u_\Region$ be the PWB solution to the Dirichlet Problem~\ref{prob:generalHeatProblem} on the set $\Region$ with resolutive boundary conditions. The heat content $\heatContent$ inside any measurable set $F\subseteq\Region$ is defined as 
    \begin{equation}
        \heatContent(t;F) := \int_F u_\Region(x,t) \,dx.
    \end{equation}
\end{definition}
When $F=\Region$, then we simply write $\heatContent(t)=E_\Region(t;\Region)$ and call it the \index{Heat content!Total heat content}\textbf{total heat content}. Note that this definition may be stated for the PWB solution for any Dirichlet problem for the heat equation with resolutive boundary conditions, viz. Problem~\ref{prob:arbitraryHeatProb}. 

In the case of Problem~\ref{prob:specificHeatProblem}, we note that the heat content $\heatContent(t;F)$ is uniformly bounded for all $t\in[0,\infty)$ and any measurable set $F$. 
\begin{proposition}[Heat Content Boundedness]
    \label{prop:heatContentBounded}
    Let $u_\Region$ be the PWB solution to Problem~\ref{prob:specificHeatProblem}. Then for any $t\in[0,\infty)$ and any measurable set $F\subseteq\Region$,
    \begin{equation}
        \label{eqn:heatContentBounded}
        |\heatContent(t;F)| \leq \measure{\Region}.
    \end{equation}
\end{proposition}
\begin{proof}
    This is an immediate corollary of the maximum principle for the heat equation and the boundedness of our chosen set $\Region$. Firstly, we have that our solution $u_\Region$ is nonnegative because $u_\Region(x,0)=0$ and the minimum principle. This implies that for any $F\subseteq\Region$,
    \[ 0\leq \heatContent(t;F) \leq \heatContent(t;\Region), \]
    which follows from properties of the Lebesgue measure. Using now the maximum principle for a solution to the heat equation, we have that $|u_\Region(x,t)|\leq 1$, the largest boundary value. Thus,
    \[ \heatContent(t;\Region) \leq \int_\Region dx = \measure{X} <\infty. \]
\end{proof}

Next, we deduce a scaling property for the heat content. This will require imposing scale-invariant boundary conditions, the simplest of which are constant boundary conditions.
\begin{proposition}[Heat Content Scaling Property]
    \label{prop:heatContentScaling}
    \index{Scale invariance}
    \index{Scale invariance!Parabolic scale invariance}
    \index{Scaling law!of heat content}
    Let $u_\Region$ and $u_{\ph[\Region]}$ be PWB solutions to Problem~\ref{prob:generalHeatProblem} on a bounded open set $\Region$ and $\ph[\Region]$ in $\RR^\dimension$, respectively, where $\ph$ is a similitude of $\RR^\dimension$ with scaling ratio $\lambda>0$. Suppose that $f$ and $g$ are chosen so that $u_\Region$ is unique, e.g. if they are both nonnegative (in which case $u_\Region\geq 0$, which is sufficient by \cite{Wid44}).
    \medskip
    
    Suppose that the boundary functions $f$ and $g$ have the following scale invariance properties with respect to $\ph$: for any $x\in\Region$, $f$ is \textbf{scale invariant} (i.e. $f(\ph(x),0)=f(x,0)$) and $g$ is \textbf{parabolic scale invariant} (i.e. $g(\ph(x),\lambda^2t)=g(x,t)$). Then 
    \[ E_{\ph[\Region]}(t) = \lambda^\dimension E_\Region(t/\lambda^2). \]
\end{proposition}
\begin{proof}
    
    
    First, we have that for any $x\in \Region$ and for all $t>0$,
    \begin{equation}
        \label{eqn:heatScaling}
        u_{\ph[X]}(\ph(x),\lambda^2 t)=u_\Region(x,t).
    \end{equation}
    To see this, we note that the function $v(x,t):=u_{\ph[\Region]}(\ph(x),\lambda^2 t)$ satisfies the heat equation for any $x\in\Region$ and for all $t>0$, since $(\partial_t-\Delta)v=\lambda^2(\partial_t-\Delta)u_{\ph[\Region]}=0$ on the points $(\ph(x),t)\in\ph[\Region]\times(0,\infty)$. 
    
    Furthermore, $v$ satisfies the same boundary conditions as the function $u_\Region$. Let us use the notation in the discussion of Problem~\ref{prob:arbitraryHeatProb}: let $u_\Region=u_{\Regionplus,F}$ and $u_{\ph[\Region]}=u_{\ph[\Regionplus],F}$, where $\Regionplus=\Region\times(0,\infty)$ and where $F$ represents respective boundary conditions $f$ and $g$ when appropriately restricted. Let $p_n=(x_n,t_n)\to q=(y,s)$ be a sequence of points in $\Regionplus$ approaching the boundary point $q$, and in the case of a semi-singular boundary points we assume that $t\to s^+$. Let $p'=(\ph(x_n),\lambda^2 t)$ and $q'=(\ph(y),\lambda^2t)$ denote the corresponding points under parabolic transformation. 
    
    By assumption, at any regular boundary point $q$ (respectively $q'$) we have that 
    \begin{align*}
        u_{\Regionplus,F}(p_n)&\to F(q) &\text{ as }p_n\to q\in\partial_e\Regionplus, \\
        u_{\ph[\Regionplus],F}(p'_n)&\to F(q') &\text{ as }p_n\to q\in\partial_e\Regionplus.
    \end{align*}
    In other words, $u_{\Regionplus,F}$ (respectively $u_{\ph[\Regionplus],F}$) converge on the regular essential boundary in the sense of Definition~\ref{def:convergenceOnEssenBd}. We note that if $q$ is a regular boundary point, then so too is $q'$ for its corresponding boundary. This can be seen from the barrier criterion (see for instance Theorem~8.46 of \cite{Wat12}). Namely, if there is a barrier $w$ defined near the point $q$, then the corresponding function $w(\ph(x),\lambda^2t)$ will be a barrier at $q'$. 
    
    By definition, $v(p)=u_{\ph[\Regionplus],F}(p')$, so $v(p)\to F(q')$. Under the parabolic scale invariance assumption, $F(q')=F(q)$. Thus, we conclude by uniqueness that $v(x,t)=u_\Region(x,t)$, which is exactly Equation~\ref{eqn:heatScaling}. 
    Using properties of the Lebesgue integral and Equation~\ref{eqn:heatScaling}, it follows that
    \begin{align*}
        E_{\ph[\Region]}(t) &= \int_\Region u_{\ph[\Region]}(\ph(x) ,t)\,d\ph(x) \\ 
        &= \lambda^{\dimension}\int_\Region u_{\ph[\Region]}(\ph(x),\lambda^2(t/\lambda^2))\,dx \\ 
        &= \lambda^{\dimension}\int_\Region u_{\Region}(x,t/\lambda^2)\,dx \\ 
        &= \lambda^{\dimension}E_\Region(t/\lambda^2).
    \end{align*}
\end{proof}

Note that an important part of this scaling property is that the heat content in a region is related to a \textit{rescaled} problem's heat content (c.f. $E_{\lambda\Region}(t)$), not a \textit{restricted} heat content (c.f. $E_\Region(t;\lambda\Region)$). Further, we note that Proposition~\ref{prop:heatContentScaling} implies that the total heat content $\heatContent(t)=E(t;\Region)$ viewed as a family of functions satisfies a $2$-scaling law in the sense of Definition~\ref{def:scalingLaw}, provided a fixed set of boundary conditions for which Proposition~\ref{prop:heatContentScaling} holds. We state this in our specific case, i.e. for Problem~\ref{prob:specificHeatProblem}.

\begin{corollary}[$2$-Scaling Law of Heat Content]
    \label{cor:heatScalingLaw}
    \index{Scaling law!of heat content}
    Given any bounded open set $A\subset\RRN$ and any similitude $\ph$ on $\RRN$, let $E_A(t)$ be the total heat content for the PWB solution $u_\Region$ of Problem~\ref{prob:specificHeatProblem} for each of the sets $\Region=A$ and $\Region=\ph^{\circ k}[A]$, $k\in\NN$. Then the family of normalized functions $t^{-\dimension/2}E_A(t)$ (with respect to the $\Phi$-closed set $\Aa=\set{\ph^k[A]\suchthat k\in\NN_0}$) satisfies a $2$-scaling law in the sense of Definition~\ref{def:scalingLaw}, viz. for every $A\in\Aa$ and $t>0$,
    \[ t^{-\dimension/2}E_A(t/\lambda^2) = t^{-\dimension/2}E_{\ph[A]}(t). \]
\end{corollary}
\begin{proof}
    The constant, nonnegative boundary conditions of Problem~\ref{prob:specificHeatProblem} ensure that $u_\Region$ is unique and are clearly scale invariant and parabolically scale invariant in the respective variables. Because $\ph$ is a similitude, every set $\ph^{\circ k}[A]$ is bounded (use the bound $\lambda^kC$, where $C$ is a bound for $A\subseteq B_C(0)$) and open (which follows because $\ph$ is injective and its inverse is also a similitude, hence continuous). 

    Applying Proposition~\ref{prop:heatContentScaling} to an arbitrary set $A\in\Aa$, we have that 
    \begin{align*}
        t^{-\dimension/2}E_{\ph[A]}(t) = t^{-\dimension/2}\lambda^\dimension E_{A}(/\lambda^2) = (t/\lambda^2)^{-\dimension/2}E_A(t/\lambda^2).
    \end{align*}
\end{proof}

Suppose now that our region $\Region$ has a self-similar boundary $\Regionbd$ corresponding to a self-similar system $\Phi$ and that $\Omega$ is an osculating set for $\Phi$. In order to use the scaling property to obtain a scaling functional equation, we will need an induced decomposition in the sense of Definition~\ref{def:inducedDecomp}. For later use, we define the remainder which occurs as the difference of the heat content and the sum of scaled copies. Given a suitable estimate of this quantity, we will be able to obtain explicit formulae for the heat content.

\begin{definition}[Decomposition Remainder]
    \label{def:heatDecompRem}
    \index{Decomposition!Decomposition remainder}
    Let $u_\Region$ be the PWB solution to Problem~\ref{prob:specificHeatProblem} on the bounded open set $\Region\subset\RRN$. Suppose that $\Regionbd=\partial\Region$ is the attractor of a self-similar system $\Phi$ and that the relative fractal drum $(\Regionbd,\Region)$ is osculant. 
    \medskip 

    We define the \textbf{decomposition remainder} of $\heatContent$ to be the quantity 
    \begin{equation}
        \label{eqn:defDecompRem}
        \decompRem(t) := \heatContent(t) - \sum_{\ph\in\Phi} E_{\ph[\Region]}(t).
    \end{equation}
    The normalized remainder is the quantity $R(t):=t^{-\dimension/2}\decompRem(t)$ for $t>0$.    
\end{definition}

Note that in the case of Problem~\ref{prob:specificHeatProblem}, we can deduce that $\decompRem$ is bounded for all $t\geq 0$ as a corollary of Proposition~\ref{prop:heatContentBounded}. This will provide an upper bound for the abscissa of convergence of its truncated Mellin transform, $\Mellin^\delta[t^{-\dimension}\decompRem(t)]$, as by Lemma~\ref{lem:MellinHolo} we may deduce that it is holomorphic in the right half plane $\HH_{\dimension/2}$. This helps establish some technical preliminaries for later proofs, allowing us to concern ourselves solely with asymptotic estimates of the form $\decompRem(t)=O(t^{\rempow/2})$ as $t\to0^+$. 

The most precise results occur when there is no error term $R$ (in which case we have estimates of the form $O(t^{n})$ at $t\to0^+$ for any $n>0$), which occurs when the set $X$ partitions into disjoint, self-similar copies (up to sets of measure zero). The novelty of this method, though, lies in its ability to handle cases where this is not true, but with explicit estimates for the degree to which the decomposition is not exact, such as in the case of generalized von Koch fractals. The key part of using this scaling function approach is obtaining good estimates for the remainder term (viz. for small values of $\sigma_0$), as the expansion will be explicit up to an order determined by these estimates. For the reader wishing to establish such functional equations for a particular class of self-similar sets, we recommend studying the methods and results in \cite{FLV95a,vdB00_squareGKF}.

\subsection{Heat Zeta Functions}
%
%

Suppose that $\Region\subset\RR^\dimension$ is a bounded, self-similar set and let $\Phi$ be a self-similar system with $\Region$ as its attractor. Suppose that $\Omega\subset\RR^\dimension$ is an osculating set for $\Phi$ so that $(X,\Omega)$ is an osculant fractal drum (see Definition~\ref{def:oscRFD}). Let 
\begin{equation}
    \label{eqn:partition}
    X = \enclose{\biguplus_{\ph\in\Phi} \ph[\Omega]} \uplus \remset
\end{equation}
be a partition of $X$ where $\remset$ is defined as $\remset:=X\setminus (\cup_{\ph\in\Phi}\ph[\Omega])$. That the union is disjoint follows from the definition of an osculating set.  

Our goal is to use the methods of Chapter~\ref{chap:SFEs} to study Problem~\ref{prob:specificHeatProblem} on the region $\Region$. If the self-similar system $\Phi$ induces a decomposition of the total heat contents (in the sense of Definition~\ref{def:inducedDecomp}), we can use Corollary~\ref{cor:heatScalingLaw} to obtain a scaling functional equation. To this end, we establish now some preliminaries and notation for the application of this theory. 

We begin with the notion of a \textit{heat zeta function}, defined in analogy with the definition for tube zeta functions and their normalized quantities which satisfy a scaling law. We define it for the general problem on a cylindrical set and in the case of Problem~\ref{prob:specificHeatProblem}, we say that the heat zeta function is a parabolic (heat) zeta function owing to the nature of the solution $u_\Region$ as a parabolic average at the points of the given set $\Region$. 
\begin{definition}[Heat Zeta Function]
    \label{def:heatZetaFn}
    \index{Heat zeta function}
    \index{Heat zeta function!Parabolic (heat) zeta function}
    Let $u_\Region$ be the PWB solution to Problem~\ref{prob:generalHeatProblem} on the open set $\Region\subset\RRNplus$ with a resolutive boundary function $F=f\1_{t=0}+g\1_{\partial\Region}$. Suppose that the total heat content $\heatContent(t)$ associated to $u_\Region$ is integrable and bounded for $t\in(0,\delta]$ and satisfies the estimate that $\heatContent(t)=O(t^{-\sigma_0})$ as $t\to0^+$. 
    \medskip

    Then we define the \textbf{heat zeta function} for $s\in\HH_{\sigma_0}$ to be
    \begin{equation}
        \label{eqn:defHeatZeta}
        \genheatzeta(s;\delta) := \Mellin^\delta[t^{-\dimension/2}\heatContent(t)](s)
    \end{equation}
    and extended by analytic continuation. If $F=\1_{\partial\Region}$ as in Problem~\ref{prob:specificHeatProblem}, we say that $\heatzeta:=\genheatzeta$ is a \textbf{parabolic} (heat) zeta function. 
\end{definition}
Note that we have that $\genheatzeta$ is holomorphic in $\HH_{\sigma_0}$ by Lemma~\ref{lem:MellinHolo}. When we establish functional equations for the heat content induced by a self-similar system $\Phi$, we also note that the relevant scaling operator will be $L_\Phi^2:=\sum_{\ph\in\Phi}M_{\lambda_\ph^2}$, but the associated scaling zeta function $\zeta_\Phi$ is still the primary function of interest since $\zeta_\Phi(2s)=\zeta_{L_\Phi^2}(s)$. Owing to the effect of this quadratic scaling, viz. $\alpha=2$ as in Chapter~\ref{chap:SFEs}, we will also state our estimates for the heat content remainders in the form $R(t)=O(t^{-\sigma_0/2})$ (as $t\to0^+$). We state our results provided some control over the decomposition remainder as in Definition~\ref{def:heatDecompRem}. 

\begin{theorem}[Formula for the Heat Zeta Function]
    \label{thm:heatZetaFormula}
    \index{Heat zeta function!Explicit formula}
    Let $\heatContent$ be the heat content of the PWB solution to Problem~\ref{prob:specificHeatProblem}. Let $\Phi$ be a self-similar system and suppose that $\heatContent$ decomposes according to $\Phi$ with decomposition remainder $\decompRem$ (as in Definition~\ref{def:heatDecompRem}) which satisfies $\decompRem(t)=O(t^{\rempow/2})$ as $t\to0^+$.
    \medskip 

    Define $\sigma_R:=\sigma_0/2$ and let $\zeta_\Phi$ be the scaling zeta function associated to $\Phi$. Then for any $\delta>0$ and for all $s\in\HH_{\sigma_R}\setminus\Dd_\Phi$,
    \begin{equation}
        \label{eqn:heatZetaFormula}
        \heatzeta(s;\delta) = \zeta_\Phi(2s)(\partialheatzeta(s;\delta)+\zeta_R(s;\delta)),
    \end{equation}
    with $\heatzeta(s;\delta)$ meromorphic in $\HH_{\sigma_R}$, having poles contained in a subset of $\Dd_\Phi$, the set of poles of $\zeta_\Phi$. Here, $\zeta_R(s;\delta):= \Mellin^\delta[t^{-\dimension/2}\decompRem(t)](s)$ is a holomorphic function in $\HH_{\sigma_R}$ and 
    \begin{equation}
        \label{eqn:partialZetaDef}
        \partialheatzeta(s;\delta):= \sum_{\ph\in\Phi}\lambda_\ph^{2s}\Mellin_\delta^{\delta/\lambda_\ph^2}[t^{-\dimension/2}\heatContent(t)](s)
    \end{equation}
    is an entire function. 
\end{theorem}

\begin{proof}
    By Corollary~\ref{cor:heatScalingLaw}, we have that $\heatContent$ satisfies a $2$-scaling law. Together with the presupposed induced decomposition, by Proposition~\ref{prop:inducedSFE} we have that the normalized $\heatContent$ satisfies the scaling functional equation 
    \begin{equation}
        \label{eqn:heatSFE}
        t^{-\dimension/2}\heatContent(t) 
            = L_\Phi^2[t^{-\dimension/2}\heatContent(t)] + t^{-\dimension/2}\decompRem(t)
    \end{equation}
    for any $t\geq 0$. The result is then a corollary of Theorem~\ref{thm:zetaFormula}, noting that the normalized remainder satisfies $R(t)=t^{-\dimension/2}\decompRem(t)=O(t^{-\sigma_R})$ at $t\to0^+$, with $\sigma_R=\sigma_0/2$.
\end{proof}

Note that this formula implies that the poles of the heat zeta function $\heatzeta$ in the half-plane $\HH_{\sigma_R}=\HH_{\sigma_0/2}$ occur exactly at the points $\omega/2$ where $\omega\in\Dd_\Phi(\HH_{\sigma_0})$. In Chapter~\ref{chap:geometry}, we see that this set controls the possible complex dimensions of the set $\Regionbd$ relative to its interior $\Region$. Here, we see that the scaling zeta function $\zeta_\Phi$, determined from the scaling ratios and their multiplicities alone, also governs the nature of the zeta functions for self-similar fractals.

As we will use this scaling functional equation in order to obtain explicit formulae in what follows, we record this result as a proposition to which we will refer later. 
\begin{proposition}[Heat Scaling Functional Equation]
    \label{prop:heatSFE}
    \index{Scaling functional equation (SFE)!Induced SFE of heat content}
    Let $\heatContent$ be the heat content of the PWB solution to Problem~\ref{prob:specificHeatProblem}. Let $\Phi$ be a self-similar system and suppose that $\heatContent$ decomposes according to $\Phi$ with decomposition remainder $\decompRem$ (as in Definition~\ref{def:heatDecompRem}) which satisfies $\decompRem(t)=O(t^{\rempow/2})$ as $t\to0^+$.
    \medskip 

    Then for all $t\geq 0$, the heat content $\heatContent$ satisfies the scaling functional equation given by Equation~\ref{eqn:heatSFE}. Note that the normalized remainder $R(t)=t^{-\dimension/2}\decompRem(t)$ satisfies the corresponding estimate $R(t)=O(t^{-\sigma_R})$ as $t\to0^+$, where $\sigma_R:=\sigma_0/2$.  
\end{proposition}
\subsection{Explicit Heat Content Formulae}
\label{sub:genHeatFormulae}
%
%

We now state some general results regarding pointwise and distributional explicit formulae for the heat contents of open bounded regions $\Region\subset\RRN$ whose boundary $\Regionbd$ is a self-similar set and for which $(\Regionbd,\Region)$ is an osculant fractal drum. The pointwise expansions require integrating the heat content to improve regularity of the function to obtain. Meanwhile, the distributional identities (while less direct) remove this technical hurdle.

First, we state the pointwise identity, given in terms of antiderivatives of the heat content. We denote by  $\heatContent^{[k]}$ the $k^{\text{th}}$ antiderivative of the heat content defined so that $\heatContent^{[k]}(0)=0$. Explicitly, it may be defined by recursion with $\heatContent^{[0]}:=\heatContent$ and for $k>0$ by
\[
    \heatContent^{[k]}(t) := \int_0^t \heatContent^{[k-1]}(\tau)\,d\tau.
\] 
As an additional preliminary to stating the result, we define the Pochhammer symbol $(z)_w:=\Gamma(z+w)/\Gamma(w)$ for $z,w\in\CC$. Note that when $w=k$ is a positive integer, this simplifies to $(z)_k=z(z+1)\cdots(z+k-1)$ or, when $w=0$, to $(z)_0=1$. Lastly, we note that the summations to follow are defined as symmetric limits of the sums taken over values of $\omega$ with bounded imaginary parts for increasingly large upper bounds. 

\begin{theorem}[Heat Content, Pointwise Expansion]
    \label{thm:heatFormulaPtw}
    \index{Heat content!Pointwise explicit formulae}
    Let $\heatContent$ be the heat content of the PWB solution to Dirichlet Problem~\ref{prob:specificHeatProblem}. Let $\Phi$ be a self-similar system and suppose that $\heatContent$ decomposes according to $\Phi$ with decomposition remainder $\decompRem$ satisfying $\decompRem(t)=O(t^{\rempow/2})$ as $t\to0^+$. Suppose either that $\sigma_0<\lowersimdim(\Phi)$ (see Definition~\ref{def:lowerSimDim}) or that the scaling ratios of $\Phi$ are arithmetically related (see Definition~\ref{def:latticeDichotomy}).
    \medskip

    Let $k\in\ZZ$ with $k\geq 2$ and let $\delta>0$. Then for all $t\in(0,\delta)$, we have that
    \[ 
        \heatContent^{[k]}(t) = \sum_{\omega\in\Dd_\Phi(\HH_{\sigma_0})} 
            \Res\Bigg(\cfrac{t^{(\dimension-s)/2+k}}{((\dimension-s)/2+1)_k}\heatzeta(s/2;\delta);\omega\Bigg)
            + \Rr^k(t),
    \]
    where $\heatzeta$ is as in Theorem~\ref{thm:heatZetaFormula}. For any $\e>0$ sufficiently small, the remainder satisfies $\Rr^k(t)=O(t^{\rempow/2-\e+k})$ as $t\to0^+$. 
\end{theorem}

\begin{proof}
    By Proposition~\ref{prop:heatSFE}, we have that the normalized heat content satisfies a scaling functional equation with respect to the operator $L_\Phi^2$ and remainder $R(t):=t^{-\dimension/2}\decompRem(t)$. Letting $\sigma_R=\sigma_0/2$, we note that $R(t)=O(t^{-\sigma_R})$ as $t\to0^+$ and that $\sigma_0<D_\ell=\lowersimdim(\Phi)$ implies that $\sigma_R<D_\ell/2$. So in the first case, we apply Corollary~\ref{cor:lowerDimAdmissibilityRescaled} to obtain that $R$ is an admissible remainder with respect to screens of the form $S_\e(\tau)\equiv \sigma_R+\e$ when $0<\e<D_\ell/2-\sigma_R$. Meanwhile, in the lattice case, we apply Theorem~\ref{thm:latticeCaseAdmissibility} to obtain that any screen of the form $S_\e(\tau)=\sigma_R+\e$ are admissible for all but finitely many $\e>0$. In either case, $S_\e$ is admissible for any sufficiently small $\e$. 
    
    We may now apply Theorem~\ref{thm:pointwiseFormula}. Here, $\beta=\dimension$, $\alpha=2$, and $\delta>0$ is arbitrary. Note that $\dimension/2\leq\simdim(\Phi)/2$ by Theorem~\ref{thm:similarityDimension} and that $\dimension/2\geq \sigma_R$ since Proposition~\ref{prop:heatContentBounded} implies that $\zeta_R$ is holomorphic in the half-plane $\HH_{\dimension/2}$ as noted in the discussion following Definition~\ref{def:heatDecompRem}. Here, with the estimate $\decompRem(t)=O(t^{\rempow/2})$ as $t\to0^+$, we obtain by Lemma~\ref{lem:MellinHolo} that $\zeta_R$ is holomorphic in $\HH_{\sigma_R}$. Note also that $S_\e$ is chosen with $S_\e(\tau)\equiv\sigma_R+\e>\sigma_R$, so $W_{S_\e}=\HH_{\sigma_R+\e}$ and is contained in $\HH_{\sigma_R}$. Also, we note that for the estimate of the remainder, $\beta/\alpha-\sup(S)+k$ in the theorem is $\rempow/2-\e+k$ since $\sup(S)=\sigma_R+\e=\sigma_0/2+\e$ and $\beta/\alpha=\dimension/2$. 

    Lastly, we simplify the sum over $\Dd_\Phi(\alpha W_S)=\Dd_\Phi(2\HH_{\sigma_R+\e})$. Firstly, note that $2\HH_{\sigma_R+\e}=\HH_{2\sigma_R+2\e}$. Since $\zeta_R$ is holomorphic in $\HH_{\sigma_R}$, $\zeta_R(s/2)$ is holomorphic when $s\in\HH_{2\sigma_R}=\HH_{\sigma_0}$. In both cases, we may choose $\e$ sufficiently small so that $\zeta_\Phi(2s)$ has no poles of the form $s=\omega$ with $\Re(\omega)\in(\sigma_0,\sigma_0+2\e)$. In the case when $\sigma_0<D_\ell$, take $\e<(D_\ell-\sigma_0)/2$ and in the lattice case choose $2\e$ smaller than the distance between $\sigma_0$ and the closest exceptional point. Then, with the change of variables, we see that when $\Re(s)\in(\sigma_0,\sigma_0+2\e)$, the function $\heatzeta(s/2)$ has no poles because $\zeta_\Phi(s)$ has no poles with $\Re(s)\in(\sigma_0,\sigma_0+2\e)$ (using Theorem~\ref{thm:heatZetaFormula}). Thus the sum over $\Dd_\Phi(\HH_{\sigma_0+2\e})$ is the same as over $\Dd_\Phi(\HH_{\sigma_0})$. 
\end{proof}

The leading order term occurs at the pole $\omega=D:=\simdim(\Phi)$, as this is the abscissa of convergence of $\heatzeta$. If $D$ is a simple pole and the only pole with real part $D$, then the leading order term is explicitly 
\[ \frac1{((\dimension-D)/2+1)_k}\Res(\heatzeta(s/2;\delta);D)\,t^{(\dimension-D)/2+k}. \]
When the scaling ratios are non-arithmetically related, $D$ will be the unique pole with real part $D$; when the scaling ratios are arithmetically related, this will not be true. Instead, the leading order terms will form the Fourier expansion of a periodic function owing to the nature of the structure of the poles of $\zeta_\Phi$ in this case: they lie on finitely many vertical lines and are periodically spaced. See for instance the proof of Theorem~\ref{thm:latticeCaseAdmissibility} in this work for an explicit proof of this and more generally we refer the reader to the discussion in Chapters~2 and 3 in \cite{LapvFr13_FGCD} regarding the structure of the complex dimensions of self-similar fractal harps (and their natural generalizations), noting that these zeta functions have the same structure as $\zeta_\Phi$, and thus controlling the poles of $\heatzeta$.


For the distributional setting, the restriction of $k\geq 2$ may be relaxed. For these formulae, recall from Chapter~\ref{chap:SFEs} that (just as in Chapter~5 of \cite{LapvFr13_FGCD} and Chapter~5 of \cite{LRZ17_FZF}) we will use as the space of test functions the set of Schwartz functions. These are the functions of rapid decrease near the boundary, given explicitly by
\begin{equation}
    \Ss(0,\delta):=\Bigg\{\testfn\in C^\infty(0,\delta)\,\Bigg|\, 
    \begin{aligned}
        &\forall m\in\ZZ,\,\forall q\in\NN,\,t^m\testfn^{(q)}(t)\to0\\
        &\text{and }(t-\delta)^m\testfn^{(q)}(t)\to0\text{ as }t\to0^+
    \end{aligned}
    \Bigg\}.
\end{equation}
The tempered distributions are elements of the dual space, $\Ss'(0,\delta)$. 

\begin{theorem}[Heat Content, Distributional Expansion]
    \label{thm:heatFormulaDist}
    \index{Heat content!Distributional explicit formulae}
    Let $\heatContent$ be the heat content of the PWB solution to Dirichlet Problem~\ref{prob:specificHeatProblem}. Let $\Phi$ be a self-similar system and suppose that $\heatContent$ decomposes according to $\Phi$ with decomposition remainder $\decompRem$ satisfying $\decompRem(t)=O(t^{\rempow/2})$ as $t\to0^+$. Suppose either that $\sigma_0<\lowersimdim(\Phi)$ (see Definition~\ref{def:lowerSimDim}) or that the scaling ratios of $\Phi$ are arithmetically related (see Definition~\ref{def:latticeDichotomy}).
    \medskip

    Let $k\in\ZZ$ and let $\delta>0$. Then the heat content, viewed as a tempered distribution, satisfies the identity
    \begin{equation}
        \label{eqn:heatFormulaDist}
        \begin{split}
            \heatContent^{[k]}(t) &= \sum_{\omega\in\Dd_\Phi(\HH_{\sigma_0})} 
                \Res\Bigg(\cfrac{t^{(\dimension-s)/2+k}}{((\dimension-s)/2+1)_k}\heatzeta(s/2;\delta);\omega\Bigg)
                + \Rr^{[k]}(t),
        \end{split} 
    \end{equation}
    as $t\to0^+$, where two distributions are equal if they action on an arbitrary test function agrees. See Equation~\ref{eqn:bracketIdentityHeat} for the explicit identity of action on test functions. Here, $\heatzeta$ is as in Theorem~\ref{thm:heatZetaFormula}. The remainder term, as a distribution, satisfies the estimate that for any $\e>0$ sufficiently small, $\Rr(t)=O(t^{\rempow/2-\e+k})$ as $t\to0^+$, in the sense of Equation~\ref{eqn:distRemEstHeat}.
\end{theorem}
\begin{proof}
    The proof is the same as Theorem~\ref{thm:heatFormulaPtw} but through application of Theorem~\ref{thm:distFormula}. 
\end{proof}

To say that Equation~\ref{eqn:heatFormulaDist} is an identity of distributions means the following. For any test function $\testfn\in\Ss(0,\delta)$, 
\begin{equation}
    \label{eqn:bracketIdentityHeat}
    \begin{split}
        \bracket{\heatContent^{[k]},\testfn} &= \sum_{\omega\in\Dd_\Phi(\HH_{\sigma_0})} 
        \Res\Bigg(\cfrac{\Mellin[\testfn]({(\dimension-s)/2+k+1})}{((\dimension-s)/2+1)_k}\heatzeta(s/2);\omega\Bigg)
        + \bracket{\Rr^{[k]},\testfn}.
    \end{split}
\end{equation}
The distributional remainder estimate is equivalent to the statement that for all $a>0$ and for all $\testfn\in\Ss(0,\delta)$,
\begin{equation}
    \label{eqn:distRemEstHeat}
    \Big\langle\Rr^{[k]}(t),\frac1a\testfn(t/a)\Big\rangle = O(a^{\rempow/2-\e+k}),
\end{equation}
as $t\to0^+$. While this formulation is less direct compared to the pointwise expansion, it does have the advantage of requiring less regularity to leverage the expansion. Explicitly, one may set $k=0$ and obtain a distributional identity for the heat content itself. Note that in this case, the Pochhammer symbol simplifies to $((\dimension-s)/2+1)_0=1$.

\section{Examples: Generalized von Koch Snowflakes}
\label{sec:meltingSnow}
%
%

In this section, we provide a direct application of our results to generalized von Koch fractals. The analysis of the heat content of these fractal snowflakes originates with the work of Fleckinger, Levitin, and Vassiliev \cite{FLV95a,FLV95b} on the von Koch snowflake and then the later work of Michiel van den Berg and collaborators \cite{vdB00_generalGKF,vdB00_squareGKF,vdBGil98,vdBHol99} on generalized von Koch fractals. The main tool in their analysis is the renewal theorem, developed by Feller \cite{Fel50} in the field of probability. 


In so doing, we will recover the results of \cite{vdB00_generalGKF,vdB00_squareGKF,FLV95a} when either $n\geq5$ or in the lattice case. Notably, this extends the results for generalized von Koch fractals to the nonlattice case when $n\geq5$. Results in the nonlattice case for $n=3$ or $4$ may be deduced provided a priori knowledge of pole-free regions and estimates for the explicit function $\zeta_\Phinr$ with the methods established in this work as well. Furthermore, we explicitly show the role of the complex dimensions in the explicit formulae for the heat content.

\subsection{Geometric Preliminaries}

Given an $(n,r)$-von Koch snowflake, let $\Knr=\Regionbd$ denote its boundary as in Definition~\ref{def:vkSnow}. We will require that $\Knr$ is a simple, closed curve so that is has a connected interior $\Region$ (i.e. the bounded component of $\RR^2\setminus\Knr$, using the Jordan curve theorem). For this reason, we impose the self-avoidance criterion of Proposition~\ref{prop:selfAvoid} for $n$ and $r$. 

In what follows, we will consider Problem~\ref{prob:specificHeatProblem}. We let $u_\Region$ be the PWB solution and $\heatContent$ be the associated heat content. Our goal will be to explicitly describe this heat content in the limit as $t\to0^+$ by applying the explicit formulae results. We note that by symmetry, we need only consider the heat content in one portion of $\Omega$ contained in a sector of angle $2\pi/n$ and that the heat content of the total set is $n$ times the heat content in this restricted region. 

We let $\Phinr$ be the self-similar system defined by Equation~\ref{eqn:defGKCsystem} associated to an $(n,r)$-von Koch curve. We define $\ell=(1-r)/2$ to be the conjugate scaling ratio to $r$ and we impose that $r\in(0,\frac13]$. Note that when $r=\ell=\frac13$ and $n=3$, this is exactly the standard von Koch snowflake (see Figure~\ref{fig:threeGKFs}). The scaling zeta function associated to $\Phinr$ is given explicitly by 
\[
    \zeta_{\Phinr}(s) = \cfrac1{1-(n-1)r^{2s}-2\ell^{2s}}.
\]

Additionally, we recall some relevant information from Section~\ref{sec:appToGKFs}. Letting $D_\ell=\lowersimdim(\Phinr)$, we have that $D_\ell>0$ when $n\geq5$, $D_\ell=0$ when $n=4$, and $D_\ell<0$ when $n=3$ by Lemma~\ref{lem:lowerDimBoundsGKF}. It is for this reason that our results for nonlattice $(n,r)$-von Koch snowflakes will be restricted to $n\geq5$ since we need this condition to apply the criterion of Theorem~\ref{thm:lowerDimAdmissibility} with the available estimates. Meanwhile, in the lattice case (where $n\geq3$), we have explicit information about the poles of $\zeta_\Phinr$. Notably, by Proposition~\ref{prop:latticeSimplePolesGKF} we have that the poles of $\Phinr$ are all simple. 

The last major preliminary is an induced decomposition (in the sense of Definition~\ref{def:inducedDecomp}) and an estimate of the decomposition remainder term, $\decompRem$. This is needed to find a scaling functional equation for the heat content $\heatContent$. A scaling functional equation for $\heatContent$ in the case of $(n,r)$-von Koch fractals was established in \cite{vdB00_generalGKF} (without this terminology). The decomposition remainder, which we will denote by $\decompRem$ is given explicitly by Equation~1.12 in \cite{vdB00_generalGKF} and Proposition~1.2 in \cite{vdB00_generalGKF} gives an explicit estimate. Namely, $\decompRem(t)=O(t)$ as $t\to0^+$, in which case we have that $\sigma_0=0$ since $\dimension/2=1$ in $\RR^2$.

\subsection{Explicit Heat Content Formulae}

We will give two types of explicit formulae. The first result is regarding the pointwise-valid formulae for $k\geq2$, and the second will be the distributional formulae specialized to the case when $k=0$. Note that the general pointwise formulae may be deduced in the same way using the general results for heat contents of self-similar sets. 

We provide the same results as in Section~\ref{sec:appToGKFs} so that the reader may directly compare between the tube formulae and heat content formulae, which are exceedingly similar. Most notably, complex numbers indexing the sums are precisely the same, controlled only by the self-similar system $\Phinr$. In other words, they depend only on the scaling ratios $r$ and $\ell$ and their corresponding multiplicities (which depend on $n$).

For the preliminaries regarding the notation and conventions defining the antiderivatives $\heatContent^{[k]}$, the Pochhammer symbol $(z)_w$, $z,w\in\CC$, and the definition of the sums over complex dimensions as symmetric limits, we refer the reader to Section~\ref{sec:heatContent}.

\begin{theorem}[Pointwise Heat Content Formulae for GKFs]
    \label{thm:pointwiseHeatFormulaGKF}
    Let $\Knr$ be an $(n,r)$-von Koch Snowflake boundary satisfying the self-avoidance criterion in Proposition~\ref{prop:selfAvoid} and let $\Omega$ be the interior region defined by $\Knr$ (i.e. the bounded component of $\RR^2\setminus\Knr$). Suppose that either $n\geq5$ or that the scaling ratios $r$ and $\ell=(1-r)/2$ are arithmetically related (i.e. the ratio of their logarithms is rational). 
    \medskip

    Then for every $k\geq 2$ in $\ZZ$, any $\delta>0$, and every $t\in(0,\delta)$, we have that the antiderivatives $\heatContent^{[k]}$ of the heat content $\heatContent$ for Problem~\ref{prob:specificHeatProblem} on $\Region$ satisfy
    \[ 
        \heatContent^{[k]}(t) = \sum_{\omega\in\Dd_\Phinr(\HH_{0})} 
            \Res\Bigg(\cfrac{t^{(2-s)/2+k}}{((2-s)/2+1)_k}\heatzeta(s/2;\delta);\omega\Bigg)
            + \Rr^k(t),
    \]
    where $\heatzeta$ is as in Theorem~\ref{thm:heatZetaFormula}. For any $\e>0$ sufficiently small, the error term satisfies $\Rr^k(t)=O(t^{1-\e+k})$ as $t\to0^+$. 
\end{theorem}
\begin{proof}
    We have that $\heatContent$ satisfies a scaling functional equation induced by $\Phinr$ with decomposition remainder $\decompRem(t)=O(t)$ as $t\to0^+$ by Proposition~1.2 of \cite{vdB00_generalGKF}, where $\Phinr$ is as in Equation~\ref{eqn:defGKCsystem}. When normalized (noting that $\dimension/2=1$ in $\RR^2$), the heat content scaling functional equation takes the form 
    \[ t^{-1}\heatContent(t) = L_\Phinr[t^{-1}\heatContent(t)](t) + t^{-1}\decompRem(t). \]
    Further, the remainder $R(t):=t^{-1}\decompRem(t)=O(t^{-\sigma_0})$, with $\sigma_0=0$, as $t\to0^+$ and is continuous on $\RR^+$. Note that we have also multiplied the expression in \cite{vdB00_generalGKF} by $n$ (and distributed by linearity) to deduce a scaling functional equation for the total heat content, rather than the amount of heat content contained in a single symmetric sector. 

    The result then follows from application of Theorem~\ref{thm:heatFormulaPtw} with the following notes. If we assume that $n\geq5$, then we have that $\lowersimdim(\Phinr)>0=\sigma_0$ by Lemma~\ref{lem:lowerDimBoundsGKF}. If $n=3$ or $n=4$, we have assumed that $r$ and $\ell$ are arithmetically related. 
\end{proof}

In the lattice case, we have by Proposition~\ref{prop:latticeSimplePolesGKF} that the poles of $\zeta_\Phinr$ are simple. Thus, we may simplify the residues to obtain an expansion in powers of $t^{(2-\omega)/2+k}$. 
\begin{corollary}[Pointwise Tube Formulae for GKFs, Lattice Case]
    \label{cor:pointwiseheatFormulaGKFLattice}
    Let $\Knr$ be an $(n,r)$-von Koch Snowflake boundary satisfying the self-avoidance criterion in Proposition~\ref{prop:selfAvoid} and let $\Omega$ be the interior region defined by $\Knr$ (i.e. the bounded component of $\RR^2\setminus\Knr$). Suppose that its scaling ratios $r$ and $\ell=(1-r)/2$ are arithmetically related (i.e. the ratio of their logarithms is rational).
    \medskip 

    Then for every $k\geq 2$ in $\ZZ$, any $\delta>0$, and every $t\in(0,\delta)$, we have that the antiderivatives $\heatContent^{[k]}$ of the heat content $\heatContent$ for Problem~\ref{prob:specificHeatProblem} on $\Region$ satisfy
    \[ 
        \heatContent^{[k]}(t) = \sum_{\omega\in\Dd_\Phinr(\HH_{0})} 
            \cfrac{r_\omega}{((2-\omega)/2+1)_k}\,t^{(2-\omega)/2+k}
            + \Rr^k(t),
    \]
    where $r_\omega$ is a constant given by $r_\omega:=\Res\enclose{\heatzeta(s/2;\delta);\omega}$, with $\heatzeta$ is as in Theorem~\ref{thm:heatZetaFormula}. For any $\e>0$ sufficiently small, the error term satisfies $\Rr^k(t)=O(t^{1-\e+k})$ as $t\to0^+$.
\end{corollary}

In order to obtain expansions for the heat content itself when $k=0$, we now move to the distributional formulation of the explicit formulae. We refer the reader to Subsection~\ref{sub:genHeatFormulae} for the relevant preliminaries and more information about the action of distributions or the estimates of their remainder terms. The space of test functions is the class of Schwartz functions on $(0,\delta)$, $\Ss(0,\delta)$, and the tempered distributions are the elements of its dual space, $\Ss'(0,\delta)$. 

\begin{theorem}[Distributional Heat Content Formula for GKFs, $k=0$]
    \label{thm:distHeatFormulaGKF}
    Let $\Knr$ be an $(n,r)$-von Koch Snowflake boundary satisfying the self-avoidance criterion in Proposition~\ref{prop:selfAvoid} and let $\Omega$ be the interior region defined by $\Knr$ (i.e. the bounded component of $\RR^2\setminus\Knr$). Suppose that either $n\geq5$ or that the scaling ratios $r$ and $\ell=(1-r)/2$ are arithmetically related (i.e. the ratio of their logarithms is rational). 
    \medskip

    Then for any $\delta>0$ and every $t\in(0,\delta)$, we have that as an equality of distributions in the Schwartz space $\Ss'(0,\delta)$ (the dual of the space defined in Equation~\ref{eqn:defSchwartzFns}),
    \[ 
        \heatContent(t) = \sum_{\omega\in\Dd_\Phinr(\HH_{0})} 
            \Res\Bigg(t^{(2-s)/2}\heatzeta(s/2;\delta);\omega\Bigg)
            + \Rr(t),
    \]
    where $\heatzeta$ is as in Theorem~\ref{thm:heatZetaFormula}. For any $\e>0$ sufficiently small, the error term satisfies $\Rr(t)=O(t^{1-\e})$ as $t\to0^+$ in the sense of Equation~\ref{eqn:distRemEstHeat}. The action of $\heatContent$ on a test function is given explicitly by Equation~\ref{eqn:bracketIdentityHeat} with $k=0$.
\end{theorem}
\begin{proof}
    The proof is the same as that of Theorem~\ref{thm:pointwiseHeatFormulaGKF} but with application of Theorem~\ref{thm:heatFormulaDist}. Note that we specialize to the case of $k=0$. 
\end{proof}

When we assume that the scaling ratios are arithmetically related, we may simplify this expression regarding the residues. Namely, we obtain an expansion in powers of $t^{(2-\omega)/2}$ with coefficients determined by the residues of $\heatzeta$ at the poles $\omega$. 
\begin{corollary}[Distributional Tube Formula for GKFs, $k=0$]
    \label{cor:distHeatFormulaGKFLattice}
    Let $\Knr$ be an $(n,r)$-von Koch Snowflake boundary satisfying the self-avoidance criterion in Proposition~\ref{prop:selfAvoid} and let $\Omega$ be the interior region defined by $\Knr$ (i.e. the bounded component of $\RR^2\setminus\Knr$). Suppose that the scaling ratios $r$ and $\ell=(1-r)/2$ are arithmetically related (viz. the ratio of their logarithms is rational). 
    \medskip

    Then for any $\delta>0$ and every $t\in(0,\delta)$, we have that as an equality of distributions in the Schwartz space $\Ss'(0,\delta)$ (the dual of the space defined in Equation~\ref{eqn:defSchwartzFns}),
    \[ 
        \heatContent(t) = \sum_{\omega\in\Dd_\Phinr(\HH_{0})} 
            r_\omega\,t^{(2-\omega)/2}
            + \Rr(t),
    \]
    where $r_\omega$ is a constant given by $r_\omega:=\Res\enclose{\heatzeta(s/2;\delta);\omega}$, with $\heatzeta$ is as in Theorem~\ref{thm:heatZetaFormula}. For any $\e>0$ sufficiently small, the error term satisfies $\Rr(t)=O(t^{1-\e})$ as $t\to0^+$ in the sense of Equation~\ref{eqn:distRemEstHeat}. The action of $V_{\Knr,\Omega}$ on a test function is given explicitly by Equation~\ref{eqn:bracketIdentityHeat} with $k=0$.
\end{corollary}

\chapter*{Conclusion}
\addcontentsline{toc}{chapter}{Conclusion}
\label{chap:conclusion}
%
%

The thesis of this work is that self-similar fractals have complex dimensions controlled by the scaling ratios in their construction and furthermore that these complex dimensions govern both geometric and spectral properties of these fractals. On the geometric side, we have analyzed tube functions and the possible complex dimensions (Chapter~\ref{chap:geometry}) and on the spectral side, we have analyzed the heat content of self-similar fractals (Chapter~\ref{chap:heat}). The connection is made evident through the shared method of proof, namely by the means of establishing and solving scaling functional equations (Chapter~\ref{chap:SFEs}). 

To conclude this work, we now compare our results in these separate chapters to see the the connections between these quantities for self-similar fractals. Then, we broaden our scope to the context in which our results fit and the questions or programs that they fall under. Finally, we leave the reader with new questions arising from our results to complement the questions which began our journey and may yet lead to new inspiration.

\section*{Discussion of the Thesis}

Let $\Omega\subset\RR^\dimension$ be a bounded open set with fractal boundary $\partial\Omega$. Suppose that $\partial\Omega$ is a self-similar set, arising as the attractor of the self-similar system $\Phi$, and suppose that $(\partial\Omega,\Omega)$ is an osculant fractal drum with respect to $\Phi$. (In the notation of Chapter~\ref{chap:geometry}, $X=\partial\Omega$.) Per Definition~\ref{def:oscRFD}, this means that $\Phi$ satisfies the open set condition with $\Omega$ as a feasible open set and that the osculating condition holds for $\Omega$. For explicit examples of such fractals, one may consider when $\Omega$ is the interior of an $(n,r)$-von Koch snowflake and $\partial\Omega$ is the snowflake boundary. 

Lastly, we assume some admissibility conditions. Given the respective remainder functions $V_{\partial\Omega,R}$ (i.e. the tube function of the residual set as defined in Theorem~\ref{thm:inducedDecompTubes}) and the decomposition remainder $R_\Omega$ (as in Definition~\ref{def:heatDecompRem}), let $\sigma_0$ be such that as $t\to0^+$,
\begin{align*}
    V_{\partial\Omega,R}(t) &= O(t^{\dimension-\sigma_0}), \\
    R_\Omega(t)             &= O(t^{(\dimension-\sigma_0)/2}).
\end{align*}
Note that we have chosen $\rho=\dimension-\sigma_0$ in the notation of Chapter~\ref{chap:heat} for this comparison. We suppose that either $\sigma_0<\lowersimdim(\Phi)$ or that the scaling ratios of the similitudes in $\Phi$ are arithmetically related. 

Under these assumptions, one may apply any of the corresponding pairs of theorems from Section~\ref{sec:fractalSFEs} and Section~\ref{sec:heatContent} for tube formulae or for heat content, respectively, either for pointwise or distributional asymptotic expansions. For the purposes of this discussion, let us view the distributional expansions with $k=0$ and suppose that the poles of the function $\zeta_\Phi$ are simple, such as for generalized von Koch fractals in the lattice case. Applying Corollary~\ref{cor:distFormulaZero} and Theorem~\ref{thm:heatFormulaDist} (with $k=0$), we obtain the following distributional identities:
\begin{align}
    \label{eqn:volumeEx}
    V_{\partial\Omega,\Omega}(t) &= \sum_{\omega\in\Dd_\Phi(\HH_{\sigma_0})} \Res(\tubezeta_{\partial\Omega,\Omega}(s;\delta);\omega)\, t^{\dimension-\omega} + \Rr_V(t); \\
    \label{eqn:heatEx}
    \heatContent(t) &= \sum_{\omega\in\Dd_\Phi(\HH_{\sigma_0})} \Res(\heatzeta(s;\delta);\omega)\, t^{(\dimension-\omega)/2} + \Rr_E(t), 
\end{align}
where $\Rr_V(t)=O(t^{\dimension-\sigma_0-\e})$ and $\Rr_E(t^{(\dimension-\sigma_0)/2-\e})$ as $t\to0^+$ for sufficiently small $\e>0$. 

Observe that Equation~\ref{eqn:volumeEx} and Equation~\ref{eqn:heatEx} are strikingly similar. In both cases, there are constants $a_\omega$ or $b_\omega$, given by the residues of the respective zeta function at $\omega$, and monomials of the form $t^{\dimension-\omega}$ or $t^{(\dimension-\omega)/2}$, with the complex numbers $\omega$ indexed by the set $\Dd_\Phi(\HH_{\sigma_0})$. By Corollary~\ref{cor:possibleCdims}, this set is exactly the possible complex dimensions of the RFD $(\partial\Omega,\Omega)$ in $\HH_{\sigma_0}$. So, the complex dimensions of the fractal directly determine the form of the expansion. 

Further, we have that the possible complex dimensions are computable directly from knowledge of the self-similar system $\Phi$ alone. Namely, by Corollary~\ref{cor:possibleCdims}, we have that 
\[ \Dd_\Phi(\HH_{\sigma_0}) = \set{\omega\in\HH_{\sigma_0}\suchthat 1 = \sum_{\ph\in\Phi} \lambda_\ph^\omega }. \]
It is through this knowledge of the from of $\zeta_\Phi$ that we can plot possible complex dimensions such as in Figure~\ref{fig:cDimPlots2D} for generalized von Koch fractals. For self-similar fractals in any dimension, provided the osculant separation conditions, we can thus see that the possible complex dimensions are exactly the expected complex dimensions prescribed by Moran's equation. 

\section*{Results in Context}

This computation of the possible complex dimensions is a partial resolution of Conjecture~6.2.36 of \cite{LRZ17_FZF}, as we have shown that self-similar fractals with appropriate separation conditions have complex dimensions controlled by the complexified Moran equation. It remains to establish exactly when the possible complex dimensions are the actual complex dimensions, which hinges on computation of the residues of the tube zeta function of $(\partial\Omega,\Omega)$. The computation of these residues, as well as the residues of the heat zeta function, is perhaps the most important follow-up question for future consideration. Not only will it answer the question of when the possible is actual, but also it allows one to use the explicit formulae developed herein numerically rather than only qualitatively. 

This work also serves to generalize the known results regarding heat asymptotics of generalized von Koch fractals, from the results of Fleckinger, Levitin, and Vassiliev on the von Koch snowflake \cite{FLV95a} to the work of van den Berg on generalized von Koch fractals \cite{vdB00_generalGKF}. Theorem~\ref{thm:pointwiseHeatFormulaGKF}, Corollary~\ref{cor:pointwiseheatFormulaGKFLattice}, Theorem~\ref{thm:distHeatFormulaGKF}, and Corollary~\ref{cor:distHeatFormulaGKFLattice} all give different types of explicit formulae for the heat content of GKFs in the nonlattice case when $n\geq5$, extending the results of the lattice case found in these works. Furthermore, the results of Section~\ref{sec:heatContent} give a general method for computing explicit formulae directly from estimating decomposition remainders.

In regards to tube formulae, this work extends the results of Lapidus and Pearse regarding the von Koch snowflake \cite{LP06} to its generalized analogues with different parameter values of $n$ and $r$. Explicitly, Theorem~\ref{thm:cDimsOfGKFs}, which originally appeared in \cite{Hof25}, gives a characterization of the possible complex dimensions of GKFs. Theorem~\ref{thm:pointwiseTubeFormulaGKF}, Corollary~\ref{cor:pointwiseTubeFormulaGKFLattice}, Theorem~\ref{thm:distTubeFormulaGKF}, and Corollary~\ref{cor:distTubeFormulaGKFLattice} all give explicit formulae for the tube functions based on the knowledge of these possible complex dimensions. The results of Section~\ref{sec:fractalSFEs} generalize the work of \cite{DKOU15} to establish tube formulae from approximate functional equations which contain remainder terms. Many more examples of fractals may be readily approached by the results of Chapter~\ref{chap:geometry}, for example the study of the Menger sponges and its generalizations which is underway. 

In the general setting, the results of Chapter~\ref{chap:SFEs} establish means by which to solve a general class of scaling functional equations. Section~\ref{sec:generalSFEs}, notably, makes precise the assumptions on the remainder term which are required for this process and establishes several sufficient criteria for admissibility. This fills in the gaps for using this method, originally explored in \cite{Hof25} for application to generalized von Koch fractals, for the general setting. Notably, we have corrected an error in the statement of Theorem~6.3 of \cite{Hof25} and have introduced terminology which makes precise the missing assumption, namely the notion of joint languidity (Definition~\ref{def:jointlyLanguid}). (We note that Theorem~6.3 is still valid under the assumption that the remainder zeta function is strongly languid, and in fact the proof of its strong languidity, bypassing the need for these extra assumptions, will be the subject of a future work.) Further, we have proven sufficient conditions for this new assumption to hold in the general case (see Theorem~\ref{thm:lowerDimAdmissibility}, Corollary~\ref{cor:lowerDimAdmissibilityRescaled}, and Theorem~\ref{thm:latticeCaseAdmissibility}). 

\section*{Questions for Future Study}

A major through line of this work is the question of how geometry and spectra are related. This study is related to Conjecture~6.2.36 in \cite{LRZ17_FZF}, which is in a sense an entire program. In this work, we have focused on two types of quantities: volumes of tubular neighborhoods and the heat content of self-similar fractals. Already, these results suggest a strong connection between the geometric complex dimensions of fractals in higher dimensions and quantities of spectra. The connection between geometry and spectra for bounded open sets in $\RR$ has been established in \cite{LapvFr13_FGCD}, and this work suggests that similar will be true at least for self-similar fractals. 

However, it is likely to be true for many more, if not all, fractals in higher dimensions that quantities of geometry and spectra are strongly related to complex dimensions. The results herein suggest that for fractals with underlying dynamics (here, fractals which are attractors of self-similar systems), there is strong reason to believe that the properties of the dynamics will control both geometric and spectral properties of the set. Establishing functional equations, even perhaps infinite or nonlinear functional equations, may be the key to extending this connection from the self-similar case. 

Beyond heat and volume, it stands to reason that for self-similar fractals the spectrum of the Laplacian on the fractals should be computable. Indeed, we conjecture that the results of Lapidus and Kigami regarding the spectra of post critically finite (p.c.f) fractals \cite{KL93}, which uses the method of Neumann-bracketing, should extend to self-similar fractals which form osculant fractal drums. These p.c.f. fractals may only have finitely many points of overlap between different images, but osculant fractal drums can have infinitely many points of overlap (but with Lebesgue measure zero). 

Further, these types of fractals allows one to study regions with fractal boundary, such as von Koch snowflakes. If the method of Neumann bracketing can be adapted to produce scaling functional equations for the spectral counting function, then the results of Section~\ref{sec:SFEsolutions} would produce explicit formulae for the spectrum and its counting function. Such a method could be an approach to extending the results of Lapidus regarding the modified Weyl-Berry conjecture \cite{Lap91} to obtain asymptotic expansions for regions with self-similar fractal boundary. 

Lastly, our work on the heat content in Chapter~\ref{chap:heat} suggests several approachable problems with generalize upon our results. The explicit formulae, combined with explicit computations or numerical estimates of the residues of heat zeta functions, can lead to finding explicit self-similar fractals which optimize heat flow into (or out of) a region with fractal boundary. 

Alternatively, the particular choice of boundary conditions considered may be generalized. The general PWB solution to the Dirichlet problem allows for a wide variety of functions as boundary conditions, including even a time-dependent boundary of the set itself. For example, were an actual snowflake melting, its boundary would shrink over time. So, the models of fractal snowflakes (such as GKFs) could be adapted to having time-dependent boundaries. However, the convergence of the solution to the boundary values, vis a vis classifying the regular and irregular boundary points and resolutive initial conditions, would need to be analyzed carefully. Also, in order to adapt the methods of this work to such situations one would need to develop a theory of time dependence of the self-similar system and perhaps restrict to boundaries which change in a uniform or self-similar way.

Beyond the heat equation, it is natural to wonder whether these same techniques apply to a general parabolic differential equation on self-similar fractals or regions with self-similar fractal boundary. Indeed, if one wishes to study a more complicated type of interaction such as convection in a rough material (viz. diffusion and advection in a region with fractal boundary) one would be led to such problems. Restricting to linear equations would be most compatible with the method of self-similar partitioning, however some form of equivariance may be appropriate for the behavior of solutions under similitude mappings.

\clearpage
\addcontentsline{toc}{chapter}{Bibliography}
\singlespacing
\setlength{\bibitemsep}{11pt}
\printbibliography

\doublespacing

\clearpage
\addcontentsline{toc}{chapter}{Index}
\printindex

\end{document}